\documentclass[a4paper,USenglish,thm-restate]{lipics-v2021}

\nolinenumbers
\pdfoutput=1

\title{A Modular Approach to Construct Signature-Free BRB Algorithms under a Message Adversary}

\titlerunning{Signature-Free BRB Algorithms under a Message Adversary}

\author{Timoth\'e Albouy}{IRISA, France}{timothe.albouy@irisa.fr}{https://orcid.org/0000-0001-9419-6646}{}

\author{Davide Frey}{Inria, France}{davide.frey@inria.fr}{https://orcid.org/0000-0002-6730-5744}{}

\author{Michel Raynal}{IRISA, France}{michel.raynal@irisa.fr}{https://orcid.org/0000-0002-3355-8719}{}

\author{Fran\c{c}ois Ta\"iani}{IRISA, France}{francois.taiani@irisa.fr}{https://orcid.org/0000-0002-9692-5678}{}

\authorrunning{T. Albouy, D. Frey, M. Raynal and F. Ta\"iani}

\Copyright{Timoth\'e Albouy, Davide Frey, Michel Raynal and Fran\c{c}ois Ta\"iani}



\acknowledgements{This work has been partially supported by the French ANR projects ByBloS (ANR-20-CE25-0002-01) and PriCLeSS (ANR-10-LABX-07-81) devoted to the design of modular distributed computing building blocks.}

\EventEditors{Eshcar Hillel, Roberto Palmieri, and Etienne Rivi\`{e}re}
\EventNoEds{3}
\EventLongTitle{26th International Conference on Principles of Distributed Systems (OPODIS 2022)}
\EventShortTitle{OPODIS 2022}
\EventAcronym{OPODIS}
\EventYear{2022}
\EventDate{December 13--15, 2022}
\EventLocation{Brussels, Belgium}
\EventLogo{}
\SeriesVolume{253}
\ArticleNo{1}

\renewcommand{\paragraph}[1]{%
  \medskip
  \noindent \textbf{#1} \hspace{.1em}
}

\newif\ifannotes

\usepackage[normalem]{ulem}
\usepackage{yfonts}
\usepackage{xspace}
\usepackage{newfloat}
\usepackage{thmtools}
\usepackage{hyperref}

\usepackage{wrapfig}
\usepackage{caption}
\usepackage{hypcap}
\usepackage{subcaption}

\usepackage{tikz}
\usetikzlibrary{arrows.meta}
\usetikzlibrary{patterns}
\usetikzlibrary{decorations.pathreplacing}

\newcommand{\ccomment}[1]{}
\newcommand{\annote}[3]{}

\ifannotes
\renewcommand{\annote}[3]{{\color{#3}%
	\colorbox{#3}{\bfseries\sffamily\tiny\textcolor{white}{#2}}
	\color{#3}
	$\blacktriangleright$\emph{#1}$\blacktriangleleft$}
}
\fi

\newcommand{\df}[1]{\annote{#1}{DF}{green}}
\newcommand{\ft}[1]{\annote{#1}{FT}{magenta}}
\newcommand{\ta}[1]{\annote{#1}{TA}{blue}}
\newcommand{\todo}[1]{\annote{#1}{Todo}{red}}

\DeclareFloatingEnvironment[fileext=frm,placement={!ht},name=Algorithm]{algorithm}
\newcounter{linecounter}
\renewcommand{\line}[1]{\refstepcounter{linecounter}\label{#1}(\arabic{linecounter})}
\newcommand{\resetline}[1]{\setcounter{linecounter}{0}#1}
\newcommand{\mylength}{0.5em}

\newcommand{\proofinappendix}{\em (Proof in Appendix~\ref{sec:sf-klcast-liveness}.)}

\newenvironment{proofsketch}{
  \renewcommand{\proofname}{Proof sketch}%
  \begin{proof}
}{%
  \end{proof}
}%

\newcommand{\urbroadcast}{\ensuremath{\mathsf{ur\_broadcast}}\xspace}

\newcommand{\receivedd}{\ensuremath{\mathsf{received}}\xspace}

\newcommand{\brbroadcast}{\ensuremath{\mathsf{brb\_broadcast}}\xspace}

\newcommand{\mbrbroadcast}{\ensuremath{\mathsf{mbrb\_broadcast}}\xspace}
\newcommand{\mbrdeliver}{\ensuremath{\mathsf{mbrb\_deliver}}\xspace}

\newcommand{\kl}{\ensuremath{k2\ell}}
\newcommand{\klcast}{\ensuremath{\kl\mathsf{\_cast}}\xspace}
\newcommand{\kldeliver}{\ensuremath{\kl\mathsf{\_deliver}}\xspace}
\newcommand{\kldelivered}{\ensuremath{\kl\mathsf{\_delivered}}\xspace}

\newcommand{\checkdelivery}{\ensuremath{\mathsf{check\_delivery}}\xspace}

\newcommand{\sfklcast}{\ensuremath{\mathsf{SigFreeK2LCast}}\xspace}
\newcommand{\sbklcast}{\ensuremath{\mathsf{SigBasedK2LCast}}\xspace}

\newcommand{\maxfn}{\ensuremath{\mathsf{max}}\xspace}

\newcommand{\tb}{\ensuremath{t}\xspace}
\newcommand{\tm}{\ensuremath{d}\xspace}

\newcommand{\qd}{\ensuremath{\mathit{q_d}}\xspace}
\newcommand{\qf}{\ensuremath{\mathit{q_f}}\xspace}
\newcommand{\single}{\ensuremath{\mathit{single}}\xspace}

\newcommand{\sn}{\ensuremath{\mathit{sn}}\xspace}
\newcommand{\id}{\ensuremath{\mathit{id}}\xspace}

\newcommand{\obje}{\ensuremath{\mathit{obj}_{\textsc{e}}}\xspace}
\newcommand{\objr}{\ensuremath{\mathit{obj}_{\textsc{r}}}\xspace}
\newcommand{\objw}{\ensuremath{\mathit{obj}_{\textsc{w}}}\xspace}

\newcommand{\kv}{\ensuremath{k'}\xspace}
\newcommand{\kld}{\ensuremath{k}\xspace}
\newcommand{\lgd}{\ensuremath{\ell}\xspace}
\newcommand{\nodpty}{\ensuremath{\delta}\xspace}

\newcommand{\lmbr}{\ensuremath{\ell_{\mathit{MBRB}}}\xspace}

\newcommand{\ki}{\ensuremath{k_I}\xspace}
\newcommand{\kf}{\ensuremath{k_F}\xspace}
\newcommand{\knf}{\ensuremath{k_{\mathit{NF}}}\xspace}
\newcommand{\knfp}{\ensuremath{k_{\mathit{NB}}}\xspace}
\newcommand{\lef}{\ensuremath{\ell_e}\xspace}
\newcommand{\NF}{\ensuremath{\mathit{NF}}\xspace}
\newcommand{\NB}{\ensuremath{\mathit{NB}}\xspace}
\newcommand{\kur}{\ensuremath{k_U}\xspace}

\newcommand{\wA}{\ensuremath{w_A}\xspace}

\newcommand{\wAc}{\ensuremath{w_A^c}\xspace}
\newcommand{\wBc}{\ensuremath{w_B^c}\xspace}
\newcommand{\wCc}{\ensuremath{w_C^c}\xspace}

\newcommand{\wAb}{\ensuremath{w_A^b}\xspace}

\newcommand{\sigiv}{\ensuremath{\mathit{sig}_i}\xspace}
\newcommand{\sigsiv}{\ensuremath{\mathit{sigs}_i}\xspace}
\newcommand{\sigsv}{\ensuremath{\mathit{sigs}}\xspace}

\newcommand{\msgm}{\textsc{msg}\xspace}
\newcommand{\initm}{\textsc{init}\xspace}
\newcommand{\echom}{\textsc{echo}\xspace}
\newcommand{\readym}{\textsc{ready}\xspace}
\newcommand{\witnessm}{\textsc{witness}\xspace}
\newcommand{\bundlem}{\textsc{bundle}\xspace}

\newcommand{\Endorsem}{\textsc{endorse}\xspace} 

\newcommand{\imp}{message\xspace}
\newcommand{\imps}{messages\xspace}
\newcommand{\EndorseMsg}{\Endorsem message\xspace}   
\newcommand{\EndorseMsgs}{\Endorsem messages\xspace} 
\newcommand{\app}{app-message\xspace}
\newcommand{\apps}{app-messages\xspace}

\newcommand{\ffalse}{\ensuremath{\mathtt{false}}\xspace}
\newcommand{\ttrue}{\ensuremath{\mathtt{true}}\xspace}

\newcommand{\klCtheo}{\textsc{\kl-Correctness}\xspace}
\newcommand{\klVprop}{\textsc{\kl-Validity}\xspace}
\newcommand{\klNDNprop}{\textsc{\kl-No-duplication}\xspace}
\newcommand{\klNDYprop}{\textsc{\kl-Conditional-no-duplicity}\xspace}
\newcommand{\klLDprop}{\textsc{\kl-Local-delivery}\xspace}

\newcommand{\klSGDprop}{\textsc{\kl-Strong-Global-delivery}\xspace}
\newcommand{\klWGDprop}{\textsc{\kl-Weak-Global-delivery}\xspace}

\newcommand{\mbrCtheo}{\textsc{MBRB-Correctness}\xspace}
\newcommand{\mbrVprop}{\textsc{MBRB-Validity}\xspace}
\newcommand{\mbrNDNprop}{\textsc{MBRB-No-duplication}\xspace}
\newcommand{\mbrNDYprop}{\textsc{MBRB-No-duplicity}\xspace}
\newcommand{\mbrLDprop}{\textsc{MBRB-Local-delivery}\xspace}
\newcommand{\mbrGDprop}{\textsc{MBRB-Global-delivery}\xspace}

\newcounter{sfklassumption}
\newcounter{sbklassumption}

\newcommand{\sfklassumref}[1]{\refstepcounter{sfklassumption}\label{#1}}
\newcommand{\sbklassumref}[1]{\refstepcounter{sbklassumption}\label{#1}}

\newcommand{\sfklassum}{sf-\kl-Assumption\xspace}
\newcommand{\sbklassum}{sb-\kl-Assumption\xspace}
\newcommand{\bassum}{\hyperref[assum:bassum]{B87-Assumption}\xspace}
\newcommand{\irassum}{\hyperref[assum:irassum]{IR16-Assumption}\xspace}

\newcommand{\sfklassums}{sf-\kl-Assumptions\xspace}

\newcommand{\irBound}{\ensuremath{5\tb+12\tm+\frac{2\tb\tm}{\tb+2\tm}}\xspace}
\newcommand{\irBoundN}{\ensuremath{5\tb+12\tm+\frac{2\tb\tm}{\tb+2\tm}}\xspace}
\newcommand{\brBound}{\ensuremath{2\tb+\tm + \sqrt{\tb^2 + 6\tb\tm + \tm^2}\xspace}}
\newcommand{\brBoundN}{\ensuremath{3\tb+2\tm + 2\sqrt{\tb\tm}\xspace}}

\newcommand{\irqdval}{\ensuremath{\left\lfloor \frac{ n+3\tb}{2}\right\rfloor +3\tm+ 1}\xspace}
\newcommand{\irqfval}{\ensuremath{\left\lfloor \frac{n+\tb}{2}\right\rfloor  + 1 }\xspace}
\newcommand{\kval}[1][]{\ensuremath{\left\lfloor \frac{c(#1\qf-1)}{c-\tm-#1\qd+#1\qf} \right\rfloor + 1}\xspace}
\newcommand{\kir}{\ensuremath{\left\lfloor{\frac{c \big\lfloor \frac{n+\tb}{2} \big\rfloor}{c -  \tb - 4\tm}}\right\rfloor + 1}\xspace}
\newcommand{\lir}{\ensuremath{ \left\lceil{c \left( 1 - \frac{\tm}{c - \left\lfloor \frac{ n+3\tb}{2}\right\rfloor -3\tm} \right)}\right\rceil  }\xspace}

\newcommand{\lval}[1][]{\ensuremath{\left\lceil c\left(1-\frac{\tm}{c-#1\qd+1}\right)\right\rceil}\xspace}

\begin{document}

\maketitle

\begin{abstract}
This paper explores how reliable broadcast can be implemented without signatures when facing a dual adversary that can both corrupt processes and remove messages.
More precisely, we consider an asynchronous $n$-process message-passing system in which up to \tb processes are Byzantine and where, at the network level, for each message broadcast by a correct process, an adversary can prevent up to \tm processes from receiving it (the integer \tm defines the power of the message adversary).
So, unlike previous works, this work considers that not only can computing entities be faulty (Byzantine processes), but, in addition, that the network can also lose messages.
To this end, the paper adopts a modular strategy and first introduces a new basic communication abstraction denoted \kl-cast, which simplifies quorum engineering, and studies its properties in this new adversarial context.
Then, the paper deconstructs existing signature-free Byzantine-tolerant asynchronous broadcast algorithms and, with the help of the \kl-cast communication abstraction, reconstructs versions of them that tolerate both Byzantine processes and message adversaries.
Interestingly, these reconstructed algorithms are also more efficient than the Byzantine-tolerant-only algorithms from which they originate.
\end{abstract}

\newpage

\section{Introduction} \label{sec-intro}

\paragraph{Context: reliable broadcast and message adversaries.}
Reliable broadcast (RB) is a fundamental abstraction in distributed computing that lies at the core of many higher-level constructions (including distributed memories, distributed agreement, and state machine replication).
Essentially, RB requires that non-faulty (i.e., correct) processes agree on the set of messages they deliver so that this set includes at least all the messages that correct processes have broadcast.


In a failure-free system, implementing reliable broadcast on top of an asynchronous network is relatively straightforward \cite{R13}.
If processes may fail, and in particular if failed processes may behave arbitrarily (a failure known as Byzantine~\cite{LSP82,PSL80}), implementing reliable broadcast becomes far from trivial as Byzantine processes may collude to fool correct processes~\cite{R18}.
An algorithm that solves reliable broadcast in the presence of Byzantine processes is known as implementing BRB (Byzantine reliable broadcast).


BRB in asynchronous networks (in which no bound is known over message delays) has been extensively studied over the last forty years~\cite{ANRX21,ARX21,AW04,B87,CGR11,GKKPST20,IR16,MDT15,MR98,NRSVX20,R18}.
Existing BRB algorithms typically assume they execute over a \textit{reliable} point-to-point network, i.e., a network in which sent messages are eventually received.
This is a reasonable assumption as most unreliable networks can be made reliable using re-transmissions and acknowledgments (e.g. a timeout-free version of the TCP protocol).

This work takes a drastic turn away from this usual assumption and explores how BRB might be provided when processes execute on an \textit{unreliable} network that might lose point-to-point messages.
Our motivation is threefold:
First, in volatile networks (e.g., mobile networks or networks under attack), processes might remain disconnected over long periods (e.g., weeks or months), leading in practice to considerable delays (a.k.a. tail latencies) when using re-transmissions. Because most asynchronous Byzantine-tolerant algorithms exploit intersecting quorums, these tail latencies can potentially limit the performance of BRB algorithms drastically, a well-known phenomenon in systems research~\cite{CSJRLWZCC21,DZ19,YPAKT22}.
Second, re-transmissions require that correct processes be eventually able to receive messages and cannot, therefore, model the permanent disconnection of correct processes.
Finally,  this question is interesting in its own right, as it opens up novel trade-offs between algorithm guarantees and network requirements, with potential application to the design of reactive distributed algorithms tolerant to both processes and network failures.

The impact of network faults on distributed algorithms has been studied in several works, in particular using the concept of message adversaries (MA). Message adversaries were initially introduced by N.~Santoro and P.~Widmayer in~\cite{SW89,SW07}\footnote{Where the terminology {\it communication failure model} and {\it ubiquitous faults} is used instead of MA.
While we consider only message losses, the work of Santoro and Widmayer also considers message additions and corruptions.}, and then used (sometimes implicitly) in many works (e.g.,~\cite{AFRT22,AG13,CS09,DPPU88,RS13,R15,SW07,TZKZ20}).
Initially proposed for synchronous networks, an MA may suppress point-to-point network messages according to rules that define its power.
For instance, a tree MA in a synchronous network might suppress any message except those transiting on an (unknown) spanning tree of the network, with this spanning tree possibly changing in each round.

The message losses that an MA causes differ fundamentally from Byzantine faults.
This is because an MA can affect the messages sent by any correct process, and can change the processes it targets during an execution, in contrast to Byzantine corruptions that are tied to a set of fixed processes (which is why MA faults are sometimes dubbed \emph{transient} or \emph{mobile}).
For instance, it may be tempting to think that Byzantine-fault-tolerant (BFT) algorithms inherently tolerate message losses from correct processes because they can only afford to wait for at most $n-t$ messages (where $n$ is the total number of processes, and $t$ the upper bound on Byzantine processes).
In an asynchronous network, a BFT algorithm could therefore miss up to $t$ messages from correct processes, if those are delayed by the scheduler.
This scenario only applies, however, in the particular circumstance where the $t$ faulty processes send messages that are received and accepted as valid by correct recipients.
This caveat is fundamental.
If the faulty processes remain silent or send contradicting messages (if they are Byzantine), then a BFT algorithm cannot afford to lose $t$ messages from correct processes.

\paragraph{Content of the paper.}
This paper combines a Message Adversary with Byzantine processes, and studies the signature-free implementation of Byzantine Reliable Broadcast (BRB) in an asynchronous, fully connected network subject to this MA and to at most \tb Byzantine faults.
The MA models lossy connections by preventing up to \tm point-to-point messages from reaching their recipient every time a correct process seeks to communicate with the rest of the network.\footnote{
A close but different notion was introduced by Dolev in~\cite{D82}, which considers static $\kappa$-connected networks.
If the adversary selects statically, for each correct sender, \tm correct processes that do not receive any of this sender's messages, the proposed model includes Dolev's model where $\kappa = n-\tm$.}

To limit as much as possible our working assumptions, we further assume that the computability power of the adversary is unbounded (except for the cryptography-based algorithm presented in Section~\ref{sec:sb-klcast}), which precludes the use of signatures.
(We do assume, however, that each point-to-point communication channel is authenticated.)\footnote{
Let us mention that the problem of designing an MA-tolerant BRB has been solved in~\cite{AFRT22} by leveraging digital signatures within a monolithic algorithm. Finding a signature-free counterpart remained, however, an open question, which we answer positively in this paper using a modular strategy.}

This represents a particularly challenging environment, as the MA may target different correct processes every time the network is used or focus indefinitely on the same (correct) victims.
Further, the Byzantine processes may collude with the MA for maximal impact. 

For clarity, in the remainder of the paper, we simply call \emph{\imps} the point-to-point network messages used internally by a BRB algorithm. (The MA may suppress these \imps.) We distinguish these \imps from the messages the BRB algorithm seeks to disseminate, which we call ``\textit{application messages}'' (\textit{\apps} for short).
%
%
%
%
In such a context, the paper presents the following results.
\begin{itemize}
    \item It first introduces a new modular abstraction, named \kl-cast, which appears to be a foundational building block to implement BRB abstractions (with or without the presence of an MA).
    This communication abstraction systematically dissociates the predicate used to forward (network) \imps from the predicate that triggers the delivery of an \app, and lies at the heart of the work presented in the paper.
    When proving the \kl-cast communication abstraction, the paper presents an in-depth analysis of the power of an adversary that controls at most \tb Byzantine processes and an MA of power \tm. 
    
    \item Then, the paper deconstructs two signature-free BRB algorithms (Bracha's~\cite{B87} and Imbs and Raynal's~\cite{IR16} algorithms) and reconstructs versions of them that tolerate {\it both} Byzantine processes and MA.
    Interestingly, when considering Byzantine failures only, these deconstructed versions use smaller quorum sizes and are, therefore, more efficient than their initial counterparts.
\end{itemize}

So, this paper is not only the first to present signature-free BRB algorithms in the context of asynchrony and MA but also the first to propose an intermediary communication abstraction that allows us to obtain efficient BRB algorithms.
For clarity, we give in Table~\ref{table:terminology} the list of acronyms and notations used in this paper.

\begin{table}[ht]
\begin{center}
\renewcommand{\baselinestretch}{1}
\small
\begin{tabular}{|c|c|}
\hline {\bf Acronyms} & {\bf Meaning}\\
\hline BRB & Byzantine-tolerant reliable broadcast \\
\hline MA & Message adversary \\
\hline MBRB & Message adversary- and Byzantine-tolerant reliable broadcast\\
\hline
\hline {\bf Notations} & {\bf Meaning}\\
\hline $n$ & number of processes in the network \\
\hline \tb & upper bound on the number of Byzantine processes \\
\hline \tm & power of the message adversary \\
\hline $c$ & effective number of correct processes in a run ($n-\tb \leq c \leq n$) \\
\hline \kld & minimal nb of correct processes that \kl-cast a message\\
\hline \lgd & minimal nb of correct processes that \kl-deliver a message\\
\hline \kv & minimal nb of correct \kl-casts if there is a correct \kl-delivery\\
\hline \nodpty & \ttrue iff no-duplicity is guaranteed, \ffalse otherwise\\
\hline \qd & size of the \kl-delivery quorum\\
\hline \qf & size of the forwarding quorum\\
\hline \single & \ttrue iff only a single message can be endorsed, \ffalse otherwise\\
\hline
\end{tabular}
\end{center}
\caption{Acronyms and notations}
\label{table:terminology}
\end{table}

\paragraph{Roadmap.}
The paper is composed of~\ref{sec-conclusion} sections and four appendices.
Section~\ref{sec-model} describes the underlying computing model.
Section~\ref{sec-klcast} presents the \kl-cast abstraction and its properties.
Section~\ref{sec-MBRB} defines the MA-tolerant BRB communication abstraction.
Section~\ref{sec-algos} shows that thanks to the \kl-cast abstraction, existing BRB algorithms can give rise to MA-tolerant BRB algorithms which, when $\tm = 0$, are more efficient than the BRB algorithms they originate from.
Section~\ref{sec:sb-klcast} presents a signature-based implementation of \kl-cast that possesses optimal guarantees.
Finally, Section~\ref{sec-conclusion} concludes the paper.
Due to page limitations, some contributions are presented in appendices: namely some proofs and a numerical evaluation of the \kl-cast abstraction.

\section{Computing Model} \label{sec-model}

\paragraph{Process model.}
The system is composed of $n$ asynchronous sequential processes denoted $p_1,...,p_n$.
Each process $p_i$ has a distinct identity, known to other processes.
For simplicity and without loss of generality, we assume that $i$ is the identity of $p_i$.

In terms of faults, up to $\tb \geq 0$ processes can be Byzantine, where a Byzantine process is a process whose behavior does not follow the code specified by its algorithm~\cite{LSP82,PSL80}.
Byzantine processes may collude to fool non-Byzantine processes (also called correct processes).
In this model, the premature stop (crash) of a process is a Byzantine failure.
In the following, given an execution, $c$ denotes the effective number of processes that behave correctly in that execution.
We always have $n-\tb \leq c \leq n$.
While this number remains unknown to correct processes, it is used to analyze and characterize (more precisely than using its worse value $n-\tb$) the guarantees provided by the proposed algorithms.

Finally, the processes have no access to random numbers, and their computability power is unbounded.
Hence, the algorithms presented in the paper are deterministic and signature-free (except the signature-based algorithm presented in Section~\ref{sec:sb-klcast}).

\paragraph{Communication model.}
Processes communicate by exchanging \imps through a fully connected asynchronous point-to-point network, assumed to be reliable in the sense it neither corrupts, duplicates, nor creates \imps.
As far as \imps losses are concerned, the network is under the control of an adversary (see below) that can suppress \imps.

Let \msgm be a \imp type and $v$ the associated value.
A process can invoke the best-effort broadcast macro-operation denoted $\urbroadcast(\msgm(v))$, which is a shorthand for ``{\bf for all} $j \in \{1,\cdots, n\}$ {\bf do} ${\sf send}$ $\msgm(v)$ {\bf to} $p_j$ {\bf end for}''.
Correct processes are assumed to invoke \urbroadcast to send \imps. When they do, we say that the \imps are {\it ur-broadcast} and {\it received}.
The operation $\urbroadcast(\msgm(v))$ is not reliable.
For example, if the invoking process crashes during its invocation, an arbitrary subset of processes receive the \imp $\msgm(v)$.
Moreover, due to its very nature, a Byzantine process can send fake \imps without using the macro-operation \urbroadcast.


\paragraph{Message adversary.}
Let \tm be an integer constant such that $0 \leq \tm < n-t$.
The communication network is controlled by an MA (as defined in Section~\ref{sec-intro}), which eliminates \imps ur-broadcast by correct processes, so these \imps are lost.
More precisely, when a correct process invokes $\urbroadcast(\msgm(v))$, the MA is allowed to arbitrarily suppress up to \tm copies of the \imp $\msgm(v)$ intended to correct processes\footnote{%
Note that this message adversary is not limited to algorithms that use the \urbroadcast macro-operation. The same adversary can be equivalently defined for an operation $\mathsf{ur\_multicast}$ that sends a \imp to a dynamically defined 
subset of processes (be it multiple recipients or only one in the case of unicast), by stipulating that the MA can still suppress up to \tm copies of this \imp.
In this case, the most robust way for correct processes to disseminate a \imp is to send it to all processes, i.e. to fall back on a \urbroadcast operation.}.
This means that, although the sender is correct, up to \tm correct processes may miss the \imp $\msgm(v)$.
The extreme case $\tm=0$ corresponds to the case where no \imp is lost.

As an example, let us consider a set $D$ of correct processes, where $1 \leq |D| \leq \tm$, such that during some period of time, the MA suppresses all the \imps sent to them.
It follows that, during this period of time, this set of processes appears as being input-disconnected from the other correct processes.
Note that the size and the content of $D$ can vary with time and are never known by the correct processes.


\section{\texorpdfstring{\kl}{k2l}-Cast Abstraction}
\label{sec-klcast}

Signature-free BRB algorithms~\cite{ART19,B87,IR16} often rely on successive waves of internal messages (e.g. the \echom or \readym messages of Bracha's algorithm~\cite{B87}) to provide safety and liveness.
Each wave is characterized by a threshold-based predicate that triggers the algorithm's next phase when fulfilled (e.g. enough \echom messages for the same \app $m$).

In this section, we introduce, implement, and prove a new modular abstraction, called \kl-cast, that encapsulates a wave/thresholding mechanism that is both Byzantine- and MA-tolerant.
As previously announced, we then use this abstraction to reconstruct MA-tolerant BRB algorithms in Section~\ref{sec-algos} from two existing BRB algorithms~\cite{B87,IR16}.

\subsection{Definition}
\kl-cast (for $k$-to-$\ell$-cast) is a many-to-many communication abstraction\footnote{
An example of this family is the binary reliable broadcast introduced in~\cite{MMR14}, which is defined by specific delivery properties---not including MA-tolerance---allowing binary consensus to be solved efficiently with the help of a common coin.}.
Intuitively, it relates the number \kld of correct processes that send a message $m$ (we say that these processes \emph{\kl-cast} $m$) with the number \lgd of correct processes that deliver $m$ (we say that they \emph{\kl-deliver} $m$).
Both \kld and \lgd are subject to thresholding constraints: enough correct processes must \kl-cast a message for it to be \kl-delivered at least once; and as soon as one (correct) \kl-delivery occurs, some minimal number of correct processes are guaranteed to \kl-deliver as well.

More formally, \kl-cast is a multi-shot abstraction, i.e. each \app $m$ that is \kl-cast or \kl-delivered is associated with an identity \id.
(Typically, such an identity is a pair consisting of a process identity and a sequence number.)
It provides two operations, \klcast and \kldeliver, whose behavior is defined by the values of four parameters: three integers \kv, \kld, \lgd, and a Boolean \nodpty.
This behavior is captured by the following six properties:
\begin{itemize} \label{def-kl-parameters}
\item Safety:
\begin{itemize}
    \item \klVprop. 
    If a correct process $p_i$ \kl-delivers an \app $m$ with identity \id, 
    then at least \kv correct processes \kl-cast $m$ with
    identity \id.
    
    \item \klNDNprop.
    A correct process \kl-delivers at most one \app $m$ with identity~\id.
    
    \item \klNDYprop.
    If the Boolean \nodpty is \ttrue, then no two different correct processes \kl-deliver different \apps with the same identity~\id.
\end{itemize}
\item Liveness\footnote{
The liveness properties comprise a \emph{local} delivery property that provides a necessary condition for the \kl-delivery of an \app by at least \emph{one} correct process, and two \emph{global} delivery properties that consider the collective behavior of correct processes.
}:
\begin{itemize}
    \item \klLDprop.
    If at least \kld correct processes
    \kl-cast an \app $m$ with identity \id and no correct
    process \kl-casts an \app $m' \neq m$ with identity~\id, then
    at least one correct process \kl-delivers the \app $m$ with
    identity~\id.

  \item \klWGDprop. 
    If a correct process \kl-delivers an \app $m$ with identity \id, then at least \lgd correct processes \kl-deliver an \app $m'$ with identity \id (each of them possibly different from $m$).
    
  \item \klSGDprop.
    If a correct process \kl-delivers an \app $m$ with identity \id, and no correct process \kl-casts an \app $m' \neq m$ with identity \id, then at least \lgd correct processes \kl-deliver the \app $m$ with identity \id.
\end{itemize} 
\end{itemize} 

This specification is \emph{parameterized} in the sense that each tuple $(\kv, \kld, \lgd, \nodpty)$ defines a specific communication abstraction with different guarantees.
This versatility explains why the \kl-cast abstraction can be used to produce highly compact reconstructions of existing BRB algorithms, rendering them MA-tolerant in the process (using four and three lines of pseudo-code respectively, see Section~\ref{sec-algos}).
Despite this versatility, however, we will see in Section~\ref{sec-klcast-sf} that \kl-cast can be implemented using a single (parameterized) algorithm, underscoring the fundamental commonalities of MA-tolerant BRB algorithms.

Intuitively, the parameters \kv, \kld, and \lgd hobble the disruption power of the Byzantine/MA adversary by setting limits on the number of correct processes that are either required or guaranteed to be involved in one communication ``wave'' (corresponding to one identity \id).
\kv sets the minimal number of correct processes that must \kl-cast for any \kl-delivery to occur: it thus limits the ability of the Byzantine/MA adversary to trigger spurious \kl-deliveries. The role of \kld is symmetrical.
It guarantees that some \kl-delivery will necessarily occur if \kld correct processes \kl-cast some message.
It thus prevents the adversary from silencing correct processes as soon as some critical mass of them participates.
Finally, \lgd captures a ``quite-a-few-or-nothing'' guarantee that mirrors the traditional ``all-or-nothing'' delivery  guarantee of traditional BRB.
As soon as one correct \kl-delivery occurs (for some identity \id), then \lgd correct processes must also \kl-deliver (with the same identity).

The fourth parameter, \nodpty, is a flag that when \ttrue enforces agreement between \kl-deliveries. When $\nodpty=\ttrue$, the \klNDYprop property implies that all the \apps $m'$ involved in the \klWGDprop property are equal to $m$.


\ta{In DMT, we need an additional property: guarantee that a correct process delivers if it is not a victim of the MA for sufficiently long (reuse formulation of SSS)}


\subsection{A Signature-Free Implementation of \texorpdfstring{\kl}{k2l}-Cast} \label{sec-klcast-sf}

Among the many possible ways of implementing  \kl-cast, this section 
presents a quorum-based\footnote{
In this paper, a quorum is a set of processes that (at the implementation level) ur-broadcast the same \imp.
This definition takes quorums in their ordinary sense. In a deliberative assembly, a quorum is the minimum number of members that must vote the same way for an irrevocable decision to be taken.
Let us notice that this definition does not require quorum intersection.
However, if quorums have a size greater than $\frac{n+\tb}{2}$, the intersection of any two quorums contains, despite Byzantine processes, at least one correct process~\cite{B87,R18}.}
signature-free implementation\footnote{Another \kl-cast implementation, which 
uses digital signatures and allows to reach optimal values for \kld and \lgd,
is presented in Section~\ref{sec:sb-klcast}.}
of the abstraction.
To overcome the disruption caused by Byzantine processes and message losses from the MA, our algorithm uses the ur-broadcast primitive (cf. our communication model in Sec.~\ref{sec-model}) to accumulate and forward \EndorseMsgs before deciding whether to deliver.
Forwarding and delivery are triggered by \emph{two thresholds} (a pattern also found, for instance, in Bracha's BRB algorithm~\cite{B87}): 
\begin{itemize}
    \item A first threshold, \qd, triggers the delivery of an \app $m$ when enough \EndorseMsgs supporting $m$ have been received.
    \item A second threshold, \qf, which is lower than \qd, controls how \EndorseMsgs are forwarded during the algorithm's execution.
\end{itemize}

\sloppy
Forwarding, which is controlled by \qf, amplifies how correct processes react to \EndorseMsgs, and is instrumental to ensure the algorithm's liveness. As soon as some critical ``mass'' of agreeing \EndorseMsgs accumulates within the system, forwarding triggers a chain reaction which guarantees that a minimum number of correct processes eventually \kl-deliver the corresponding \app.

\begin{figure}[tb]
\begin{minipage}[b]{.68\textwidth}
    \input{algorithms/sf-klcast}
    \renewcommand{\figurename}{Algorithm}
    \captionof{algorithm}{Signature-free \kl-cast (code~for~$p_i$)}
    \label{alg:sf-klcast}
\end{minipage}%
\hfil
\begin{minipage}[b]{.27\textwidth}
    \centering
    \resizebox{1\textwidth}{!}{
      \centering
\begin{tikzpicture}
    \def\totalHeight{5.5}

    \draw[-{>[scale=2.5,length=2,width=3]}] (0,\totalHeight) -- (0,\totalHeight/2+.6);
    \draw[-{>[scale=2.5,length=2,width=3]}] (0,\totalHeight/2) -- (0,.6);

    \node[draw,align=center,fill=white,text width=8.5em] at (0,\totalHeight) {
    Underlying system\\
    \smallskip
    $\langle n,\tb,\tm,c \rangle$};
    
    \node[draw,fill=white,rounded corners=.5em] at (0,3*\totalHeight/4) {\sfklassums~\ref{assum:base}-\ref{assum:r0}};
    
    \node[draw,align=center,fill=white,text width=8.5em] at (0,\totalHeight/2) {
    Implementation\\
    \smallskip
    $\langle \qd,\qf,\single \rangle$};
    
    \node[draw,fill=white,rounded corners=.5em] at (0,\totalHeight/4) {Theorem~\ref{theo:sf-kl-correctness}};
    
    \node[draw,align=center,fill=white,text width=8.5em] at (0,0) {
    \kl-cast object\\
    \smallskip
    $\langle \kv,\kld,\lgd,\nodpty \rangle$};
\end{tikzpicture}
    }
    \vspace{-.5em}
    \captionof{figure}{From the system parameters to a \kl-cast implementation}
    \label{fig-global-view} 
\end{minipage}
\end{figure}

More concretely, our algorithm provides an object (\sfklcast, Alg.~\ref{alg:sf-klcast}), instantiated using the function 
$\sfklcast(\qd,\qf,\single)$, using three input parameters:
\begin{itemize}
    \item \qd: the number of matching \EndorseMsgs that must be received from distinct processes in order to \kl-deliver an \app.
    
    \item \qf: the number of matching \EndorseMsgs that must be received from distinct processes for the local $p_i$ to endorse the corresponding \app (if it has not yet).
    
    \item \single: 
      a Boolean that controls whether a given correct process can endorse different \apps for the same identity \id ($\single=\ffalse$), or not ($\single=\ttrue$).
\end{itemize}

The algorithm provides the operations \klcast and \kldeliver.
Given an \app $m$ with identity \id, the operation $\klcast(m,\id)$ ur-broadcasts $\Endorsem{}(m,\id)$ provided $p_i$ has not yet endorsed any different \app for the same identity \id (lines~\ref{SFKL-cond-bcast}-\ref{SFKL-end-cond-bcast}). 
When $p_i$ receives a message $\Endorsem{}(m,\id)$, its executes two steps.
If the forwarding quorum \qf has been reached, $p_i$ first retransmits $\Endorsem{}(m,\id)$ (Forwarding step, lines~\ref{SFKL-cond-fwd}-\ref{SFKL-end-cond-fwd}). Then, if the \kl-delivery quorum \qd is attained, $p_i$ \kl-delivers $m$ (Delivery step, lines~\ref{SFKL-cond-dlv}-\ref{SFKL-end-cond-dlv}).


For brevity, we define $\alpha = n + \qf -\tb -\tm - 1$.
Given an execution defined by the system parameters $n$, \tb, \tm, and $c$, Alg.~\ref{alg:sf-klcast} requires the following assumptions to hold for the input parameters \qf and \qd of a \kl-cast instance (a global picture linking all parameters is presented in Fig.~\ref{fig-global-view}).
The prefix ``sf'' stands for signature-free.

\begin{itemize}
    \item \sfklassum~\ref{assum:base}: 
    $c-\tm \geq \qd \geq \qf+\tb \geq 2\tb+1$, \sfklassumref{assum:base}
    
    \item \sfklassum~\ref{assum:disc}: $\alpha^{2} - 4 (\qf - 1)(n - \tb) \geq 0$, \sfklassumref{assum:disc}

    \item \sfklassum~\ref{assum:r1}: $\alpha (\qd - 1) - (\qf-1) (n - \tb) - (\qd - 1)^{2} > 0$, \sfklassumref{assum:r1}
    
    \item \sfklassum~\ref{assum:r0}: $\alpha(\qd -1 -\tb) - (\qf -1)(n - \tb)  - (\qd -1 - \tb)^2 \geq 0$. \sfklassumref{assum:r0}
\end{itemize}

In particular, the safety of Alg.~\ref{alg:sf-klcast} algorithm relies solely on~\sfklassum~\ref{assum:base}, while its liveness relies on all four of them.
\sfklassum~\ref{assum:disc} through~\ref{assum:r0} constrain the solutions of a second-degree inequality resulting from the combined action of the MA, the Byzantine processes, and the message-forwarding behavior of Alg.~\ref{alg:sf-klcast}.
We show in Appendix~\ref{sec:proofsMBRB} that, in practical cases, these assumptions can be satisfied by a bound of the form $n > \lambda \tb+ \xi \tm + f(\tb,\tm)$, where $\lambda, \xi \in \mathbb{N}$ and $f(\tb,0)=f(0,\tm)=0$.
Together, the assumptions allow Alg.~\ref{alg:sf-klcast} to provide a \kl-cast abstraction (with values of the parameters \kv, \kld, \lgd, and \nodpty defining a specific \kl-cast instance) as stated by the following theorem.

\begin{theorem}[\klCtheo] \label{theo:sf-kl-correctness}
If \sfklassums~{\em\ref{assum:base}}--{\em \ref{assum:r0}} are verified, Alg.~{\em \ref{alg:sf-klcast}} implements \kl-cast with the following guarantees:
\begin{itemize}
    \item {\em \klVprop} with $\kv = \qf-n+c$,
    \item {\em \klNDNprop},
    \item {\em \klNDYprop} with $\begin{aligned}
    \nodpty = \left(\qf > \frac{n+\tb}{2}\right) \lor \left(\single \land \qd > \frac{n+\tb}{2}\right)
    \end{aligned}$,
    \item {\em \klLDprop} with $\begin{aligned}\textstyle
    \kld = \left\lfloor \frac{c(\qf-1)}{c-\tm-\qd+\qf} \right\rfloor + 1,
    \end{aligned}$
    
    \item 
    $\left.\begin{aligned}\begin{cases}
    \text{if } \single = \ffalse, & {\text{\em \klWGDprop}}\\
    \text{if } \single = \ttrue,& {\text{\em \klSGDprop}}
    \end{cases}\end{aligned}\right\}$
    with 
    $\begin{aligned}
      \lgd = 
        {\textstyle \left\lceil c\left(1-\frac{\tm}{c-\qd+1}\right)\right\rceil} 
    \end{aligned}.$

\end{itemize}
\end{theorem}

\subsection{Proof of Algorithm~\ref{alg:sf-klcast}}

The proofs of the \kl-cast safety properties stated in Theorem~\ref{theo:sf-kl-correctness} (\klVprop, \klNDNprop, and \klNDYprop) are fairly straightforward.
To save space, these proofs (Lemmas~\ref{lemma:n-bcast-if-kldv}-\ref{lemma:sf-kl-conditional-no-duplicity}) are provided in Appendix~\ref{sec:sf-klcast-safety}.

The proofs of the \kl-cast liveness properties (\klLDprop, \klWGDprop, \klSGDprop) are sketched informally below (Lemmas~\ref{lem:l-fact-min}-\ref{lemma:sf-strong-kl-global-delivery}). Their full development can be found in Appendix~\ref{sec:sf-klcast-liveness}.

When seeking to violate the liveness properties of \kl-cast, the attacker can use the MA to control in part how many \EndorseMsgs are received by each correct process, thus interfering with the quorum mechanisms defined by \qd and \qf.
To analyze the joint effect of this interference with Byzantine faults, our proofs consider seven well-chosen subsets of correct processes ($A$, $B$, $C$, $U$, $F$, \NF, and \NB, depicted in Fig.~\ref{fig:taxonomy}).

These subsets are defined for an execution of Alg.~\ref{alg:sf-klcast} in which \ki correct processes \kl-cast $(m,\id)$ (the $I$ in \ki is for ``Initial''), and \lef correct processes receive at least \qd message $\Endorsem{}(m,\id)$.
The first three subsets, $A$, $B$, and $C$, partition correct processes based on the number of $\Endorsem{}(m,\id)$ messages they receive. 
\begin{itemize}
    \item $A$ contains the \lef correct processes that receive at least \qd $\Endorsem{}(m,\id)$ messages (be it from correct or from Byzantine processes), and thus \kl-deliver some message.\footnote{
Because of the condition at line~\ref{SFKL-cond-dlv}, these processes do not necessarily \kl-deliver $(m,\id)$, but all do \kl-deliver an \app for identity \id.
}
    
    \item $B$ contains the correct processes that receive at least \qf but less than \qd $\Endorsem{}(m,\id)$ messages and thus do not \kl-deliver $(m,\id)$.
    
    \item $C$ contains the remaining correct processes that receive less than \qf $\Endorsem{}(m,\id)$ messages. They neither forward nor deliver any message for identity \id (since $\qf \leq \qd$).
\end{itemize}

In our proofs, we count how many messages $\Endorsem{}(m,\id)$ ur-broadcast by correct processes are received by the processes of $A$ (resp. $B$ and $C$). We note these quantities \wAc, \wBc, and \wCc, and use them to bootstrap our proofs using bounds on messages (see below).

The last four subsets intersect with $A$, $B$ and $C$, and distinguish correct processes based on the ur-broadcast operations they
perform. 
\begin{itemize}
    \item $U$ consists of the correct processes that ur-broadcast $\Endorsem{}(m,\id)$ at line~\ref{SFKL-bcast}.
    
    \item $F$ denotes the correct processes of $A \cup B$ that ur-broadcast $\Endorsem{}(m,\id)$ at line~\ref{SFKL-fwd} (i.e., they perform  forwarding).
    
    \item $\NF$ denotes the correct processes of $A \cup B$ that ur-broadcast $\Endorsem{}(m,\id)$ at line~\ref{SFKL-bcast}.
    
    \item \NB denotes the correct processes of $A \cup B$ that never ur-broadcast $\Endorsem{}(m,\id)$, be it at line~\ref{SFKL-bcast} or at line~\ref{SFKL-fwd}.
    These processes have received at least \qf messages $\Endorsem{}(m,\id)$, but do not forward $\Endorsem{}(m,\id)$, because they have already ur-broadcast $\Endorsem{}(m',\id)$ at line~\ref{SFKL-bcast} or at line~\ref{SFKL-fwd} for an \app $m' \neq m$. 
\end{itemize}

\paragraph{Proof strategy.}
We note $\kur=|U|$, $\kf=|F|$, $\knf=|\NF|$, $\knfp=|\NB|$.
Observe that $\kur \leq \ki$ and $\knf \leq \kur$, since all (correct) processes in $U$ and \NF invoke \klcast.
Also, $(\kur+\kf)$ represents the total number of correct processes that ur-broadcast a message $\Endorsem{}(m,\id)$.
Fig.~\ref{fig:msg-dist} illustrates how these quantities constrain the distribution of \EndorseMsgs across $A$, $B$ and $C$.
Our core proof strategy consists in bounding the areas shown in Fig.~\ref{fig:msg-dist}. (For instance, observe that $\wAc \leq |A|\times (\kur+\kf)$, since each of the \lef correct processes in $A$ can receive at most one \EndorseMsg from each of the $(\kur+\kf)$ correct processes that send them.) This reasoning on bounds yields a polynomial involving $\lef=|A|$, $\ki$, and $\kur$, whose roots can then be constrained to yield the liveness guarantees required by the \kl-cast specification.

\begin{figure}[t]
\vspace{-.5cm}
\begin{center}
\hspace{-1cm}
\begin{subfigure}[b]{0.40\textwidth}
  \centering
  \resizebox{1\textwidth}{!}{
    \centering
\begin{tikzpicture}
    \def\widthS{2}
    \def\heightS{4}
    \def\interS{.15}
    \def\marginText{.4}
    
    \def\colorA{black}
    \def\colorB{black}
    \def\colorC{black}
    \def\colorU{blue}
    \def\colorNF{orange}
    \def\colorF{blue}
    \def\colorNB{red}
    \def\colorR{red}
    
    \def\thirdW{\widthS-2*\interS}
    \def\thirdH{(\heightS-4*\interS)/3}
    
    \draw[rounded corners,very thick,\colorA] (0,0) rectangle (\widthS,\heightS);
    \node[\colorA] at (\widthS/2,-\marginText) {$A$};
    
    \draw[rounded corners,very thick,\colorB] (\widthS+\interS,0) rectangle (2*\widthS+\interS,\heightS);
    \node[\colorB] at (3*\widthS/2+\interS,-\marginText) {$B$};
    
    \draw[rounded corners,very thick,\colorC] (2*\widthS+2*\interS,0) rectangle (3*\widthS+2*\interS,\heightS);
    \node[\colorC] at (5*\widthS/2+2*\interS,-\marginText) {$C$};
    
    \draw[rounded corners,very thick,\colorU] (\interS,\interS) rectangle (3*\widthS+\interS,{\thirdH+\interS});
    \node[\colorU] at (3*\widthS+\interS-\marginText,{\thirdH+\interS-\marginText}) {$U$};
    
    \def\heightNF{\thirdH-2*\interS}
    \draw[rounded corners,very thick,\colorNF] (2*\interS,2*\interS) rectangle (2*\widthS,{\heightNF+2*\interS});
    \node[\colorNF] at (2*\interS+\marginText,{\heightNF+2*\interS-\marginText}) {\NF};
    
    \draw[rounded corners,very thick,\colorF] (\interS,{\thirdH+2*\interS}) rectangle (2*\widthS,{2*\thirdH+2*\interS});
    \node[\colorF] at (\interS+\marginText,{2*\thirdH+2*\interS-\marginText}) {$F$};
    
    \draw[rounded corners,very thick,\colorNB] (\interS,{2*\thirdH+3*\interS}) rectangle (2*\widthS,{3*\thirdH+3*\interS});
    \node[\colorNB] at (\interS+\marginText,{3*\thirdH+3*\interS-\marginText}) {\NB};
    
    \draw[rounded corners,very thick,\colorR] ({2*\thirdW+5*\interS},{\thirdH+2*\interS}) rectangle (3*\widthS+\interS,\heightS-\interS);
    \node[align=center,\colorR] at ({(5*\thirdW+6*\interS)/2},{(4*\thirdH+5*\interS)/2}) {remaining\\correct\\processes};
    
\end{tikzpicture}
    }
\caption{Subsets of correct processes based on the number of received \EndorseMsgs ($A$, $B$ and $C$) and based on their ur-broadcast actions ($U$, $F$, $\NF$, and $\NB$)}
\label{fig:taxonomy}
\end{subfigure}
\hspace{.2cm}
\begin{subfigure}[b]{0.50\textwidth}
  \centering
  \resizebox{1.07\textwidth}{!}{
    \centering
\begin{tikzpicture}
    \def\heightA{3}
    \def\heightB{2.25}
    \def\heightC{1.3}
    
    \def\widthA{1.5}
    \def\widthAB{3}
    \def\widthABC{4.5}
    
    \def\colorA{black}
    \def\colorB{black}
    \def\colorC{black}

    \draw[-{>[scale=2.5,length=2,width=3]}] (0,-.25) -- (0,\heightA+.5);
    \node[align=center] at (-.8,\heightA+.7) {\# received\\msgs.};
    \draw[-{>[scale=2.5,length=2,width=3]}] (-.25,0) -- (\widthABC+.5,0);
    \node[align=center] at (\widthABC+1,-.6) {\# correct\\processes};
    
    \node at (-.8,\heightA) {$\kur{+}\kf$};
    \draw (-.1,\heightA) -- (\widthA,\heightA) -- (\widthA,-.1);
    \node at (\widthA,-.35) {\lef};
    \draw[\colorA,decorate,decoration={brace}] (\widthA-.1,-.1) -- (.1,-.1);
    \node[\colorA] at (\widthA/2,-.4) {$A$};
    \node[\colorA] at (\widthA/2,\heightA/2) {\wAc};
    
    \node at (-.6,\heightB) {$\qd{-}1$};
    \draw (0,\heightB) -- (-.1,\heightB);
    \draw[dashed] (0,\heightB) -- (\widthA,\heightB);
    \draw (\widthA,\heightB) -- (\widthAB,\heightB) -- (\widthAB,-.1);
    \draw[dashed] (\widthAB, 0) -- (\widthAB, -0.7);
    \node at (\widthAB,-.8) {$\kf{+}\knf{+}\knfp$}; 
    \draw[\colorB,decorate,decoration={brace}] (\widthAB-.1,-.1) -- (\widthA+.1,-.1);
    \node[\colorB] at ({(\widthAB+\widthA)/2},-.4) {$B$};
    \node[\colorB] at ({(\widthAB+\widthA)/2},\heightB/2) {\wBc};
    
    \node at (-.6,\heightC) {$\qf{-}1$};
    \draw (0,\heightC) -- (-.1,\heightC);
    \draw[dashed] (0,\heightC) -- (\widthAB,\heightC);
    \draw (\widthAB,\heightC) -- (\widthABC,\heightC) -- (\widthABC,-.1);
    \node at (\widthABC,-.35) {$c$};
    \draw[\colorC,decorate,decoration={brace}] (\widthABC-.1,-.1) -- (\widthAB+.1,-.1);
    \node[\colorC] at ({(\widthABC+\widthAB)/2},-.4) {$C$};
    \node[\colorC] at ({(\widthABC+\widthAB)/2},\heightC/2) {\wCc};
\end{tikzpicture}
  }
\caption{Distribution of \EndorseMsgs among correct processes of $A$, $B$, and $C$, sorted by decreasing number of \EndorseMsgs received}
\label{fig:msg-dist}
\end{subfigure}
\end{center}
\caption{Subsets of correct processes and distribution of \EndorseMsgs among them}
\label{fig:msgs}
\vspace{-0.5cm}
\end{figure}

\paragraph{Observation.}
In the same way we have bounded \wAc, we can also bound \wBc by observing that there are $(\knf+\knfp+\kf-\lef)$ processes in $B$ and that each can receive at most $\qd-1$ \EndorseMsgs.
Similarly, we can bound \wCc by observing that the $(c-\knf-\knfp-\kf)$ processes of $C$ can receive at most $\qf-1$ \EndorseMsgs.
Thus:
\begin{align}
  \wAc &\leq (\kur+\kf)\lef, \label{eq:sup:on:wAc}\\
  \wBc &\leq (\qd-1)(\knf+\knfp+\kf-\lef), \label{eq:sup:on:wBc}\\
  \wCc &\leq (\qf-1)(c-\knf-\knfp-\kf). \label{eq:sup:on:wCc}
\end{align}

Moreover, the MA cannot suppress more than \tm copies of each individual \EndorseMsg ur-broadcast to the $c$ correct processes.
Thus, the total number of \EndorseMsgs received by correct processes $(\wAc+\wBc+\wCc)$ is such that: 
\begin{align}
  \wAc+\wBc+\wCc \geq (\kur+\kf)(c-\tm). \label{eq:sup:all:witness}
\end{align}

\begin{restatable}{rlemma}{lfactmin}
\label{lem:l-fact-min}
$\lef\times(\kur+\kf-\qd+1) \geq (\kur+\kf)(c-\tm-\qd+\qf) - c(\qf-1) - \knfp(\qd-\qf)$.
\end{restatable}

\begin{proofsketch}
We get this result by combining (\ref{eq:sup:on:wAc}), (\ref{eq:sup:on:wBc}), (\ref{eq:sup:on:wCc}) and (\ref{eq:sup:all:witness}), and using \sfklassum~\ref{assum:base} with the fact that $\knf \leq \kur$.
(Full derivations in Appendix~\ref{sec:sf-klcast-liveness}.)
\end{proofsketch}

\begin{restatable}{rlemma}{klcastiffwd}
\label{lem:klcast-if-fwd}
If no correct process \kl-casts $(m',\id)$ with $m' \neq m$, then no correct process forwards $\Endorsem{}(m',\id)$ at line~{\em{\ref{SFKL-fwd}}} (and then $\knfp = 0$).
\proofinappendix
\end{restatable}

\begin{restatable}[\klLDprop]{rlemma}{sfkllocaldelivery}
\label{lemma:sf-kl-local-delivery}
If at least $\kld=\Big\lfloor \frac{c(\qf-1)}{c-\tm-\qd+\qf} \Big\rfloor + 1$ correct processes \kl-cast an \app $m$ with identity \id and no correct process \kl-casts any \app $m'$ with identity \id such that $m \neq m'$, then at least one correct process $p_i$ \kl-delivers $m$ with identity \id.
\end{restatable}

\begin{proofsketch}
From the hypotheses, Lemma~\ref{lem:klcast-if-fwd} helps us determine that $\knfp=0$.
Then, the property is proved by contraposition, by assuming that no correct process \kl-delivers $(m,\id)$, which leads us to $\lef=0$.
Using prior information and \sfklassum~\ref{assum:base}, we can rewrite the inequality of Lemma~\ref{lem:l-fact-min} to get the threshold of \kl-casts above which there is at least one \kl-delivery. 
(Full derivations in Appendix~\ref{sec:sf-klcast-liveness}.)
\end{proofsketch}

\begin{restatable}{rlemma}{singleifknfp}
\label{lem:single-if-knfp}
$(\single = \ffalse) \implies (\knfp = 0)$.
\proofinappendix
\end{restatable}


\begin{restatable}{rlemma}{polynomifone}
\label{lem:polynom-if-one}
If at least one correct process \kl-delivers $(m,\id)$ and $x=\kur+\kf$ (the number of correct processes that ur-broadcast $\Endorsem{}(m,\id)$ at line~{\em \ref{SFKL-bcast}} or~{\em \ref{SFKL-fwd}}), then $x \geq \qd-\tb$ and $x^2 - x(c-\tm+\qf-1-\knfp) \geq -(c-\knfp)(\qf-1)$.
\end{restatable}

\begin{proofsketch}
We prove this lemma by counting the total number of messages (sent by Byzantine or correct processes) that are received by the processes of $A$, and by using (\ref{eq:sup:on:wAc}), (\ref{eq:sup:on:wCc}) (\ref{eq:sup:all:witness}), and \sfklassum~\ref{assum:base}.
(Full derivations in Appendix~\ref{sec:sf-klcast-liveness}.)
\end{proofsketch}

\begin{restatable}{rlemma}{enoughifone}
\label{lemma:ifone:then:enough:internal}
If $\knfp=0$, and at least one correct process \kl-delivers $(m,\id)$, then $\kur+\kf \geq \qd$.
\end{restatable}

\begin{proofsketch}
Given that $\knfp=0$, we can rewrite the inequality of Lemma~\ref{lem:polynom-if-one}, which gives us a second-degree polynomial (where $x=\kur+\kf$ is the unknown variable). We compute its roots and show that  the smaller one contradicts Lemma~\ref{lem:polynom-if-one}, and that the larger one is greater than or equal to \qd. The fact that $x$ must be greater than or equal to the larger root to satisfy Lemma~\ref{lem:polynom-if-one} proves the lemma.
(Full derivations in Appendix~\ref{sec:sf-klcast-liveness}.)
\end{proofsketch}

\begin{restatable}{rlemma}{ellifenough}
\label{lemma:if:enough:internal:then:enough:ell}
If $\knfp = 0$ and $\kur+\kf \geq \qd$, then at least \lval correct processes \kl-deliver some \app with identity \id (not necessarily $m$).
\end{restatable}

\begin{proofsketch}
From the hypotheses, we can rewrite the inequality of Lemma~\ref{lem:l-fact-min} to get a lower bound on \lef.
Using \sfklassum~\ref{assum:r1}, we can determine that this lower bound is decreasing with the number of ur-broadcasts by correct processes ($x=\kur+\kf$).
Hence, this lower bound is minimum when $x$ is maximum, that is, when $x=c$.
This gives us the minimum number of correct processes that \kl-deliver under the given hypotheses.
(Full derivations in Appendix~\ref{sec:sf-klcast-liveness}.)
\end{proofsketch}

\begin{restatable}[\klWGDprop]{rlemma}{sfklweakglobaldelivery}
\label{lemma:sf-weak-kl-global-delivery}
If $\single = \ffalse$, and a correct process \kl-delivers an \app $m$ with identity \id, then at least $\lgd = \lval$ correct processes \kl-deliver an \app $m'$ with identity \id (each possibly different from $m$).
\end{restatable}

\begin{proofsketch}
As $\single = \ffalse$ and one correct process \kl-delivers $(m,\id)$, Lemmas~\ref{lem:single-if-knfp} and~\ref{lemma:ifone:then:enough:internal} apply, and we have $\knfp = 0$ and $\kur+\kf \geq \qd$.
This provides the prerequisites for Lemma~\ref{lemma:if:enough:internal:then:enough:ell}, which concludes the proof.
(Full derivations in Appendix~\ref{sec:sf-klcast-liveness}.)
\end{proofsketch}

\begin{restatable}[\klSGDprop]{rlemma}{sfklstrongglobaldelivery}
\label{lemma:sf-strong-kl-global-delivery}
If $\single = \ttrue$, and a correct process \kl-delivers an \app $m$ with identity \id, and no correct process \kl-casts an \app $m' \neq m$  with identity \id, then at least $\lgd = \left\lceil c\left(1-\frac{\tm}{c-\qd+1}\right)\right\rceil$ correct processes \kl-deliver $m$ with identity \id.
\end{restatable}

\begin{proofsketch}
As $\single = \ttrue$, Lemma~\ref{lem:klcast-if-fwd} holds and implies that $\knfp = 0$. As  above, Lemma~\ref{lemma:ifone:then:enough:internal} and Lemma~\ref{lemma:if:enough:internal:then:enough:ell} hold, yielding the lemma.
(Full derivations in Appendix~\ref{sec:sf-klcast-liveness}.)
\end{proofsketch}


\section{BRB in the Presence of Message Adversary (MBRB): Definition}
\label{sec-MBRB}

Before using the \kl-cast abstraction to reconstruct MA-tolerant BRB algorithms, we first specify what a Byzantine- and MA-tolerant broadcast should precisely achieve.
We call such a broadcast an MBR-broadcast (for Message-adversarial Byzantine Reliable Broadcast), or MBRB for short.
The MBRB abstraction provides two matching operations, $\mbrbroadcast$ and $\mbrdeliver$.
It is a multishot abstraction, i.e, it associates an identity $\langle \sn, i \rangle$ (sequence number, sender identity) with each \app, and assumes that correct processes never reuse the same sequence number for different $\mbrbroadcast$ invocations.

When, at the application level, a process $p_i$ invokes $\mbrbroadcast(m,\sn)$, where $m$ is the \app, we say it ``mbrb-broadcasts $(m,\sn)$''.
Similarly, when the invocation of $\mbrdeliver$ by $p_i$ returns the tuple $(m,\sn,j)$ to the client application (where $p_j$ is the sender process), we say it ``mbrb-delivers $(m,\sn,j)$''.
So, the \app are {\it mbrb-broadcast} and {\it mbrb-delivered}.
Because of the MA, we cannot always guarantee that an \app mbrb-delivered by a correct process is eventually received by all correct processes.
Hence, in the MBR-broadcast specification, we introduce a variable \lmbr (reminiscent of the \lgd of \kl-cast) which indicates the strength of the global delivery guarantee of the primitive: if one correct process mbrb-delivers an \app, then \lmbr correct processes eventually mbrb-deliver this \app\footnote{
If there is no MA (i.e. $\tm = 0$), we should have $\lmbr=c \geq n-\tb$.}.
MBRB is defined by the following properties:
\begin{itemize}
\item Safety:
\begin{itemize}
\item \mbrVprop.
  If a correct process $p_i$
  mbrb-delivers an \app $m$ from a correct process $p_j$ with
  sequence number \sn, then $p_j$ mbrb-broadcast $m$ with sequence
  number~\sn.
\item \mbrNDNprop.
  A correct process $p_i$ mbrb-delivers at most one \app from a process $p_j$ with sequence number \sn.
\item \mbrNDYprop.
  No two distinct correct processes mbrb-deliver different \apps from a process $p_i$ with the same sequence number~\sn.
 \end{itemize}
\item Liveness:
\begin{itemize}

\item \mbrLDprop.  If a correct process $p_i$ mbrb-broadcasts
  an \app $m$ with sequence number~\sn, then at least one
  correct process $p_j$ eventually mbrb-delivers $m$ from $p_i$ with
  sequence number~\sn.

\item \mbrGDprop.  If a correct process $p_i$ mbrb-delivers an \app $m$ from a process $p_j$ with sequence number~\sn, then at least \lmbr correct processes mbrb-deliver $m$ from
  $p_j$ with sequence number~\sn.
\end{itemize}
\end{itemize}

It is implicitly assumed that a correct process does not use the same sequence number twice.
Let us observe that, as at the implementation level, the MA can always suppress all the \imps sent to a fixed set $D$ of \tm processes, these mbrb-delivery properties are the strongest that can be implemented.
More generally, the best-guaranteed value for \lmbr is $c-\tm$.
So, the  previous specification boils down to Bracha's specification~\cite{B87} for $\lmbr=c$.

\section{\texorpdfstring{\kl}{k2l}-Cast in Action: From Classical BRB to MA-Tolerant BRB (MBRB) Algorithms}
\label{sec-algos}
This section uses \kl-cast to reconstruct two signature-free BRB algorithms~\cite{B87,IR16} initially introduced in a pure Byzantine context (i.e., without any MA). This reconstruction produces Byzantine-MA-tolerant versions of the initial algorithms that implement the MBRB specification of Section~\ref{sec-MBRB}.
Moreover, when $\tm = 0$, our two reconstructed BRB algorithms are strictly more efficient than the original algorithms that gave rise to them (they terminate earlier).

More precisely, the original and reconstructed versions of Bracha's BRB are identical in terms of communication cost, time complexity, and \tb-resilience (when $\tm=0$).
The same comparison holds for the original and reconstructed versions of Imbs and Raynal's BRB.
However, both reconstructed BRB algorithms use smaller quorums than their original versions, and therefore require fewer messages to progress. 
In an actual network, this means a lower latency in practice, as practical networks typically exhibit a long tail distribution of latencies (a phenomenon well-studied by system and networking researchers~\cite{CSJRLWZCC21,DZ19,YPAKT22}).

%
To help readers familiar with the initial algorithms, we use the same \imp types (\initm, \echom, \readym, \witnessm) as in the original publications. 
It has been shown in~\cite{AFRT22} that the MBRB problem can be solved if and only if $n > 3\tb+2\tm$.

\subsection{Bracha's BRB algorithm reconstructed}
\paragraph{Reconstructed version.}
Bracha's BRB algorithm comprises three phases. 
When a process invokes $\brbroadcast(m,\sn)$, it disseminates the \app $m$ an \initm \imp (first phase).
The reception of this \imp by a correct process triggers its participation in a second phase implemented by the exchange of \imps tagged \echom.
Finally, when a process has received \echom \imps from ``enough'' processes, it enters the third phase, in which \readym messages are exchanged, at the end of which it brb-delivers the \app $m$.
Alg.~\ref{alg:b-mbrb} is a reconstructed version of the Bracha's BRB, which assumes $n > \brBoundN$.

\begin{algorithm}[tb]
\centering
\fbox{
  \begin{minipage}[t]{150mm}
\footnotesize
\resetline
\begin{tabbing}
{\bf init:} \=  $\obje \gets \sfklcast(\qd{=}\lfloor\frac{n+\tb}{2}\rfloor{+}1, \qf{=}\tb{+}1, \single{=}\ttrue)$;\\
\> $\objr \gets \sfklcast(\qd{=}2\tb{+}\tm{+}1, \qf{=}\tb{+}1, \single{=}\ttrue)$.\\[\mylength]

\line{BMBRB-mbrb}
{\bf operation} $\mbrbroadcast(m,\sn)$ {\bf is} $\urbroadcast(\initm(m,\sn))$.\\[\mylength]

\line{BMBRB-klc-echo}
{\bf when} $\initm(m,\sn)$ {\bf is} $\receivedd$ {\bf from} $p_j$ {\bf do} $\obje.\klcast(\echom(m),(\sn,j))$.\\[\mylength]

\line{BMBRB-klc-ready}
{\bf when} $(\echom(m),(\sn,j))$ {\bf is} $\obje.\kldelivered$ {\bf do} $\objr.\klcast(\readym(m),(\sn,i))$.\\[\mylength]

\line{BMBRB-mbrb-dlv}
{\bf when} $(\readym(m),(\sn,j))$ {\bf is} $\objr.\kldelivered$ {\bf do} $\mbrdeliver(m,\sn,j)$.

\end{tabbing}
\normalsize
\end{minipage}
}

\caption{\kl-cast-based reconstruction of Bracha's BRB algorithm (code of $p_i$)}
\label{alg:b-mbrb}
\end{algorithm}

The algorithm requires two instances of \kl-cast, denoted \obje and \objr, associated with the \echom \imps and the \readym \imps, respectively.
For both these objects, the Boolean \single is set to \ttrue.
For the quorums, we have the following:\\
\centerline{
$\bullet$ \obje:~~$\qf = \tb+1$ and $\qd= \lfloor \frac{n+\tb}{2} \rfloor + 1,$ \hspace{0.5cm} $\bullet$ 
\objr:~~$\qf= \tb+1$ and $\qd = 2\tb+\tm+1.$
}
The integer~\sn is the sequence number of the \app $m$ mbrb-broadcast by $p_i$.
The identity of $m$ is consequently the pair $\langle \sn,i \rangle$.

\noindent Alg.~\ref{alg:b-mbrb} provides $\lmbr = \left\lceil c\left(1-\frac{\tm}{c-2\tb-\tm}\right) \right\rceil$ under:
\label{page-algo-1}
\newcommand{\assumBrachaTshld}{\brBound} 
\begin{itemize}
    \item \bassum: $n > \brBoundN$; \label{assum:bassum}
\end{itemize}
its proof of correctness can be found in Appendix~\ref{sec-proof-b-mbrb} (B87 stands for Bracha 1987).

\paragraph{Comparison (Table~\ref{table:comparison-B87}).}
When $\tm=0$, both Bracha's algorithm and its reconstruction use the same quorum size for the \readym phase.
The quorums of the \echom phase are however different (Table~\ref{table:comparison-B87}).
As the algorithm requires $n>3\tb$, we define $\Delta=n-3\tb$ as the slack between the lower bound on $n$ and the actual value of $n$. 
When considering the forwarding threshold \qf, we have $\big\lfloor \frac{n+\tb}{2}\big\rfloor +1 =2\tb + \big\lfloor \frac{\Delta}{2} \big\rfloor +1>\tb+1$.
As a result, the reconstruction of Bracha's algorithm always uses a lower forwarding threshold for \echom messages than the original. It therefore forwards messages more rapidly and reaches the delivery quorum faster.


\begin{table}[ht]
\begin{center}
\small
\begin{tabular}{|c|c|c|}
\hline
\rule[-.75em]{0pt}{2em}
{\bf Threshold} & {\bf Original version} (\echom phase) & {\bf \kl-cast-based version} (\obje) \\
\hline
\rule[-1em]{0pt}{2.5em}
{\bf Forwarding  \qf} & $\displaystyle \Big\lfloor \frac{n+\tb}{2} \Big\rfloor+1$ & $\tb+1$ \\
\hline
\rule[-1em]{0pt}{2.5em}
{\bf Delivery \qd} & $\displaystyle \Big\lfloor \frac{n+\tb}{2} \Big\rfloor + 1$ & $\displaystyle \Big\lfloor \frac{n+\tb}{2} \Big\rfloor + 1$  \\
\hline
\end{tabular}
\end{center}
\caption{Bracha's original version vs. \kl-cast-based reconstruction  when $\tm = 0$}
\label{table:comparison-B87}
\end{table}

\subsection{Imbs and Raynal's BRB algorithm reconstructed}
\paragraph{Reconstructed version.}
Imbs and Raynal's BRB is another BRB implementation, which achieves an optimal good-case latency (only two communication steps) at the cost of a non-optimal \tb-resilience.
Its reconstructed version requires $n > \irBoundN$.

\begin{algorithm}[tb]
\centering
\fbox{
\begin{minipage}[t]{150mm}
\footnotesize
\resetline
\begin{tabbing}
{\bf init:} $\objw \gets \sfklcast(\qd{=}\irqdval, \qf{=}\irqfval, \single{=}\ffalse)$.\\[\mylength]

\line{IRMBRB-mbrb}
{\bf operation} $\mbrbroadcast(m,\sn)$ {\bf is} $\urbroadcast(\initm(m, \sn))$.\\[\mylength]

\line{IRMBRB-klcast-witness}
{\bf when} $\initm(m,\sn)$ {\bf is} $\receivedd$ {\bf from} $p_j$ {\bf do} $\objw.\klcast(\witnessm(m),(\sn,j))$.\\[\mylength]

\line{IRMBRB-mbrb-dlv}
{\bf when} $(\witnessm(m),(\sn,j))$ {\bf is} $\objw.\kldelivered$ {\bf do} $\mbrdeliver(m,\sn,j)$.

\end{tabbing}
\normalsize
\end{minipage}
}

\caption{\kl-cast-based reconstruction of Imbs and Raynal's BRB algorithm (code of $p_i$)}
\label{alg:ir-mbrb}
\end{algorithm}

The algorithm requires a single \kl-cast object, denoted \objw, associated with the \witnessm \imp, and which is instantiated with $\qf = \irqfval$ and $\qd = \irqdval$, and the Boolean $\single=\ffalse$.
Similarly to Bracha's reconstructed BRB, an identity of \app in this algorithm is a pair $\langle \sn,i \rangle$ containing a sequence number \sn and a process identity $i$.

\noindent Alg.~\ref{alg:ir-mbrb} provides $\lmbr = \lir$ under:
\begin{itemize}
    \item \irassum: $n > \irBoundN$; (where $\tb+\tm>0$) \label{assum:irassum}
\end{itemize}
its proof of correctness can be found in Appendix~\ref{sec-proof-ir-mbrb} (IR16 stands for Imbs-Raynal 2016).


\paragraph{Comparison (Table~\ref{table:comparison-IR}).}
Table~\ref{table:comparison-IR} compares Imbs and Raynal's original algorithm against its \kl-cast reconstruction for $\tm=0$.
Recall that this algorithm saves one communication step with respect to Bracha's at the cost of a weaker \tb-tolerance, i.e., it requires $n > 5\tb$.
As for Bracha, let us define the slack between $n$ and its minimum as $\Delta = n - 5\tb$, we have $\Delta \geq 1$.
\begin{itemize}
    \item Let us first consider the size of the forwarding quorum (first line of the table). 
    We have $n-2\tb = 3\tb+\Delta$ and $\lfloor \frac{n+\tb}{2} \rfloor+1 = 3\tb +\lfloor \frac{\Delta}{2} \rfloor+1$.
    When $\Delta>2$, we always have $\Delta > \lfloor \frac{\Delta}{2} \rfloor+1$, it follows that the forwarding predicate of the reconstructed version is equal or weaker than the one of the original version.
  
    \item The same occurs for the size of the delivery quorum (second line of the table).
    We have $n-\tb= 4\tb+\Delta$ and $\lfloor \frac{n+3\tb}{2} \rfloor+1 = 4\tb+\lfloor \frac{\Delta}{2} \rfloor+1$.
    So both reconstructed quorums are lower than those of the original version when $\Delta>2$, making the reconstructed algorithm quicker as soon as $n \geq 5\tb+3$.
    The two versions behave identically for $5\tb+3\geq n\geq 5\tb+2 \; (\Delta \in \{1,2\})$. 
\end{itemize}

\begin{table}[ht]
\begin{center}
\renewcommand{\baselinestretch}{1}
\small
\begin{tabular}{|c|c|c|}
\hline
\rule[-.75em]{0pt}{2em}
{\bf Threshold} & {\bf Original version} (\witnessm phase) & {\bf \kl-cast-based version} (\objw) \\
\hline
\rule[-1em]{0pt}{2.5em}
{\bf Forwarding \qf} & $\displaystyle  n-2\tb$ &
$\displaystyle \Big\lfloor \frac{n+\tb}{2} \Big\rfloor + 1$ \\
\hline
\rule[-1em]{0pt}{2.5em}
{\bf Delivery \qd} & $n-\tb$ & $\displaystyle \Big\lfloor \frac{n+3\tb}{2} \Big\rfloor + 1$  \\
\hline
\end{tabular}
\end{center}
\caption{Imbs and Raynal's original version vs. \kl-cast-based reconstruction when $\tm = 0$}
\label{table:comparison-IR}
\end{table}

\subsection{Numerical evaluation of the MBRB algorithms}
\label{sec:numer-eval-mbrb}

Fig.~\ref{fig:Bracha-IR-full-heatmap} provides a numerical evaluation of the delivery guarantees of both \kl-cast-based MBRB algorithms (Algs.~\ref{alg:b-mbrb} and~\ref{alg:ir-mbrb}) in the presence of Byzantine processes and an MA.
Results were obtained for $n=100$ and $c=n-t$, and show the values of \lmbr for different values of \tb and \tm.
For instance, Fig.~\ref{fig:Bracha-full-heatmap} shows that with $6$ Byzantine processes and an MA suppressing up to $9$ ur-broadcast \imps, Alg.~\ref{alg:b-mbrb} ensures the \mbrGDprop property with $\lmbr=83$.
The figures illustrate that the reconstructed Bracha algorithm performs in a broader range of parameter values, mirroring the bounds on $n$, \tb, and \tm captured by \bassum and \irassum.
Nonetheless, both algorithms exhibit values of \lmbr that can support real-world applications in the presence of an MA.

\begin{figure}[tb]
\centering{
\begin{subfigure}[t]{0.45\textwidth}\centering
\includegraphics[scale=0.35]{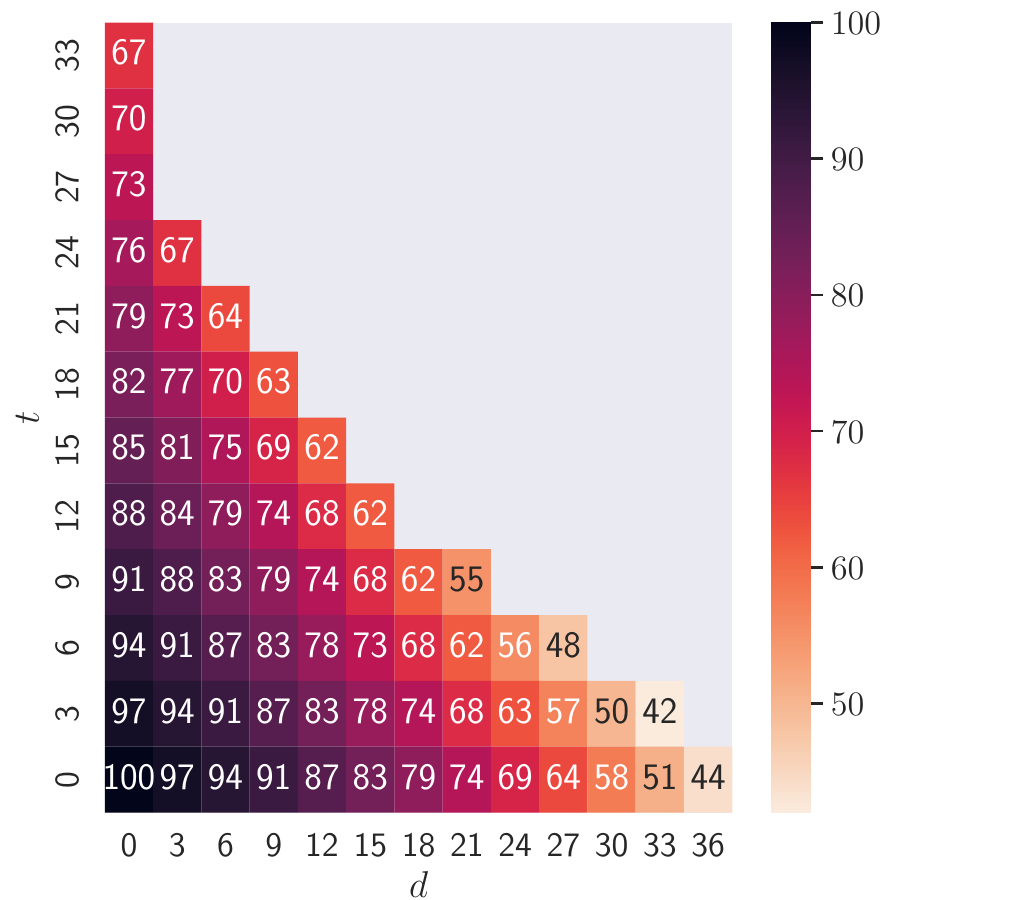}
\caption{Reconstructed Bracha MBRB (Alg.~\ref{alg:b-mbrb})}
\label{fig:Bracha-full-heatmap}
\end{subfigure}
\begin{subfigure}[t]{0.45\textwidth}\centering
\includegraphics[scale=0.35]{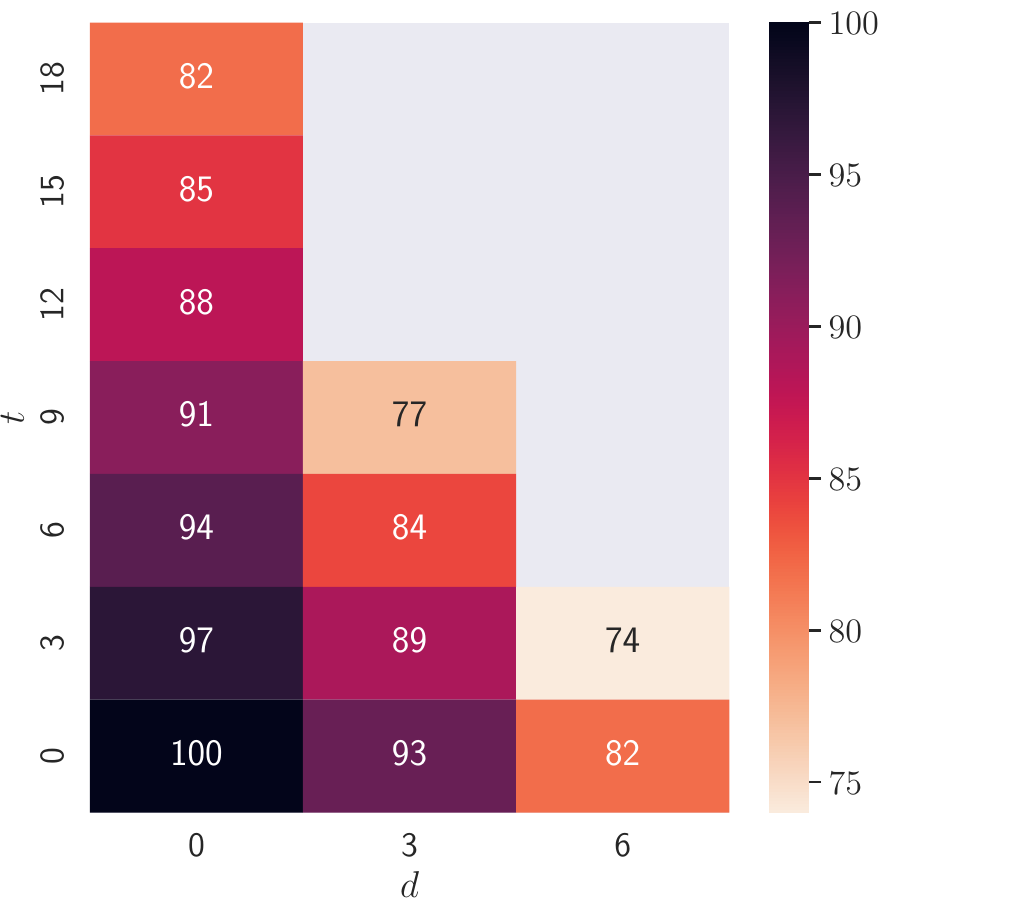}
\caption{Reconstructed Imbs-Raynal MBRB (Alg.~\ref{alg:ir-mbrb})}
\label{fig:IR-full-heatmap}
\end{subfigure}
}%
\caption{Values of \lmbr for the reconstructed BRB algorithms when varying \tb and \tm ($n=100$ and $c=n-t$) within the ranges that satisfy \bassum and \irassum} 
\label{fig:Bracha-IR-full-heatmap}
\end{figure}

\section{A Signature-Based Implementation of \texorpdfstring{\kl}{k2l}-Cast} \label{sec:sb-klcast}
This section presents an implementation of \kl-cast based on digital signatures.
The underlying model is the same as that of Section~\ref{sec-model} (page~\pageref{sec-model}), except that the computing power of the attacker is now bounded, which allows us to leverage asymmetric cryptography.

\subsection{Algorithm}
The signature-based algorithm is described in Alg.~\ref{alg:sb-klcast}.
It uses an asymmetric cryptosystem to sign \imps and verify their authenticity.
Every process has a public/private key pair.
Public keys are known to everyone, but private keys are only known to their owner. (Byzantine processes may exchange their private keys.)
Each process also knows the mapping between process indexes and associated public keys, and each process can produce a unique, valid signature for a given \imp, and check if a signature is valid.

\begin{algorithm}[tb]
\centering
\fbox{
\begin{minipage}[t]{150mm}
\footnotesize
\resetline
\begin{tabbing}
aaaA\=aA\=aA\=aA\=aA\=aaaaaaaaaaaaaaaaaaaaaaaaaaaaaaaaaaaaaaaaaaaaaaaaaaaaaA\=\kill

{\bf object} $\sbklcast(\qd)$ {\bf is}\\[\mylength]

\line{SBKL-klc}
\>{\bf operation} $\klcast(m,\id)$ {\bf is}\\

\line{SBKL-klc-cond}
\>\>{\bf if} \big($(-,\id)$ not already signed by $p_i$\big) {\bf then}\\

\line{SBKL-klc-sign}
\>\>\>$\sigiv \gets \text{signature of }(m,\id)\text{ by }p_i$;\\

\line{SBKL-klc-gather-sigs}
\>\>\>$\sigsiv \gets \{\text{all valid signatures for }(m,\id)\text{ ur-broadcast by }p_i\} \cup \{\sigiv\}$;\\

\line{SBKL-klc-bcast}
\>\>\>$\urbroadcast(\bundlem(m,\id,\sigsiv))$;\\

\line{SBKL-klc-chk-dlv}
\>\>\>$\checkdelivery()$\\

\line{SBKL-klc-end-cond}
\>\>{\bf end if}.\\[\mylength]

\line{SBKL-rcv}
\>{\bf when} $\bundlem(m,\id,\sigsv)$ {\bf is} $\receivedd$ {\bf do}\\

\line{SBKL-rcv-cond}
\>\>{\bf if} \big(\sigsv contains valid signatures for $(m,\id)$ not already ur-broadcast by $p_i$\big) {\bf then}\\

\line{SBKL-rcv-gather-sigs}
\>\>\>$\sigsiv \gets \{\text{all valid signatures for }(m,\id)\text{ ur-broadcast by }p_i\}$\\
\>\>\>\>$\cup\ \{\text{all valid signatures for }(m,\id)\text{ in }\sigsv\}$;\\

\line{SBKL-rcv-bcast}
\>\>\>$\urbroadcast(\bundlem(m,\id,\sigsiv))$;\\

\line{SBKL-rcv-chk-dlv}
\>\>\>$\checkdelivery()$\\

\line{SBKL-rcv-end-cond}
\>\>{\bf end if}.\\[\mylength]

\line{SBKL-chk-dlv}
\>{\bf internal operation} $\checkdelivery()$ {\bf is}\\

\line{SBKL-cond-dlv}
\>\>{\bf if} \big($p_i$ ur-broadcast at least \qd valid signatures for $(m,\id)$\\
\>\>\>$\land$ $(-,\id)$ not already \kl-delivered\big)\\

\line{SBKL-dlv}
\>\>\>{\bf then} $\kldeliver(m,\id)$\\

\line{SBKL-end-cond-dlv}
\>\>{\bf end if}.\\[\mylength]

{\bf end object}.
\end{tabbing}
\normalsize
\end{minipage}
}

\caption{\kl-cast implementation with signatures (code for $p_i$)}
\label{alg:sb-klcast}
\end{algorithm}

It is a simple algorithm that ensures that an \app must be \kl-cast by at least \kld correct processes to be \kl-delivered by at least \lgd correct processes.
For the sake of simplicity, we say that a correct process $p_i$ ``ur-broadcasts a set of signatures'' if it ur-broadcasts a $\bundlem(m,\id,\sigsiv)$ in which \sigsiv contains the signatures at hand.
A correct process $p_i$ ur-broadcasts an \app $m$ with identity \id
at line~\ref{SBKL-klc-bcast} or line~\ref{SBKL-rcv-bcast}.
\begin{itemize}
\item
  If this occurs at line~\ref{SBKL-klc-bcast}, $p_i$ includes in the \imp it ur-broadcasts all the signatures it has already received  for $(m,\id)$ plus its own signature.
\item
  If this occurs at line~\ref{SBKL-rcv-bcast}, $p_i$  has just received a \imp containing a set of signatures \sigsv for the pair $(m,\id)$.
  The process $p_i$ then aggregates in \sigsiv the valid signatures it just received with the ones it did know about beforehand (line~\ref{SBKL-rcv-gather-sigs}).
\end{itemize}
This algorithm simply assumes: (the prefix ``sb'' stands for signature-based)
\begin{itemize}
    \item \sbklassum~\ref{assum:no-partition}: $c > 2\tm$, \sbklassumref{assum:no-partition}
    
    \item \sbklassum~\ref{assum:dlv-tshld}: 
    $c-\tm \geq \qd \geq \tb+1$. \sbklassumref{assum:dlv-tshld}
\end{itemize}

Thanks to digital signatures, processes can relay the \imps of other processes in Alg.~\ref{alg:sb-klcast}. The algorithm, however, does not use forwarding in the same way Alg.~\ref{alg:sf-klcast} did: there is no equivalent of \qf here, that is, the only way to ``endorse'' an \app (which, in this case, is equivalent to signing this \app) is to invoke the \klcast operation.
Furthermore, only one \app can be endorsed by a correct process for a given identity (which is the equivalent of $\single = \ttrue$ in the signature-free version).

Although this implementation of \kl-cast provides better guarantees than Alg.~\ref{alg:sf-klcast}, using it to reconstruct signature-free BRB algorithms would be counter-productive.
This is because signatures allow for MA-tolerant BRB algorithms that are more efficient in terms of round and message complexity than those that can be constructed using \kl-cast~\cite{AFRT22}.

However, a signature-based \kl-cast does make sense in contexts in which many-to-many communication patterns are required~\cite{ART19}, and, we believe, opens the path to novel ways to handle local state resynchronization resilient to Byzantine failures and message adversaries.
For instance, we are using the following algorithm in our own work to design churn-tolerant money transfer systems tolerating Byzantine failures and temporary disconnections.

\subsection{Guarantees}
The proof of the following theorem can be found in Appendix~\ref{sec:sb-klcast-proofs}.
\begin{theorem}[\klCtheo] \label{theo:sb-kl-correctness}
If \sbklassum~{\em \ref{assum:no-partition}} and {\em \ref{assum:dlv-tshld}} are verified, Alg.~{\em \ref{alg:sb-klcast}} implements \kl-cast with the following guarantees:
(i) $\kv = \qd-n+c$, (ii) $\kld = \qd$, (iii) $\lgd = c-\tm$, and (iv) $\nodpty = \qd > \frac{n+\tb}{2}$.
\end{theorem}

\section{Conclusion} \label{sec-conclusion}
This paper discussed reliable broadcast in asynchronous systems where an adversary can control some Byzantine processes and can suppress \imps. 
Its starting point was the design of generic reliable broadcast abstractions suited to applications that do not require total order on the delivery of application messages (distributed money transfers are such applications~\cite{AFRT20,BDS20,GKMPS19}).
However, the ability to thwart an adversary controlling Byzantine processes and a message adversary is new.
This approach can be applied to the design of a wide range of quorum-based distributed algorithms other than reliable broadcast.
For instance, we conjecture that \kl-cast could benefit self-stabilizing and self-healing distributed systems~\cite{ADDP19}, where a critical mass of messages from other processes is needed in order to re-synchronize the local state of a given process.



\bibliography{main}


\appendix

\renewcommand{\proofinappendix}{}

\section{Proof of the Signature-Free \texorpdfstring{\kl}{k2l}-cast Implementation (Algorithm~\ref{alg:sf-klcast})}
\label{sec:sf-klcast-proofs}

\subsection{Safety Proof}
\label{sec:sf-klcast-safety}

\begin{lemma} \label{lemma:n-bcast-if-kldv}
If a correct process $p_i$ \kl-delivers $(m,\id)$, then at least $(\qf-n+c)$ correct processes have ur-broadcast {\em \Endorsem{}($m$,\id)} at line~{\em \ref{SFKL-bcast}}.
\end{lemma}

\begin{proof}
If $p_i$ \kl-delivers $(m,\id)$ at line~\ref{SFKL-dlv}, then it received \qd copies of $\Endorsem{}(m,\id)$ (because of the predicate at line~\ref{SFKL-cond-dlv}).
The effective number of Byzantine processes in the system is $n-c$, such that $0 \leq n-c \leq \tb$.
Therefore, $p_i$ must have received at least $\qd-n+c$ (which is strictly positive because $\qd \geq \qf > \tb \geq n-c$ by \sfklassum~\ref{assum:base}) messages $\Endorsem{}(m,\id)$ that correct processes ur-broadcast, either during a $\klcast(m,\id)$ invocation at line~\ref{SFKL-bcast}, or during a forwarding step at line~\ref{SFKL-fwd}.
There are two cases.
\begin{itemize}
    \item If no correct process has forwarded $\Endorsem{}(m,\id)$ at line~\ref{SFKL-fwd}, then at least $\qd-n+c \geq \qf-n+c$ (as $\qd \geq \qf$ by \sfklassum~\ref{assum:base}) correct processes have ur-broadcast $\Endorsem{}(m,\id)$ at line~\ref{SFKL-bcast}.
    
    \item If at least one correct process forwarded $\Endorsem{}(m,\id)$, then let us consider $p_j$, the first correct process that forwards $\Endorsem{}(m,\id)$.
    Because of the predicate at line~\ref{SFKL-cond-fwd}, $p_j$ must have received at least \qf distinct copies of the $\Endorsem{}(m,\id)$ message, out of which at most $n-c$ have been ur-broadcast by Byzantine processes, and at least $\qf-n+c$ (which is strictly positive because $\qf > \tb \geq n-c$ by \sfklassum~\ref{assum:base}) have been sent by correct processes.
    Moreover, as $p_j$ is the first correct process that forwards $\Endorsem{}(m,\id)$, all of the $\qf-n+c$ \EndorseMsgs it receives from correct processes must have been sent at line~\ref{SFKL-bcast}. \qedhere
\end{itemize}
\end{proof}

\begin{lemma}[\klVprop] \label{lemma:sf-kl-validity}
If a correct process $p_i$ \kl-delivers an \app $m$ with identity \id, then at least $\kv=\qf-n+c$ correct processes have \kl-cast $m$ with \id.
\end{lemma}

\begin{proof}
The condition at line~\ref{SFKL-cond-bcast} implies that the correct processes that ur-broadcast $\Endorsem{}(m,\id)$ at line~\ref{SFKL-bcast} constitute a subset of those that \kl-cast $(m,\id)$. Thus, by Lemma~\ref{lemma:n-bcast-if-kldv}, their number is at least $\kv=\qf-n+c$.
\end{proof}

\begin{lemma}[\klNDNprop] \label{lemma:sf-kl-no-duplication}
A correct process $p_i$ \kl-delivers an \app $m$ with identity \id at most once.
\end{lemma}

\begin{proof}
This property derives trivially from the predicate at line~\ref{SFKL-cond-dlv}.
\end{proof}

\begin{lemma}[\klNDYprop] \label{lemma:sf-kl-conditional-no-duplicity}
If the Boolean $\nodpty = \big((\qf > \frac{n+\tb}{2}) \lor (\single \land \qd > \frac{n+\tb}{2})\big)$ is \ttrue, then no two different correct processes \kl-deliver different \apps with the same identity \id.
\end{lemma}

\begin{proof}
Let $p_i$ and $p_j$ be two correct processes that respectively \kl-deliver $(m,\id)$ and $(m',\id)$. We want to prove that, if the predicate $\big((\qf > \frac{n+\tb}{2}) \lor (\single \land \qd > \frac{n+\tb}{2})\big)$ is satisfied, then $m=m'$. There are two cases. 

\begin{itemize}
    \item Case $\big( \qf > \frac{n+\tb}{2} \big)$.
    
    We denote by $A$ and $B$ the sets of correct processes that have respectively ur-broadcast $\Endorsem{}(m,\id)$ and $\Endorsem{}(m',\id)$ at line~\ref{SFKL-bcast}. By Lemma~\ref{lemma:n-bcast-if-kldv}, we know that $|A| \geq \qf-n+c > \frac{n+\tb}{2}-n+c$ and $|B| \geq \qf-n+c > \frac{n+\tb}{2}-n+c$. As $A$ and $B$ contain only correct processes, we have $|A \cap B| > 2(\frac{n+\tb}{2}-n+c)-c = \tb-n+c \geq \tb-\tb = 0$. Hence, at least one correct process $p_x$ has ur-broadcast both $\Endorsem{}(m,\id)$ and $\Endorsem{}(m',\id)$ at line~\ref{SFKL-bcast}. But because of the predicate at line~\ref{SFKL-cond-bcast}, $p_x$ ur-broadcasts at most one message $\Endorsem{}(-,\id)$ at line~\ref{SFKL-bcast}. We conclude that $m$ is necessarily equal to $m'$.

    \item Case $\big( \single \land \qd > \frac{n+\tb}{2} \big)$.
    
    Thanks to the predicate at line~\ref{SFKL-cond-dlv}, we can assert that $p_i$ and $p_j$ must have respectively received at least \qd distinct copies of $\Endorsem{}(m,\id)$ and $\Endorsem{}(m',\id)$, from two sets of processes, that we respectively denote $A$ and $B$, such that $|A| \geq \qd > \frac{n+\tb}{2}$ and $|B| \geq \qd > \frac{n+\tb}{2}$. We have $|A \cap B| > 2\frac{n+\tb}{2}-n = \tb$. Hence, at least one correct process $p_x$ has ur-broadcast both $\Endorsem{}(m,\id)$ and $\Endorsem{}(m',\id)$. But because of the predicates at lines~\ref{SFKL-cond-bcast} and \ref{SFKL-cond-fwd}, and as $\single = \ttrue$, $p_x$ ur-broadcasts at most one message $\Endorsem{}(-,\id)$, either during a $\klcast(m,\id)$ invocation at line~\ref{SFKL-bcast} or during a forwarding step at line~\ref{SFKL-fwd}. We conclude that $m$ is necessarily equal to $m'$. \qedhere
\end{itemize}
\end{proof}

\subsection{Liveness Proof}
\label{sec:sf-klcast-liveness}

\lfactmin*

\begin{proof}
Combining (\ref{eq:sup:on:wAc}), (\ref{eq:sup:on:wBc}), (\ref{eq:sup:on:wCc}) and (\ref{eq:sup:all:witness}) yields:
\begin{align*}
    (\kur+\kf)\lef + (\qd-1)&(\knf+\knfp+\kf-\lef) \,+\\
    &\hspace{5em} (\qf-1)(c-\knf-\knfp-\kf) \geq (\kur+\kf)(c-\tm),\\
    \lef\times(\kur+\kf-\qd+1) &\geq (\kur+\kf)(c-\tm) - (\qd-1)(\knf+\knfp+\kf) \,-\\
    &\hspace{14em} (\qf-1)(c-\knf-\knfp-\kf),\\
    &\geq (\kur+\kf)(c-\tm) - (\qd-\qf)(\knf+\knfp+\kf) - c(\qf-1).
\intertext{Using \sfklassum~\ref{assum:base}, we have $\qd-\qf \geq 0$.
By definition, we also have $\knf \leq \kur$, which yields:}
    \lef\times(\kur+\kf-\qd+1) &\geq (\kur+\kf)(c-\tm) - (\qd-\qf)(\kur+\kf+\knfp) - c(\qf-1),\\
    &\geq (\kur+\kf)(c-\tm-\qd+\qf) - c(\qf-1) - \knfp(\qd-\qf). \qedhere
\end{align*}
\end{proof}

\klcastiffwd*

\begin{proof}
Assume there is a correct process that ur-broadcasts $\Endorsem{}(m',\id)$ at line~\ref{SFKL-fwd} with $m' \neq m$.
Let us consider the first such process $p_i$.
To execute line~\ref{SFKL-fwd}, $p_i$ must first receive \qf messages $\Endorsem{}(m',\id)$ from distinct processes.
Since $\qf > \tb$ (\sfklassum~\ref{assum:base}), at least one of these processes, $p_j$, is correct.
Since $p_i$ is the first correct process to forward $\Endorsem{}(m',\id)$ at line~\ref{SFKL-fwd}, the $\Endorsem{}(m',\id)$ message of $p_j$ must come from line~\ref{SFKL-bcast}, and $p_j$ must have \kl-cast $(m',\id)$.
We have assumed that no correct process \kl-cast $m'\neq m$, therefore $m'=m$. Contradiction.

We conclude that, under these assumptions, no correct process ur-broadcasts $\Endorsem{}(m',\id)$ with $m' \neq m$, be it at line~\ref{SFKL-bcast} (by assumption) or at line~\ref{SFKL-fwd} (shown by this proof). As a result, $\knfp = 0$.
\end{proof}

\sfkllocaldelivery*

\begin{proof}
Let us assume that no correct process \kl-casts $(m',\id)$ with $m' \neq m$.
No correct process therefore ur-broadcasts $\Endorsem{}(m',\id)$ with $m'\neq m$ at line~\ref{SFKL-bcast}.
Lemma~\ref{lem:klcast-if-fwd} also applies and no correct process forwards $\Endorsem{}(m',\id)$ with $m'\neq m$ at line~\ref{SFKL-fwd} either, so $\knfp = 0$.
Because no correct process ur-broadcasts $\Endorsem{}(m',\id)$  with $m'\neq m$ whether at line~\ref{SFKL-bcast} or~\ref{SFKL-fwd}, a correct process receives at most $\tb$ messages $\Endorsem{}(m',\id)$ (all coming from Byzantine processes).
As by \sfklassum~\ref{assum:base}, $\tb < \qd$, no correct process \kl-delivers $(m',\id)$ with $m' \neq m$ at line~\ref{SFKL-dlv}.

We now prove the contraposition of the Lemma.
Let us assume no correct process \kl-delivers $(m,\id)$.
Since, by our earlier observations, no correct process \kl-delivers $(m',\id)$ with $m'\neq m$ either, the condition at line~\ref{SFKL-cond-dlv} implies that no correct process ever receives at least $\qd$ $\Endorsem{}(m,\id)$, and therefore $\lef = 0$.
By Lemma~\ref{lem:l-fact-min} we have $c(\qf-1) \geq (\kur+\kf)(c-\tm-\qd+\qf)$. 
\sfklassum~\ref{assum:base} implies that $c-\tm-\qd \geq 0 \iff c-\tm-\qd+\qf > 0$ (as $\qf \geq \tb+1 \geq 1$), leading to $\kur+\kf \leq \frac{c(\qf-1)}{c-\tm-\qd+\qf}$.
Because of the condition at line~\ref{SFKL-cond-bcast}, a correct process $p_j$ that has \kl-cast $(m,\id)$ but has not ur-broadcast $\Endorsem{}(m,\id)$ at line~\ref{SFKL-bcast} has necessarily ur-broadcast $\Endorsem{}(m,\id)$ at line~\ref{SFKL-fwd}.
We therefore have $\ki \leq \kur + \kf$, which gives $\ki \leq \frac{c(\qf-1)}{c-\tm-\qd+\qf}$.
\ta{maybe we lose some tightness here: by minoring $\kur+\kf$ by \ki, we only consider the case where the number of \kl-casts alone is sufficient, not the case where the number of \kl-casts PLUS the number of forwards is sufficient}
By contraposition, if $\ki > \frac{c(\qf-1)}{c-\tm-\qd+\qf}$, then at least one correct process must \kl-deliver $(m,\id)$.
Hence, we have $\kld = \left\lfloor \frac{c(\qf-1)}{c-\tm-\qd+\qf} \right\rfloor + 1$.
\end{proof}

\singleifknfp*

\begin{proof}
Let us consider a correct process $p_i \in A \cup B$. If we assume $p_i \not\in F$, $p_i$ never executes line~\ref{SFKL-fwd} by definition. Because $p_i \in A \cup B$, $p_i$ has received at least \qf messages $\Endorsem{}(m,\id)$, and therefore did not fulfill the condition at line~\ref{SFKL-cond-fwd} when it received its $\qf$\textsuperscript{th} message $\Endorsem{}(m,\id)$.
As $\single = \ffalse$ by Lemma assumption, to falsify this condition, $p_i$ must have had already ur-broadcast $\Endorsem{}(m,\id)$ when this happened. Because $p_i$ never executes line~\ref{SFKL-fwd}, this implies that $p_i$ ur-broadcasts $\Endorsem{}(m,\id)$ at line~\ref{SFKL-bcast}, and therefore $p_i \in \NF$. This reasoning proves that $A \cup B \setminus F \subseteq \NF$. As the sets $F$, $\NF$ and $\NB$ partition $A \cup B$, this shows that $\NB=\emptyset$, and $\knfp=|\emptyset|=0$.
\end{proof}

\polynomifone*

\begin{proof}
Let us write \wAb the total number of $\Endorsem{}(m,\id)$ messages from Byzantine processes received by the processes of $A$, and $\wA=\wAc+\wAb$ the total of number $\Endorsem{}(m,\id)$ messages received by the processes of $A$, whether these \EndorseMsgs originated from correct or Byzantine senders. By definition, 
$\wAb \leq \tb\lef$ and $\wA \geq \qd \lef$.
By combining these two inequalities with (\ref{eq:sup:on:wAc}) on \wAc we obtain:
\begin{align}
    \qd\lef \leq \wA = \wAc + \wAb &\leq (\kur+\kf)\lef + \tb\lef = (\tb+\kur+\kf)\lef, \nonumber\\
    \qd &\leq \tb+\kur+\kf
    \tag{as $\lef>0$}, \nonumber\\
    \qd-\tb &\leq \kur+\kf = x. \label{eq:kur:lambda:qmt}
\end{align}

This proves the first inequality of the lemma.
The processes in $A \cup B$ each receive at most $\kur+\kf$ distinct $\Endorsem{}(m,\id)$ messages from correct processes, so we have $\wAc + \wBc \leq (\knf+\kf+\knfp)(\kur+\kf)$. \ta{tight?}
Combined with the inequalities~(\ref{eq:sup:on:wCc}) on \wCc and (\ref{eq:sup:all:witness}) on $\wAc+\wBc+\wCc$ that remain valid in this case, we now have:
\begin{align}
    & (\knf+\kf+\knfp)(\kur+\kf) + (\qf-1)(c-\knf-\knfp-\kf) \geq (\kur+\kf)(c-\tm), \nonumber\\
    & (\knf+\kf+\knfp)(\kur+\kf-\qf+1) \geq (\kur+\kf)(c-\tm) - c(\qf-1). \label{eq:kappaLambda:geq}
\end{align}
Let us determine the sign of $(\kur+\kf-\qf+1)$.
We derive from (\ref{eq:kur:lambda:qmt}):
\begin{align}
    \kur+\kf-\qf+1 &\geq \qd-\tb-\qf+1 \nonumber\\
    &\geq 1 > 0. \tag{as $\qd-\qf \geq \tb$ by \sfklassum~\ref{assum:base}}
\end{align}
As $(\kur+\kf-\qf+1)$ is positive and we have $\kur \geq \knf$ by definition, we can transform (\ref{eq:kappaLambda:geq}) into: \ta{maybe tightness loss here}
\begin{align*}
    (\kur+\kf+\knfp)(\kur+\kf-\qf+1) &\geq (\kur+\kf)(c-\tm) - c(\qf-1),\\
    (x+\knfp)(x-\qf+1)&\geq x(c-\tm) - c(\qf-1), \tag{as $x=\kur+\kf$}\\
    x^2 - x(c-\tm+\qf-1-\knfp) &\geq -(c-\knfp)(\qf-1). \qedhere
\end{align*}
\end{proof}

\enoughifone*

\begin{proof}
By Lemma~\ref{lem:polynom-if-one} we have:
\begin{align}
    x^2& - x(c-\tm+\qf-1-\knfp) \geq -(c-\knfp)(\qf-1), \label{eq:polynom-if-one:start}
\end{align}
As (\ref{eq:polynom-if-one:start}) holds for all, values of $c \in [n-\tb,n]$, we can in particular consider $c=n-\tb$. Moreover, as by hypothesis, $\knfp=0$, we have. 
\begin{align}
  x^2& - x(n-\tb -\tm+ \qf - 1 ) + (\qf -1) (n-\tb) \geq 0, \nonumber\\
    x^2& - \alpha x + (\qf -1)(n-\tb) \geq 0. & \text{(by definition of $\alpha$)} \label{eq:ineq:polynomial:no:knfp} 
\end{align}

Let us first observe that the discriminant of the second-degree polynomial in (\ref{eq:ineq:polynomial:no:knfp}) is non negative, i.e. $\alpha^2-4(\qf-1)(n-\tb) \geq 0$ by \sfklassum~\ref{assum:disc}.
This allows us to compute the two real-valued roots as follows:
\begin{align*}
    r_0 &= \frac{\alpha}{2} - \frac{\sqrt{\alpha^{2} - 4 (\qf -1)(n - \tb)}}{2} \text{~~~~~~and~~~~~~}
    r_{1} = \frac{\alpha}{2} + \frac{\sqrt{\alpha^{2} - 4 (\qf -1)(n - \tb)}}{2}.
\end{align*}
Thus (\ref{eq:ineq:polynomial:no:knfp}) is satisfied if and only if $x \leq r_0 \lor x \geq r_1$.

\begin{itemize}
    \item Let us prove $r_0 \leq \qd-1-\tb$.
    We need to show that: 
    \begin{align*}
        \frac{\alpha}{2} - \frac{\sqrt{\alpha^{2} - 4 (\qf -1)(n - \tb)}}{2} & \leq \qd -1 -\tb \\
        \frac{\alpha}{2}-(\qd-1)+\tb &\leq \frac{\sqrt{\alpha^2-4(\qf-1)(n-\tb)}}{2} \\
        \frac{\sqrt{\alpha^2-4(\qf-1)(n-\tb)}}{2} &\geq \frac{\alpha}{2}-(\qd-1)+\tb \\
        \sqrt{\alpha^2-4(\qf-1)(n-\tb)} &\geq \alpha-2(\qd-1)+2\tb.
    \end{align*}
    
    The inequality is trivially satisfied if $\alpha-2(\qd-1)+2\tb < 0$.
    For all other cases, we need to verify that: 
    \begin{align*}
        \alpha^{2} - 4 (\qf -1)(n - \tb)  & \geq (\alpha - 2(\qd -1) + 2\tb)^2, \\
        \alpha^{2} - 4 (\qf -1)(n - \tb)  & \geq \alpha^2 + 4(\qd -1)^2 + 4\tb^2 -4 \alpha(\qd -1) +4\alpha\tb-8\tb(\qd-1), \\
        - 4 (\qf -1)(n - \tb)  & \geq  4(\qd -1)^2 + 4\tb^2 -4 \alpha(\qd -1) +4\alpha\tb-8\tb(\qd-1), \\
        - (\qf -1)(n - \tb)  & \geq  (\qd -1)^2 + \tb^2 - \alpha(\qd -1) +\alpha\tb-2\tb(\qd-1), \\
        - (\qf -1)(n - \tb)  & \geq  (\qd -1 - \tb)^2 - \alpha(\qd -1 -\tb),
    \end{align*}
    and thus $\alpha(\qd-1-\tb)-(\qf-1)(n-\tb)-(\qd-1-\tb)^2 \geq 0$, which is true by \sfklassum~\ref{assum:r0}.

    \item Let us prove $r_1 > \qd-1$.
    We want to show that: 
    \begin{align}
        \frac{\alpha}{2} + \frac{\sqrt{\alpha^{2} - 4 (\qf -1)(n - \tb)}}{2} & > \qd -1 \nonumber
    \end{align}
    
    Let us rewrite the inequality as follows: 
    \begin{align}
        \alpha + \sqrt{\alpha^{2} - 4 (\qf -1)(n - \tb)} & > 2(\qd -1) \nonumber\\ 
        \sqrt{\alpha^{2} - 4 (\qf -1)(n - \tb)} & > 2(\qd-1)-  \alpha   \nonumber
    \end{align}
    The inequality is trivially satisfied if $2(\qd-1) -\alpha < 0$. 
    For all other cases, we can take the squares as follows:
    \begin{align*}
        \alpha^{2} - 4 (\qf -1)(n - \tb) & > (2(\qd-1)- \alpha)^2, \\
        \alpha^{2} - 4 (\qf -1)(n - \tb) & > 4(\qd-1)^2 + \alpha^2 -4\alpha(\qd -1), \\
        - 4 (\qf -1)(n - \tb) & > 4(\qd-1)^2  -4\alpha(\qd -1), \\
        4\alpha(\qd -1) - 4 (\qf -1)(n - \tb) - 4(\qd-1)^2 & > 0, \\
        \alpha(\qd -1) -  (\qf -1)(n - \tb) - (\qd-1)^2 & > 0,
    \end{align*}
    which is true by \sfklassum~\ref{assum:r1}.
\end{itemize}

We now know that $r_0 \leq \qd-1-\tb$ and that $r_1 > \qd-1$.
In addition, as $x \leq r_0 \lor x \geq r_1$, we have $x \leq \qd-\tb-1 \lor x > \qd-1$.
But Lemma~\ref{lem:polynom-if-one} states that $x \geq \qd-\tb$, which is incompatible with $x \leq \qd-\tb-1$.
So we are left with $x > \qd - 1$, which implies, as $\qd$ and $x$ are integers that $x \geq \qd$, thus proving the lemma for $c=n-\tb$. 

Let us now consider the set $E_0$ of all executions in which $\tb$ processes are Byzantine, and therefore $c=n-\tb$, and a set $E_c$ of executions in which there are fewer Byzantine processes, and thus $c > n-\tb$ correct processes. 
We show that $E_c \subseteq E_0$ in that a Byzantine process can always simulate the behavior of a correct process.
In particular, if the simulated correct process is not subject to the message adversary, the simulating Byzantine process simply operates like a correct process.
If, on the other hand, the simulated correct process misses some messages as a result of the message adversary, the Byzantine process can also simulate missing such messages.
As a result, the executions that can happen when $c > n-\tb$ can also happen when $c=n - \tb$.
Thus our result proven for $c=n-\tb$ can be extended to all possible values of $c$.
\end{proof}

\ellifenough*

\begin{proof}
As $\knfp=0$ and $\kur+\kf\geq \qd$, we can rewrite the inequality of Lemma~\ref{lem:l-fact-min} into:
\begin{align*}
    \lef\times(\kur+\kf-\qd+1) \geq (\kur+\kf)(c-\tm-\qd+\qf) - c(\qf-1).
\end{align*}

From $\kur+\kf \geq \qd$ we derive $\kur+\kf-\qd+1>0$, and we transform the above inequality into:
\begin{align*}
    \lef
    &\geq \frac{(\kur+\kf)(c-\tm-\qd+\qf) - c(\qf-1)}{\kur+\kf-\qd+1}.
\end{align*}
Let us now focus on the case in which $c=n - \tb$, we obtain:
\begin{align*}
    \lef
    &\geq \frac{(\kur+\kf)(n-\tb-\tm-\qd+\qf) - (n-\tb)(\qf-1)}{\kur+\kf-\qd+1}.
\end{align*}
The right side of the inequality is of the form: 
\begin{align}
    \lef &\geq \frac{\phi x - \beta}{x-\gamma} = \phi + \frac{\phi\gamma - \beta}{x-\gamma}
    \label{eq:case1:ell:ge}
\end{align}
with:
\begin{align*}
    x &= \kur+\kf, \\
    \gamma &= \qd-1, \\
    \alpha &= n-\tb-\tm+\qf-1, \\
    \phi &= n-\tb-\tm-\qd+\qf, \\
    \beta &= c(\qf-1).
\end{align*}
Since, by hypothesis, $x = \kur+\kf \geq \qd$, we have:
\begin{equation}
    x-\gamma = \kur+\kf-\qd+1 > 0. \label{eq:xmgamma:gz}
\end{equation}
We also have:
\begin{align}
    \phi\gamma - \beta &= (\alpha -\gamma)\gamma - c(\qf-1) = \alpha\gamma-\gamma^2 - c(\qf-1), \nonumber\\
    &= \alpha(\qd-1)-(\qd - 1)^2 - (n-\tb)(\qf -1) > 0, \tag{by \sfklassum~\ref{assum:r1}} \nonumber\\
    \phi\gamma-\beta &> 0. \label{eq:phigammambeta:gz}
\end{align}

Injecting (\ref{eq:xmgamma:gz}) and (\ref{eq:phigammambeta:gz}) into (\ref{eq:case1:ell:ge}), we conclude that $\phi + \frac{\phi\gamma - \beta}{x-\gamma}$ is a {\it decreasing  hyperbole} defined over $x \in ]\gamma,\infty]$ with {\it asymptotic value} $\phi$ when $x \rightarrow \infty$. As $x$ is a number of correct processes, $x\leq c$. The decreasing nature of the right-hand side of (\ref{eq:case1:ell:ge}) leads us to:
$ \lef \geq \phi + \frac{\phi\gamma - \beta}{c-\gamma}
= \frac{\phi c - \beta}{c-\gamma}
\geq \frac{c(c-\tm-\qd+\qf) - c(\qf-1)}{c-\qd+1}
\geq c\times\frac{c-\tm-\qd+1}{c-\qd+1}
= c\left(1-\frac{\tm}{c-\qd+1}\right)$.

Since \lef is a positive integer, we conclude that at least $\ell_{\mathsf{min}}=\lval$ correct processes receive at least \qd message $\Endorsem{}(m,id)$ at line~\ref{SFKL-cond-dlv}.
As each of these processes either \kl-delivers $(m,\id)$ when this first happens, or has already \kl-delivered another \app $m' \neq m$ with identity \id, we conclude that at least $\ell_{\mathsf{min}}$ correct processes \kl-deliver some \app (whether it be $m$ or $m' \neq m$) with identity \id when $c=n - \tb$.
The reasoning for extending this result to any value of $c \in [n-\tb,n]$ is identical to the one at the end of the proof of Lemma~\ref{lemma:ifone:then:enough:internal} just above.
\end{proof}

\sfklweakglobaldelivery*

\begin{proof}
Let us assume $\single = \ffalse$, and one correct process \kl-delivers $(m,\id)$.
By Lemma~\ref{lem:single-if-knfp}, $\knfp = 0$.
The prerequisites for Lemma~\ref{lemma:ifone:then:enough:internal} are verified, and therefore $\kur+\kf\geq \qd$.
This provides the prerequisites for Lemma~\ref{lemma:if:enough:internal:then:enough:ell}, from which we conclude that at least $\lgd = \lval$ correct processes \kl-deliver an \app $m'$ with identity \id, which concludes the proof of the lemma.
\end{proof}

\sfklstrongglobaldelivery*

\begin{proof}
Let us assume that \textit{(i)} $\single = \ttrue$, \textit{(ii)} no correct process \kl-casts $(m',\id)$ with $m' \neq m$, and \textit{(iii)} one correct process \kl-delivers $(m,\id)$.
Lemma~\ref{lem:klcast-if-fwd} holds and implies that $\knfp = 0$.
From there, as above, Lemmas~\ref{lemma:ifone:then:enough:internal} and~\ref{lemma:if:enough:internal:then:enough:ell} hold, and at least $\lgd=\lval$ correct processes \kl-deliver an \app for identity \id.
  
By hypothesis, no correct process ur-broadcasts $\Endorsem{}(m',\id)$ at line~\ref{SFKL-bcast} with $m'\neq m$. Similarly, because of Lemma~\ref{lem:klcast-if-fwd}, no correct process ur-broadcasts $\Endorsem{}(m',\id)$ at line~\ref{SFKL-fwd} with $m'\neq m$.
As a result, a correct process can receive at most receive \tb messages $\Endorsem{}(m',\id)$ at line~\ref{SFKL-cond-dlv} (all from Byzantine processes).
As $\qd>\tb$ (by \sfklassum~\ref{assum:base}), the condition of line~\ref{SFKL-cond-dlv} never becomes true for $m'\neq m$, and as result no correct process delivers an \app $m' \neq m$ with identity \id.
All processes that \kl-deliver an \app with identity \id, therefore, \kl-deliver $m$, which concludes the lemma.
\end{proof}

\section{Proof of the Signature-Free MBRB Implementations}
\label{sec:proofsMBRB}
The proofs that follow use integer arithmetic. Given a real number $x$ and an integer $i$, let us recall that $x-1 < \lfloor x \rfloor \leq x \leq \lceil x \rceil < x+1$, $\lfloor x+i\rfloor = \lfloor x\rfloor +i$, $\lceil x+i\rceil = \lceil x\rceil +i$, $\lfloor -x \rfloor = -\lceil x \rceil$, $(i > x) \iff (i \geq \lfloor x \rfloor+1)$, $(i < x) \iff (i \leq \lceil x \rceil-1)$.
\ft{I don't think the inequalities on the rations are true. Removing them.}%

\subsection{Proof of MBRB with Bracha's reconstructed algorithm (Algorithm~\ref{alg:b-mbrb})} \label{sec-proof-b-mbrb}

\subsubsection{Instantiating the parameters of the \texorpdfstring{\kl}{k2l}-cast objects}
\label{sec:inst-param-klcast}
In Alg.~\ref{alg:b-mbrb} (page~\pageref{alg:b-mbrb}), we instantiate the \kl-cast objects \obje and \objr using the signature-free implementation presented in Section \ref{sec-klcast-sf}. Let us mention that, given that $\obje.\single = \objr.\single = \ttrue$, then we use the strong variant of the global-delivery property of \kl-cast (\klSGDprop) for both objects \obje and \objr.
Moreover, according to the definitions of $k'$, $k$, $\ell$ and $\delta$ (page~\pageref{def-kl-parameters}) and their values stated in
Theorem~\ref{theo:sf-kl-correctness}, we have:
\begin{itemize}
    \item $\obje.\kv = \obje.\qf-n+c = \tb+1-n+c \geq \tb+1-\tb = 1$,
    
    \item $\begin{aligned}[t]
    \obje.\kld &=
    \left\lfloor \frac{c(\obje.\qf-1)}{c-\tm-\obje.\qd+\obje.\qf} \right\rfloor + 1
    = \left\lfloor \frac{c(\tb+1-1)}{c-\tm-\lfloor \frac{n+\tb}{2} \rfloor-1+\tb+1} \right\rfloor + 1 \\
    &= \left\lfloor \frac{c\tb}{c-\tm-\lfloor \frac{n-\tb}{2} \rfloor} \right\rfloor + 1,
    \end{aligned}$
    
    \item $\begin{aligned}[t]
    \obje.\lgd &=
    \left\lceil c\left(1-\frac{\tm}{c-\obje.\qd+1}\right) \right\rceil
    = \left\lceil c\left(1-\frac{\tm}{c-\lfloor \frac{n+\tb}{2} \rfloor-1+1}\right) \right\rceil \\
    &= \left\lceil c\left(1-\frac{\tm}{c-\lfloor \frac{n+\tb}{2} \rfloor}\right) \right\rceil,
    \end{aligned}$
    
    \item $\begin{aligned}[t]
    \obje.\nodpty &= \left(\left(\obje.\qf > \frac{n+\tb}{2}\right) \lor \left(\obje.\single \land \obje.\qd > \frac{n+\tb}{2}\right)\right) \\
    &= \left(\left(\tb+1 > \frac{n+\tb}{2}\right) \lor \left(\ttrue \land \left\lfloor \frac{n+\tb}{2} \right\rfloor +1 > \frac{n+\tb}{2} \right)\right) \\
    &= (\ffalse \lor (\ttrue \land \ttrue))
    = \ttrue,
    \end{aligned}$
    
    \item $\objr.\kv = \objr.\qf-n+c = \tb+1-n+c \geq \tb+1-\tb = 1$,
    
    \item $\begin{aligned}[t]
    \objr.\kld
    &= \left\lfloor \frac{c(\objr.\qf-1)}{c-\tm-\objr.\qd+\objr.\qf} \right\rfloor + 1
    = \left\lfloor \frac{c(\tb+1-1)}{c-\tm-2\tb-\tm-1+\tb+1} \right\rfloor + 1 \\
    &= \left\lfloor \frac{c\tb}{c-2\tm-\tb} \right\rfloor + 1,
    \end{aligned}$
    
    \item $\begin{aligned}[t]
    \objr.\lgd
    &= \left\lceil c\left(1-\frac{\tm}{c-\objr.\qd+1}\right) \right\rceil
    = \left\lceil c\left(1-\frac{\tm}{c-2\tb-\tm-1+1}\right) \right\rceil \\
    &= \left\lceil c\left(1-\frac{\tm}{c-2\tb-\tm}\right) \right\rceil,
    \end{aligned}$
    
    \item $\begin{aligned}[t]
    \objr.\nodpty
    &= \left(\left( \objr.\qf > \frac{n+\tb}{2} \right) \lor \left( \objr.\single \land \objr.\qd > \frac{n+\tb}{2} \right)\right)\\
    &= \left(\left( \tb+1 > \frac{n+\tb}{2} \right) \lor \left( \ttrue \land 2\tb+\tm+1 > \frac{n+\tb}{2} \right)\right)
    \in \{\ttrue,\ffalse\}
    \end{aligned}$
\end{itemize}

We recall that parameter \nodpty controls the conditional no-duplicity property.
The value for $\obje.\nodpty$ is \ttrue, but that of value for $\objr.\nodpty$ may be either \ttrue or \ffalse depending on the values of $n$, \tb, and \tm.
This is fine because, in Bracha's reconstructed algorithm (Alg.~\ref{alg:b-mbrb}), it is the first round (\obje) that ensures no-duplicity. Once this has happened, the second round (\objr) does not need to provide no-duplicity but only needs to guarantee the termination properties of local and global delivery.
This observation allows \objr to operate with lower values of \qd and \qf.

Finally, we observe that for Alg.~\ref{alg:b-mbrb}, \sfklassum~\ref{assum:base} through~\ref{assum:r0} are all satisfied by \bassum $n > \brBoundN$. 
We prove this fact in Appendix~\ref{lem:proof-objE}.
In the following, we prove that $\brBoundN \geq \brBound \geq 3\tb+2\tm$.
\df{I reworked this piece of text after the itemize, removing what was there and adding these two paragraphs.}



%


%
    

\begin{observation}
\label{obs:boundBracha}
For $\tm, \tb \in \mathbb{N}_0$ non-negative integers\df{It actually works for real numbers too. What shall we say?}, we have: 
\begin{align*}
    \brBoundN \geq \brBound \geq 3\tb+2\tm.
\end{align*}
\end{observation}

\begin{proof}
Let us start by proving the first inequality. 
\begin{align*}
    \tb^2 + 6\tb\tm + \tm^2 + 4\sqrt{\tb\tm}(\tb+\tm) & \geq \tb^2 + 6\tb\tm + \tm^2, \\
    \tb^2 + \tm^2 + 4\tb\tm + 4\tb\sqrt{\tb\tm}+4\tm\sqrt{\tb\tm}+2\tb\tm & \geq \tb^2 + 6\tb\tm + \tm^2, \\
    (\tb + \tm +2\sqrt{\tb\tm})^2 & \geq \tb^2 + 6\tb\tm + \tm^2, \\
    \tb + \tm +2\sqrt{\tb\tm} & \geq \sqrt{\tb^2 + 6\tb\tm + \tm^2}, \\
    3\tb + 2\tm +2\sqrt{\tb\tm} & \geq 2\tb+\tm + \sqrt{\tb^2 + 6\tb\tm + \tm^2}.
\end{align*}
Let us then prove the second inequality:
\begin{align*}
    \tb^2+6\tb\tm+\tm^2 &\geq \tb^2+2\tb\tm+\tm^2 = (\tb+\tm)^2, \\
    \sqrt{\tb^2 + 6\tb\tm + \tm^2} &\geq \tb+\tm, \\
    2\tb+\tm + \sqrt{\tb^2 + 6\tb\tm + \tm^2} &\geq 3\tb+2\tm. \qedhere
\end{align*}
\end{proof}

\subsubsection{Proof of satisfaction of the assumptions of Algorithm~\ref{alg:sf-klcast}}
In this section, we prove that all the assumptions of the signature-free \kl-cast implementation presented in Alg.~\ref{alg:sf-klcast} (page~\pageref{alg:sf-klcast}) are well respected for the two \kl-cast instances used in Alg.~\ref{alg:b-mbrb} (\obje and \objr).

\begin{lemma}
\label{lem:proof-objE}
Alg.~{\em \ref{alg:sf-klcast}}'s assumptions are well respected for \obje.
\end{lemma}

\begin{proof}
Let us recall that 
$\qf = \tb+1$ and $\qd= \lfloor \frac{n+\tb}{2} \rfloor+1$ for \obje.

\begin{itemize}
\item \textit{Proof of satisfaction of \sfklassum~{\em\ref{assum:base}}} ($c-\tm \geq \obje.\qd \geq \obje.\qf+\tb \geq 2\tb+1$):

By \bassum and Observation~\ref{obs:boundBracha}, we have the following:
\begin{align}
    c-\tm &\geq n-\tb-\tm = \frac{2n-2\tb-2\tm}{2}, \tag{by definition of $c$}\\
    &> \frac{n+3\tb+2\tm-2\tb-2\tm}{2} = \frac{n+\tb}{2}, \tag{as $n > 3\tb+2\tm$}\\
    &\geq \left\lfloor \frac{n+\tb}{2} \right\rfloor+1. \label{b-echo-asmp-1-1}\\
\intertext{We also have:}
    \left\lfloor \frac{n+\tb}{2} \right\rfloor+1 &\geq \left\lfloor \frac{3\tb+2\tm+1+\tb}{2} \right\rfloor+1, \tag{as $n > 3\tb+2\tm$}\\
    &\geq \lfloor 2\tb+\tm+1/2 \rfloor+1 = 2\tb+\tm+1 \geq 2\tb+1. \label{b-echo-asmp-1-2}\\
\intertext{By combining (\ref{b-echo-asmp-1-1}) and (\ref{b-echo-asmp-1-2}), we get:}
    & c-\tm \geq \left\lfloor \frac{n+\tb}{2} \right\rfloor +1 \geq 2\tb+1 \geq 2\tb+1, \nonumber\\
    & c-\tm \geq \obje.\qd \geq \obje.\qf+\tb \geq 2\tb+1. \tag{\sfklassum~\ref{assum:base}}
\end{align}

\item \textit{Proof of satisfaction of \sfklassum~{\em\ref{assum:disc}}} ($\alpha^2-4(\obje.\qf-1)(n-\tb) \geq 0$):

Let us recall that, for object \obje, we have $\qf=\tb+1$ and $\qd= \lfloor \frac{n+\tb}{2} \rfloor + 1$.
We therefore have $\alpha= n+\qf-\tb-\tm-1 = n-\tm$.
Let us now consider the quantity:
\begin{align}
    \Delta &= \alpha^{2} - 4 (\qf - 1)(n - \tb) = (n - \tm)^2 - 4 \tb (n - \tb)\nonumber\\
    & = 4 \tb^{2} + \tm^{2} + n^{2} + n \left(- 4 \tb - 2 \tm\right)\nonumber
\end{align}
The inequality is satisfied if $n>2\sqrt{\tb\tm} + 2\tb+\tm$, which is clearly the case as $n>\brBoundN$.
This proves \sfklassum~\ref{assum:disc}.

\item \textit{Proof of satisfaction of \sfklassum~{\em\ref{assum:r1}}} ($\alpha(\obje.\qd-1)-(\obje.\qf-1)(n-\tb)-(\obje.\qd-1)^2 > 0$):

Let us consider the quantity on the left-hand side of \sfklassum~\ref{assum:r1} and substitute 
$\qf=\tb+1$, $\qd= \lfloor \frac{n+\tb}{2} \rfloor + 1$:
\begin{align}
    & ~~~~ \alpha (\qd - 1) - (\qf-1) (n-\tb) - (\qd - 1)^{2}, \nonumber\\
    & =(n + \qf -\tb -\tm - 1)(\qd - 1) - (\qf-1) (n-\tb) - (\qd - 1)^{2}, \nonumber\\
    & =(n  -\tm )\left(\left \lfloor \frac{n+\tb}{2} \right\rfloor\right) - \tb (n-\tb) - \left(\left\lfloor \frac{n+\tb}{2} \right\rfloor\right)^{2}. \label{eq:brEr1-beforeCases}
\end{align}

We now observe that $\left(\left \lfloor \frac{n+\tb}{2} \right\rfloor\right) = \left(\frac{n+\tb-\epsilon}{2} \right)$ with $\epsilon=0$ if $n+\tb=2k$ is even, and $\epsilon=1$ if $n+\tb=2k+1$ is odd.
We thus rewrite (\ref{eq:brEr1-beforeCases}) as follows:
\begin{align}
    & ~~~~ (n  -\tm )\left(\frac{n+\tb-\epsilon}{2} \right) - \tb (n-\tb) - \left( \frac{n+\tb-\epsilon}{2} \right)^{2}, \nonumber\\
    & =\frac{n+\tb - \epsilon}{2} \times \frac{2n-2\tm-n-\tb+\epsilon}{2} - \tb (n-\tb), \nonumber\\
    & =\frac{(n+\tb - \epsilon)( n-2\tm-\tb+\epsilon) - 4\tb (n-\tb)}{4}, \nonumber\\
    & =\frac{n^{2} - \tb^{2} - 2 \tb \tm + 2 \tb \epsilon - 2n \tm  + 2 \tm \epsilon  - \epsilon^{2} - 4n\tb +4\tb^2}{4}, \nonumber\\
    & =\frac{n^{2} +3 \tb^{2} - 2 \tb \tm - 2 n(\tm +2\tb) + \epsilon(2 \tb   + 2 \tm   - \epsilon) }{4}. \nonumber
\end{align}

As we want to show that the above quantity is positive, the result will not change if we multiply it by 4:
{\allowdisplaybreaks
\begin{align}
    & ~~~~ n^{2} +3 \tb^{2} - 2 \tb \tm - 2 n(\tm +2\tb) + \epsilon(2 \tb   + 2 \tm   - \epsilon)>0. \label{eq:assum:r1brEbeforeN}
\end{align}
}
We now solve the inequality to obtain: 
\begin{align*}
    n &> 2 \tb + \tm + \sqrt{\tb^{2} + 6 \tb \tm  + \tm^{2} - \epsilon(2 \tb + 2 \tm - \epsilon)}.
\end{align*}

We observe that, for $\tb+\tm\geq 1$, the quantity $- \epsilon(2 \tb + 2 \tm - \epsilon)$ is strictly negative if $\epsilon=1$, therefore if $\epsilon=1\lor\tb+\tm\geq 1$:
\begin{align*}
    n & > \brBoundN, \\
    & \geq \tb + \tm + \sqrt{\tb^{2} + 6 \tb \tm  + \tm^{2}}, \tag{by Observation~\ref{obs:boundBracha}}\\
    & \geq 2 \tb + \tm + \sqrt{\tb^{2} + 6 \tb \tm  + \tm^{2} - \epsilon(2 \tb + 2 \tm - \epsilon)}.
\end{align*}

This leaves out the case $(\tb=\tm=0) \land (n=2k+1 \text{ is odd})$, for which we can show that (\ref{eq:assum:r1brEbeforeN}) is positive or null for $n \geq 1$:
\begin{align*}
    (\ref{eq:assum:r1brEbeforeN}):& ~~~~ n^{2} +3 \tb^{2} - 2 \tb \tm - 2 n(\tm +2\tb) + \epsilon(2 \tb   + 2 \tm   - \epsilon), \\
    &=  n^{2} - 1 \geq 0 \text{ for } n \geq 1.
\end{align*}
This completes the proof of \sfklassum~\ref{assum:r1}.

\item \textit{Proof of satisfaction of \sfklassum~{\em\ref{assum:r0}}} ($\alpha(\obje.\qd-1-\tb)-(\obje.\qf-1)(n-\tb)-(\obje.\qd-1-\tb)^2 \geq 0$):

Let us consider the quantity on the left-hand side of \sfklassum~\ref{assum:r0} and substitute 
$\qf=tb+1$, $\qd= \lfloor \frac{n+\tb}{2} \rfloor + 1$:
\begin{align}
    & ~~~~ \alpha(\qd -1 -\tb) - (\qf -1)(n - \tb)  - (\qd -1 - \tb)^2,  \nonumber\\
    & = (n + \qf -\tb -\tm - 1)(\qd -1 -\tb) - (\qf -1)(n - \tb)  - (\qd -1 - \tb)^2, \nonumber\\
    & = (n  -\tm)\left(\left\lfloor \frac{n+\tb}{2} \right \rfloor  -\tb\right) - \tb (n - \tb)  - \left(\left\lfloor \frac{n+\tb}{2} \right \rfloor - \tb\right)^2. \label{eq:brEr0-beforeCases}
\end{align}

Like before, we observe that $\left(\left \lfloor \frac{n+\tb}{2} \right\rfloor\right) = \left(\frac{n+\tb-\epsilon}{2} \right)$ with $\epsilon=0$ if $n+\tb=2k$ is even, and $\epsilon=1$ if $n+\tb=2k+1$ is odd.
We thus rewrite (\ref{eq:brEr0-beforeCases}) as follows:
\begin{align*}
    & ~~~~ (n  -\tm)\left( \frac{n+\tb-\epsilon}{2}  -\tb\right) - \tb (n - \tb)  - \left( \frac{n+\tb-\epsilon}{2}  - \tb\right)^2, \\
    & = (n  -\tm)\cdot \frac{n-\tb-\epsilon}{2}   - \tb (n - \tb)  - \left( \frac{n-\tb-\epsilon}{2} \right)^2, \\
    & = \frac{n-\tb-\epsilon}{2}  \cdot\frac{2n  -2\tm - n+\tb+\epsilon}{2} - \tb (n - \tb), \\
    & = \frac{(n-\tb-\epsilon)(n  -2\tm +\tb+\epsilon) -4n\tb +4\tb^2}{4}, \\
    & = \frac{- \tb^{2} + 2 \tb \tm - 2 \tb \epsilon + 2 \tm \epsilon - 2 \tm n - \epsilon^{2} + n^{2}}{4}.
\end{align*}

As we want to show that the above quantity is non-negative, the result will not change if we multiply it by 4:
\begin{align*}
    & ~~~~ - \tb^{2} + 2 \tb \tm - 2 \tb \epsilon + 2 \tm \epsilon - \epsilon^{2} - 2 \tm n  + n^{2}.
\end{align*}

We then solve the inequality to obtain: $n\geq \maxfn(\tb + \epsilon, - \tb + 2 \tm - \epsilon)$, which is clearly satisfied as $n\geq \brBoundN+1$. This proves all previous inequality and thus  \sfklassum~\ref{assum:r0}. \qedhere

\end{itemize}%
\end{proof}

\begin{lemma}
Alg.~{\em \ref{alg:sf-klcast}}'s assumptions are well respected for \objr.
\end{lemma}

\begin{proof}
Let us recall that
$\qf= \tb+1$ and $\qd = 2\tb+\tm+1$ for \objr.
Let us observe that we have then $\qd-\qf-\tb-\tm=0$.
\begin{itemize}
\item \textit{Proof of satisfaction of \sfklassum~{\em\ref{assum:base}}} ($c-\tm \geq \objr.\qd \geq \objr.\qf+\tb \geq 2\tb+1$):

From Observation~\ref{obs:boundBracha}, we have:
\begin{align*}
    & c-\tm \geq n-\tb-\tm \geq 3\tb+2\tm+1-\tb-\tm \geq 2\tb+\tm+1, \tag{as $n > 3\tb+2\tm$}\\
    & c-\tm \geq 2\tb+\tm+1 \geq 2\tb+1 \geq 2\tb+1, \\
    & c-\tm \geq \objr.\qd \geq \objr.\qf+\tb \geq 2\tb+1. \tag{\sfklassum~\ref{assum:base}}
\end{align*}

\item \textit{Proof of satisfaction of \sfklassum~{\em\ref{assum:disc}}} ($\alpha^2-4(\objr.\qf-1)(n-\tb) \geq 0$):

Let us recall that, for object \objr, we have $\qf=\tb+1$ and $\qd=2\tb+\tm + 1$.
As \sfklassum~\ref{assum:disc} depends on \qd but not on \qf, and since $\obje.\qf = \objr.\qf$, we refer the reader to the proof we gave in Lemma~\ref{lem:proof-objE} for \obje.

\item \textit{Proof of satisfaction of \sfklassum~{\em\ref{assum:r1}}} ($\alpha(\objr.\qd-1)-(\objr.\qf-1)(n-\tb)-(\objr.\qd-1)^2 > 0$):

Let us consider the quantity on the left-hand side of \sfklassum~\ref{assum:r1}: 
\begin{align}
    & ~~~~ \alpha (\qd - 1) - (\qf-1) (n-\tb) - (\qd - 1)^2, \nonumber \\
    & =(n + \qf -\tb -\tm - 1)(\qd - 1) - (\qf-1) (n-\tb) - (\qd - 1)^2, \nonumber \\
    & =(n -\tm )(2\tb+\tm) - \tb (n-\tb) - (2\tb+\tm)^2, \nonumber \\
    & =2n\tb +n\tm -2\tb\tm -\tm^2 - n\tb +\tb^2 - 4\tb^2-\tm^2-4\tb\tm, \nonumber \\
    & =n(\tb +\tm) -6\tb\tm -2\tm^2 -3\tb^2, \nonumber \\
    & =n(\tb +\tm) -(6\tb\tm +2\tm^2 +3\tb^2). \label{eq:assum:r1brR:beforebound}
\end{align}

Then, we observe that we can lower bound the quantity on the left side of (\ref{eq:assum:r1brR:beforebound}) by substituting \bassum, i.e. $n>\brBoundN \geq \brBound$.
For convenience, in the following we write $\rho=\tb^2+6\tb\tm+\tm^2$, thus $n>2\tb+\tm+\sqrt{\rho}$.
We get:
\begin{align*}
    & ~~~~ n(\tb+\tm)-(3\tb^2+6\tb\tm+2\tm^2), \nonumber\\
    &> (2\tb+\tm+\sqrt{\rho})(\tb+\tm)-(3\tb^2+6\tb\tm+2\tm^2), \\
    &= \sqrt{\rho}(\tb+\tm) - \tm^2 - \tb^2 - 3 \tb \tm.
\end{align*}

We now want to show that the above quantity is positive or null, i.e.: 
\begin{align}
    & ~~~~ \sqrt{\rho}(\tb+\tm) - \tm^{2} - \tb^{2} - 3 \tb \tm \geq 0.
    \label{eq:assum:r1brR:afterbound2}
\end{align}
We now rewrite (\ref{eq:assum:r1brR:afterbound2}) as follows:
\begin{align*}
   \sqrt{\rho}(\tb+\tm) &\geq \tm^{2} + \tb^{2} +  2\tb \tm + \tb\tm, \\
   \sqrt{\rho}(\tb+\tm) &\geq (\tm + \tb)^{2} + \tb\tm, \\
   (\tb^2+6\tb\tm+\tm^2)(\tb+\tm)^2 & \geq ((\tm + \tb)^{2} + \tb\tm)^2, \tag{as $(\tm + \tb)^{2} + \tb\tm \geq 0$} \\
   ((\tb+\tm)^2+4\tb\tm)(\tb+\tm)^2 & \geq ((\tm + \tb)^{2} + \tb\tm)^2, \\
   (\tb+\tm)^4+4\tb\tm(\tb+\tm)^2 & \geq (\tm + \tb)^4 + (\tb\tm)^2 + 2\tb\tm(\tb+\tm)^2, \\
   2\tb\tm(\tb+\tm)^2 &\geq (\tb\tm)^2, \\
   2\tb\tm(\tb^2+\tm^2+2\tb\tm) &\geq (\tb\tm)^2, \\
   2\tb\tm(\tb^2+\tm^2)+4(\tb\tm)^2 &\geq  (\tb\tm)^2, \\
   2\tb\tm(\tb^2+\tm^2)+3(\tb\tm)^2 &\geq 0.
\end{align*}

This proves (\ref{eq:assum:r1brR:afterbound2}) and all previous inequalities and ultimately \sfklassum~\ref{assum:r1}.

\item \textit{Proof of satisfaction of \sfklassum~{\em\ref{assum:r0}}} ($\alpha(\objr.\qd-1-\tb)-(\objr.\qf-1)(n-\tb)-(\objr.\qd-1-\tb)^2 \geq 0$):

Let us consider the quantity on the left-hand side of \sfklassum~\ref{assum:r0}:
\begin{align}
    & ~~~~ \alpha(\qd -1 -\tb) - (\qf -1)(n - \tb)  - (\qd -1 - \tb)^2, \label{eq:assum:r0brR:goal} \\
    & = (n + \qf -\tb -\tm - 1)(\qd -1 -\tb) - (\qf -1)(n - \tb)  - (\qd -1 - \tb)^2, \nonumber\\
    & = (n +  -\tm)(\tb+\tm) - \tb(n - \tb)  - (\tb+\tm)^2, \nonumber\\
    & = (\tb+\tm)(n +  -2\tm -\tb) - \tb(n - \tb), \nonumber\\
    & = n\tb+n\tm   -2\tb\tm -2\tm^2 -\tb^2 -\tb\tm - n\tb + \tb^2, \nonumber\\
    & = n\tm   -3\tb\tm -2\tm^2, \nonumber\\
    & = n\tm   -3\tb\tm -2\tm^2, \nonumber\\
    & = \tm (n - 3 \tb -2 \tm). \label{eq:assum:r0brR:beforebound} 
\end{align}

Like before, we observe that we can lower bound the quantity on the left side of (\ref{eq:assum:r0brR:beforebound}) by substituting \bassum, i.e., $n>\brBoundN \geq 3\tb+2\tm$, so we have:
\begin{align}
(\ref{eq:assum:r0brR:beforebound}): & ~~~~ \tm (n - 3 \tb -2 \tm) \nonumber\\
& > \tm (3\tb+2\tm-2\tm-3\tb) = 0. \label{eq:assum:r0brR:final}
\end{align}
which recursively proves that (\ref{eq:assum:r0brR:goal}) is positive or zero and thus \sfklassum~\ref{assum:r0}. \qedhere
\end{itemize}
\end{proof}

\subsubsection{Correctness proof}
This section proves the following theorem:
\begin{theorem}[\mbrCtheo] \label{theo:b-mbrb-correctness}
If \bassum is verified, then Alg.~{\em \ref{alg:b-mbrb}} implements {\em MBRB} with the guarantee $\lmbr = \left\lceil c\left(1-\frac{\tm}{c-2\tb-\tm}\right) \right\rceil$.
\end{theorem}

\noindent The proof follows from the next lemmas.

\begin{lemma} \label{lemma:echo-sufficient}
  $c-\tm \geq \obje.\kld$.
\end{lemma}

\begin{proof}
We want to show that:
\begin{align}
    c - \tm &\geq \left\lfloor \frac{c\tb}{c-\tm-\lfloor \frac{n-\tb}{2} \rfloor} \right\rfloor + 1 = \obje.\kld. \label{eq:proofEk:start}
\end{align}
As the left-hand side is also an integer, we can rewrite (\ref{eq:proofEk:start}) as follows:
\begin{align}
    c - \tm &> \frac{c\tb}{c-\tm-\lfloor \frac{n-\tb}{2} \rfloor}, \\
    (c - \tm)(c-\tm-\lfloor \frac{n-\tb}{2} \rfloor) &> c\tb. \tag{as $(c-\tm-\lfloor \frac{n-\tb}{2} \rfloor)>0$}
\end{align}

We now observe that $\left(\left \lfloor \frac{n+\tb}{2} \right\rfloor\right) = \left(\frac{n+\tb-\epsilon}{2} \right)$ with $\epsilon=0$ if $n+\tb=2k$ is even, and $\epsilon=1$ if $n+\tb=2k+1$ is odd, which leads us to:
\begin{align*}
    (c - \tm)(c-\tm- \frac{n-\tb-\epsilon}{2}) &> c\tb, \\
    (c - \tm)(2c-2\tm- n+\tb+\epsilon ) &> 2c\tb, \\
    (c - \tm)(2c-2\tm- n+\tb+\epsilon ) -2c\tb &> 0.
\end{align*}

Like for the proofs of Lemma~\ref{lemma:ifone:then:enough:internal} and Lemma~\ref{lemma:if:enough:internal:then:enough:ell}, we leverage the fact that the executions that can happen when $c > n-\tb$ can also occur when $c=n - \tb$.
We thus rewrite our inequality for $c=n-\tb$:
\begin{align*}
    (n-\tb - \tm)(n-\tb-2\tm+\epsilon ) - 2(n-\tb)\tb &> 0, \\
    (n-\tb)(n-\tb-2\tm+\epsilon -2\tb ) - \tm(n-\tb-2\tm+\epsilon ) &> 0, \\
    (n - \tb)^2 + (n-\tb)(-2\tm+\epsilon -2\tb ) - \tm(n-\tb-2\tm+\epsilon ) &> 0, \\
    n^2+\tb^2-2n\tb - 2 n\tm +n \epsilon  - 2n\tb + 2 \tb \tm - \tb \epsilon +   2 \tb^{2} -n\tm+\tb\tm+2\tm^2-\epsilon\tm &> 0, \\
    n^2+3\tb^2-4n\tb - 3 n\tm +n \epsilon + 3 \tb \tm - \tb \epsilon +2\tm^2-\epsilon\tm &> 0, \\
    n^2-n (4\tb + 3 \tm - \epsilon) +3\tb^2 + 3 \tb \tm +2\tm^2 - \epsilon (\tb+\tm) &> 0.
\end{align*}

We now solve the second-degree inequality with respect to $n$. It is easy to see that the discriminant is non-negative for non-negative values of $\tb$ and $\tm$. So we obtain: 
\begin{align*}
    & n > 2 \tb + \frac{3 \tm}{2} - \frac{\epsilon}{2} + \frac{\sqrt{4 \tb^{2} + 12 \tb \tm - 4 \tb \epsilon + \tm^{2} - 2 \tm \epsilon + \epsilon^{2}}}{2}, \\
    & - 4 \tb - 3 \tm + \epsilon + 2 n - \sqrt{4 \tb^{2} + 12 \tb \tm - 4 \tb \epsilon + \tm^{2} - 2 \tm \epsilon + \epsilon^{2}} > 0,
\end{align*}
which is implied by the following as $n\geq 3\tb+2\tm+2\sqrt{\tb\tm} + 1$:
\begin{align*}
    & 4 \sqrt{\tb\tm} + 2 \tb + \tm + 2 + \epsilon - \sqrt{4 \tb^{2} + 12 \tb \tm - 4 \tb \epsilon + \tm^{2} - 2 \tm \epsilon + \epsilon^{2}}  > 0, \\
    & 4 \sqrt{\tb\tm} + 2 \tb + \tm + 2 + \epsilon > \sqrt{4 \tb^{2} + 12 \tb \tm - 4 \tb \epsilon + \tm^{2} - 2 \tm \epsilon + \epsilon^{2}}.
\end{align*}

Taking the squares as both the argument of the square root and the left-hand side are non-negative leads to:
\begin{align*}
    & \left(4 \sqrt{\tb\tm} + 2 \tb + \tm + \epsilon + 2\right)^{2} > 4 \tb^{2} + 12 \tb \tm - 4 \tb \epsilon + \tm^{2} - 2 \tm \epsilon + \epsilon^{2}, \\
    & 16 \tb^{\frac{3}{2}} \sqrt{\tm} + 8 \sqrt{\tb} \tm^{\frac{3}{2}} + 8 \sqrt{\tb} \sqrt{\tm} \epsilon + 16 \sqrt{\tb} \sqrt{\tm} + 4 \tb^{2} + 20 \tb \tm + 4 \tb \epsilon + 8 \tb + \tm^{2} + 2 \tm \epsilon \\
    & ~~~~ + 4 \tm + \epsilon^{2} + 4 \epsilon + 4 > 4 \tb^{2} + 12 \tb \tm - 4 \tb \epsilon + \tm^{2} - 2 \tm \epsilon + \epsilon^{2},
\end{align*}
which simplifies to:
\begin{align}
    & 16 \tb^{\frac{3}{2}} \sqrt{\tm} + 8 \sqrt{\tb} \tm^{\frac{3}{2}} + 8 \sqrt{\tb} \sqrt{\tm} \epsilon + 16 \sqrt{\tb} \sqrt{\tm} + 8 \tb \tm + 8 \tb \epsilon + 8 \tb + 4 \tm \epsilon \nonumber\\
    & ~~~~ + 4 \tm + 4 \epsilon + 4 > 0. \label{eq:proofek:final}
\end{align}

We can then easily  observe that the left-hand side of (\ref{eq:proofek:final}) is strictly positive, thereby proving all previous inequalities and thus the lemma. 
\end{proof}

\begin{lemma}\label{lemma:ready-sufficient}
    $\obje.\lgd \geq \objr.\kld$.
\end{lemma}

\begin{proof}
We need to prove: 
\begin{align}
    \obje.\lgd &= \left\lceil c\left(1-\frac{\tm}{c-\lfloor \frac{n+\tb}{2} \rfloor}\right) \right\rceil \geq \left\lfloor \frac{c\tb}{c-2\tm-\tb} \right\rfloor+1 = \objr.\kld. \label{eq:proofel:start}
\end{align}

We observe that $x \geq \lfloor m \rfloor + 1$ if and only if $x > m$, and that $m \geq \lfloor m \rfloor $. Therefore (\ref{eq:proofel:start}) is implied by the following: 
\begin{align*}
    c \left(1 - \frac{\tm}{c - \frac{n +\tb -\epsilon}{2}} \right) & \geq \frac{c\tb}{c - \tb - 2 \tm }, \\
    c - \frac{2 \tm c}{ 2 c - \tb -n + \epsilon}  & > \frac{c \tb }{c - \tb - 2 \tm }.
\end{align*}
As both denominators are positive, we can solve: 
\begin{align*}
    & - \tb \left(- \tb + 2 c + \epsilon - n\right) - 2 \tm \left(- \tb - 2 \tm + c\right) + \left(- \tb - 2 \tm + c\right) \left(- \tb + 2 c + \epsilon - n\right) > 0, \\
    & - \tb \left(- \tb + 2 c + \epsilon - n\right) + \left(- \tb - 2 \tm + c\right) \left(- \tb - 2 \tm + 2 c + \epsilon - n\right) > 0, \\
    & - \tb \left(- 2 \tb + c + \epsilon\right) + \left(- \tb - 2 \tm + c\right) \left(- 2 \tb - 2 \tm + c + \epsilon\right) > 0, \tag{as $c\geq n - \tb$} \\
    & - \tb \left(- 2 \tb + c + \epsilon\right) + \left(- \tb - 2 \tm + c\right) \left(- 2 \tb - 2 \tm + c + \epsilon\right) > 0, \\ 
    & - \tb \left(- \tb + 2 c - n\right) + \epsilon \left(- 3 \tb - 2 \tm + c\right) + \left(- \tb - 2 \tm + c\right)^{2} > 0, \\
    & \tb^{2} - 2 \tb c + \tb n + \epsilon \left(- 3 \tb - 2 \tm + c\right) + \left(- \tb - 2 \tm + c\right)^{2} > 0, \\
    & \tb^{2} - 2 \tb c + \tb \left(2 \sqrt{\tb} \sqrt{\tm} + 3 \tb + 2 \tm + 1\right) + \epsilon \left(- 3 \tb - 2 \tm + c\right) + \left(- \tb - 2 \tm + c\right)^{2} > 0, \tag{as $n \geq  3 \tb + 2 \tm +2 \sqrt{\tb\tm}$}\\
    & 2 \tb^{\frac{3}{2}} \sqrt{\tm} + 4 \tb^{2} + 2 \tb \tm - 2 \tb c + \tb + \epsilon \left(- 3 \tb - 2 \tm + c\right) + \left(- \tb - 2 \tm + c\right)^{2} > 0, \\
    & 2 \tb^{\frac{3}{2}} \sqrt{\tm} + 5 \tb^{2} + 6 \tb \tm - 4 \tb c + \tb + 4 \tm^{2} - 4 \tm c + c^{2} + \epsilon \left(- 3 \tb - 2 \tm + c\right) > 0.
\end{align*}

We now consider the two possible values of $\epsilon$:
\begin{itemize}
\item $\epsilon = 0$:
\begin{align}
    & 2 \tb^{\frac{3}{2}} \sqrt{\tm} + 5 \tb^{2} + 6 \tb \tm - 4 \tb c + \tb + 4 \tm^{2} - 4 \tm c + c^{2} > 0\label{eq:proofel:tosolve}
\end{align}

We solve the inequality with respect to $c$ to obtain (when the discriminant is positive): 
\begin{align}
c &> 2 \tb + 2 \tm + \sqrt{- 2 \tb^{\frac{3}{2}} \sqrt{\tm} - \tb^{2} + 2 \tb \tm - \tb}\nonumber
\end{align}
which we prove by observing that $c \geq n-\tb \geq 2\tb +2\tm+2\sqrt{\tb\tm}+1$ and that: 
\begin{align*}
    & 2\tb +2\tm+2\sqrt{\tb\tm}+1 > 2 \tb + 2 \tm + \sqrt{- 2 \tb^{\frac{3}{2}} \sqrt{\tm} - \tb^{2} + 2 \tb \tm - \tb},
\end{align*}
as all terms except $2\tb\tm$ inside the square root are negative. 
When the discriminant is negative (e.g. for $\tm=0$),  inequality (\ref{eq:proofel:tosolve}) is satisfied for all values of $c$. 

\item $\epsilon = 1$:

In this case, we obtain:
\begin{align}
& 2 \tb^{\frac{3}{2}} \sqrt{\tm} + 5 \tb^{2} + 6 \tb \tm - 4 \tb c - 2 \tb + 4 \tm^{2} - 4 \tm c - 2 \tm + c^{2} + c > 0, \nonumber
\end{align}
which is implied by  a negative discriminant or by:
\begin{align*}
    & c > 2 \tb + 2 \tm + \sqrt{- 2 \tb^{\frac{3}{2}} \sqrt{\tm} -  \tb^{2} + 2 \tb \tm + 1/4} - \frac{1}{2}.
\end{align*}

Like before, we simply observe that:
\begin{align*}
    2\sqrt{\tb\tm}&  \geq \sqrt{2 \tb \tm} +\frac{1}{2} - \frac{1}{2}, \\
    & \geq \sqrt{2 \tb \tm + 1/4} - \frac{1}{2}, \\
    & \geq \sqrt{- 2 \tb^{\frac{3}{2}} \sqrt{\tm} -  \tb^{2} + 2 \tb \tm + 1/4} - \frac{1}{2},
\end{align*}
thereby proving the second case and the lemma. \qedhere
\end{itemize}
\end{proof}

\begin{lemma}[\mbrVprop] \label{lemma:b-mbrb-validity}
If a correct process $p_i$ mbrb-delivers an \app $m$ from a correct process $p_j$ with sequence number~\sn, then $p_j$ mbrb-broadcast $m$ with sequence number~\sn.
\end{lemma}

\begin{proof}
If $p_i$ mbrb-delivers $(m,\sn,j)$ at line~\ref{BMBRB-mbrb}, then it \kl-delivered $(\readym(m),(\sn,j))$ using \objr. From \klVprop, and as $\objr.\kv = 1$, we can assert that at least one correct process $p_x$ \kl-cast $(\readym(m),(\sn,j))$ at line~\ref{BMBRB-klc-ready}, after having \kl-delivered $(\echom(m),(sn,j))$ using \obje. Again, from \klVprop, we can assert that at least $\obje.\kv = 1$ correct process $p_y$ \kl-cast $(\echom(m),(\sn,j))$ at line~\ref{BMBRB-klc-echo}, after having received an $\initm(m,\sn)$ \imp from $p_j$. And as $p_j$ is correct and the network channels are authenticated, then $p_j$ has ur-broadcast $\initm(m,\sn)$ at line~\ref{BMBRB-mbrb}, during a $\mbrbroadcast(m,\sn)$ invocation.
\end{proof}

\begin{lemma}[\mbrNDNprop] \label{lemma:b-mbrb-no-duplication}
A correct process $p_i$ mbrb-delivers at most one \app from a process $p_j$ with sequence number~\sn.
\end{lemma}

\begin{proof}
By \klNDNprop, we know that a correct process $p_i$ can \kl-deliver at most one $\readym(-)$ with identity $(\sn,j)$. Therefore, $p_i$ can mbrb-deliver only one \app from $p_j$ with sequence number~\sn. 
\end{proof}

\begin{lemma}[\mbrNDYprop] \label{lemma:b-mbrb-no-duplicity}
No two different correct processes mbrb-deliver different \apps from a process $p_i$ with the same sequence number~\sn.
\end{lemma}

\begin{proof}
We proceed by contradiction. Let us consider two correct processes $p_w$ and $p_x$ that respectively mbrb-deliver~$(m,\sn,i)$ and $(m',\sn,i)$ at line~\ref{BMBRB-mbrb-dlv}, such that $m \neq m'$. It follows that $p_w$ and $p_x$ respectively \kl-delivered $(\readym(m),(\sn,i))$ and $(\readym(m'),(\sn,i))$ using \objr.

From \klVprop, and as $\objr.\kv \geq 1$, we can assert that two correct processes $p_y$ and $p_z$ respectively \kl-cast $(\readym(m),(\sn,i))$ and $(\readym(m'),(\sn,i))$ at line~\ref{BMBRB-klc-ready}, after having respectively \kl-delivered $(\echom(m),(\sn,i))$ and $(\echom(m'),(\sn,i))$ using \obje. But as $\obje.\nodpty = \ttrue$, then, by \klNDYprop, we know that $m = m'$. There is a contradiction.
\end{proof}

\begin{lemma}[\mbrLDprop] \label{lemma:b-mbrb-local-delivery}
If a correct process $p_i$ mbrb-broadcasts an \app $m$ with sequence number~\sn, then at least one correct process $p_j$ eventually mbrb-delivers $m$ from $p_i$ with sequence number~\sn.
\end{lemma}

\begin{proof}
If $p_i$ mbrb-broadcasts $(m,\sn)$ at line~\ref{BMBRB-mbrb}, then it invokes ur-broadcasts $\initm(m,\sn)$.
By the definition of the MA, the \imp $\initm(m,\sn)$ is then received by at least $c-\tm$ correct processes at line~\ref{BMBRB-klc-echo}, which then \kl-cast $(\echom(m),\sn,i)$.
As $p_i$ is correct and ur-broadcasts only one \imp $\initm(-,\sn)$, then no correct process \kl-casts any different $(\echom(-),\sn,i)$.
Moreover, thanks to Lemma~\ref{lemma:echo-sufficient}, we know that:
\begin{align*}
    c-\tm \geq \obje.\kld = \left\lfloor \frac{c\tb}{c-\tm-\left\lfloor \frac{n-\tb}{2} \right\rfloor} \right\rfloor + 1.
\end{align*}

\sloppy Hence, from \klLDprop and \klSGDprop, at least $\obje.\lgd = \left\lceil c\left(1-\frac{\tm}{c-\lfloor \frac{n+\tb}{2} \rfloor}\right) \right\rceil$ correct processes eventually \kl-deliver $(\echom(m),(\sn,i))$ using \obje and then \kl-cast $(\readym(m),(\sn,i))$ using \objr at line~\ref{BMBRB-klc-ready}.
By \klVprop, and as $\objr.\kv \geq 1$, then no correct process can \kl-cast a different $(\readym(-),(\sn,i))$, because otherwise it would mean that at least one correct process would have \kl-cast a different $(\echom(-),(\sn,i))$, which is impossible (see before).
Moreover, thanks to Lemma~\ref{lemma:ready-sufficient}, we know that:
\begin{align*}
    \left\lceil c\left(1-\frac{\tm}{c-\lfloor \frac{n+\tb}{2} \rfloor}\right) \right\rceil = \obje.\lgd \geq \objr.\kld = \left\lfloor \frac{c\tb}{c-2\tm-\tb} \right\rfloor + 1.
\end{align*}

\sloppy Therefore, \klLDprop applies and we know that at least one correct processes eventually \kl-delivers~$(\readym(m),(\sn,i))$ using \objr and then mbrb-delivers~$(m,\sn,i)$ at line~\ref{BMBRB-mbrb-dlv}.
\end{proof}

\begin{lemma}[\mbrGDprop] \label{lemma:b-mbrb-global-delivery}
\sloppy
If a correct process $p_i$ mbrb-delivers an \app $m$ from a process $p_j$ with sequence number~\sn, then at least \linebreak
$\lmbr = \left\lceil c\left(1-\frac{\tm}{c-2\tb-\tm}\right) \right\rceil$ correct processes mbrb-deliver $m$ from $p_j$ with sequence number~\sn.
\end{lemma}

\begin{proof}
If $p_i$ mbrb-delivers $(m,\sn,j)$ at line~\ref{BMBRB-mbrb-dlv}, then it has \kl-delivered $(\readym(m),(\sn,j))$ using \objr. From \klVprop, we know that at least $\objr.\kv \geq 1$ correct process \kl-cast $(\readym(m),(\sn,j))$ using \objr at line~\ref{BMBRB-klc-ready} and thus \kl-delivered $(\echom(m),(\sn,j))$ using \obje.
From \klNDYprop, and as $\obje.\nodpty = \ttrue$, we can state that no correct process \kl-delivers any $(\echom(m'),(\sn,j))$ where $m' \neq m$ using \obje, so no correct process \kl-casts any $(\readym(m'),(\sn,j))$ where $m' \neq m$ using \objr at line~\ref{BMBRB-klc-ready}.
It means that \klSGDprop applies, and we can assert that at least $\objr.\lgd = \left\lceil c\left(1-\frac{\tm}{c-2\tb-\tm}\right) \right\rceil = \lmbr$ correct processes eventually \kl-deliver~$(\readym(m),(\sn,j))$ 
using \objr and thus mbrb-deliver $(m,\sn,j)$ at line~\ref{BMBRB-mbrb-dlv}.
\end{proof}

\subsection{Proof of MBRB with Imbs and Raynal's reconstructed algorithm (Algorithm~\ref{alg:ir-mbrb})} \label{sec-proof-ir-mbrb}

\subsubsection{Instantiating the parameters of the \texorpdfstring{\kl}{k2l}-cast object}
\label{sec:inst-param-klcast-ibr}
In Alg.~\ref{alg:ir-mbrb} (page~\pageref{alg:ir-mbrb}), we instantiate the \kl-cast object \objw using the signature-free implementation presented in Section \ref{sec-klcast-sf} with parameters $\qd=\irqdval$, $\qf=\irqfval$, and $\single=\ffalse$.
Based on Theorem~\ref{theo:sf-kl-correctness} (page~\pageref{theo:sf-kl-correctness}), these parameters lead to the following values for  \kv, \kld, \lgd and \nodpty.
\begin{itemize}
    \item $\begin{aligned}[t]
    \objw.\kv &= \objw.\qf-n+c
    = \irqfval -n+c \\
    &\geq \irqfval -n+ n-\tb = \left\lfloor \frac{n-\tb}{2} \right\rfloor+1,
    \end{aligned}$
    
    \item $\begin{aligned}[t]
    \objw.\kld
    &= \kval[\objw.] \\
    &= \left\lfloor \frac{c(\irqfval -1)}{c-\tm-(\irqdval) + \irqfval} \right\rfloor + 1 \\
    &= \kir,
    \end{aligned}$
    
    \item $\begin{aligned}[t]
    \objw.\lgd
    &= \lval[\objw.]
    = \left\lceil c\left(1-\frac{\tm}{c- (\irqdval) +1}\right) \right\rceil \\
    &= \lir,
    \end{aligned}$
    
    \item $\begin{aligned}[t]
    \objw.\nodpty
    &= \left(\left(\objw.\qf > \frac{n+\tb}{2} \right) \lor \left(\objw.\single \land \objw.\qd > \frac{n+\tb}{2} \right)\right) \\
    &= \left(\left( \irqfval > \frac{n+\tb}{2} \right) \lor \left(\ffalse \land \irqdval > \frac{n+\tb}{2} \right)\right) \\
    %
    &= (\ttrue \lor (\ffalse \land \ttrue)) = \ttrue.
  \end{aligned}$
\end{itemize}
%

Finally, we observe that for Alg.~\ref{alg:ir-mbrb}, \sfklassum~\ref{assum:base} through~\ref{assum:r0} are all satisfied by \irassum ($n > \irBound$), as we prove in Appendix~\ref{sec:proof-satisf-object-W}. 

\subsubsection{Proof of satisfaction of the assumptions of Algorithm~\ref{alg:sf-klcast}}
\label{sec:proof-satisf-object-W}
This section proves that all the assumptions of the signature-free \kl-cast implementation presented in Alg.~\ref{alg:sf-klcast} (page~\pageref{alg:sf-klcast}) are well respected for the \kl-cast instance used in Alg.~\ref{alg:ir-mbrb} (\objw).

\begin{lemma}
Alg.~{\em \ref{alg:sf-klcast}}'s \sfklassums are well respected for \objw.
\end{lemma}

\begin{proof}
Let us recall that $\qf =\irqfval$
and $\qd =\irqdval$ for object \objw.

\begin{itemize}
\item \textit{Proof of satisfaction of \sfklassum~{\em\ref{assum:base}}} ($c-\tm \geq \objw.\qd \geq \objw.\qf+\tb \geq 2\tb+1$):

From \irassum ($n>\irBound$), we get that $n > 5\tb+8\tm$, which yields:
\begin{align}
    c-\tm &\geq n-\tb-\tm = \frac{2n-2\tb-2\tm}{2}, \tag{by definition of $c$}\\
    &> \frac{n+5\tb+8\tm-2\tb-2\tm}{2} = \frac{n+3\tb}{2}, \tag{as $n > 5\tb+8\tm$}\\
    &\geq \left\lfloor \frac{n+3\tb+6\tm}{2} \right\rfloor+1 = \left\lfloor \frac{n+3\tb}{2} \right\rfloor+3\tm+1. \label{eq:satisf-objw-1-1}
\intertext{We also have:}
    \left\lfloor \frac{n+3\tb}{2} \right\rfloor+1 &> \left\lfloor \frac{5\tb+8\tm+3\tb}{2} \right\rfloor+1 = 4\tb+4\tm+1, \tag{as $n > 5\tb+8\tm$}\\
    &\geq 2\tb+1. \label{eq:satisf-objw-1-2}\\
\intertext{By combining (\ref{eq:satisf-objw-1-1}) and (\ref{eq:satisf-objw-1-2}), we obtain:}
    c-\tm &\geq \left\lfloor \frac{n+3\tb}{2} \right\rfloor+3\tm+1 \geq \left\lfloor \frac{n+3\tb}{2} \right\rfloor+1 \geq 2\tb+1, \nonumber\\
    c-\tm &\geq \objw.\qd \geq \objw.\qf+\tb \geq 2\tb+1. \nonumber
\end{align}

\item \textit{Proof of satisfaction of \sfklassum~{\em\ref{assum:disc}}} ($\alpha^2-4(\objw.\qf-1)(n-\tb) \geq 0$):

Let us recall that for object \objw we have $\qf=\irqfval$ and $\qd=\irqdval$.
We therefore have $\alpha= \left\lfloor \frac{3n-\tb}{2}\right\rfloor -\tm $.
Let us now consider the following quantity:
\begin{align}
    \Delta &= \alpha^{2} - 4 (\qf - 1)(n - \tb), \nonumber\\
    & = \left( \left\lfloor \frac{3n-\tb}{2} \right\rfloor -\tm \right)^2 - 4 \left\lfloor \frac{n+\tb}{2} \right\rfloor (n-\tb). \label{eq:irD-beforecases}
\end{align}

We now observe that $\left(\left \lfloor \frac{m}{2} \right\rfloor\right) = \left(\frac{m-\epsilon}{2} \right)$ with $\epsilon=0$ if $m=2k$ is even, and $\epsilon=1$ if $m=2k+1$ is odd.
We thus rewrite (\ref{eq:irD-beforecases}) as follows:
\begin{align*}
    & ~~~~ \left( \frac{3n-\tb-\epsilon}{2} -\tm \right)^2 - 4  \frac{n+\tb-\epsilon}{2}  (n-\tb), \\
    &= \left(  \frac{3n-\tb-\epsilon-2\tm}{2}\right)^2 - 4  \frac{n+\tb-\epsilon}{2}  (n-\tb), \\
    &= \frac{\tb^{2} + 4 \tb \tm + 2 \tb \epsilon - 6 \tb n + 4 \tm^{2} + 4 \tm \epsilon - 12 \tm n + \epsilon^{2} - 6 \epsilon n + 9 n^2}{4} \\
    & ~~~~ +  \frac{8\tb^2 -8\tb \epsilon +8 \epsilon n -8 n^2}{4}, \\
    &= \frac{9 \tb^{2} + 4 \tb \tm - 6 \tb \epsilon - 6 \tb n + 4 \tm^{2} + 4 \tm \epsilon - 12 \tm n + \epsilon^{2} + 2 \epsilon n + n^{2}}{4}, \\
    &= \frac{9 \tb^{2} - 6 \tb n + n^2 + 4\tb\tm - 12 \tm n + 4 \tm^{2} + 4 \tm \epsilon - 6 \tb \epsilon + \epsilon^{2} + 2 \epsilon n }{4}, \\
    &= \frac{(n-3\tb)^2 + 4 \tm (\tb   - 3 n +  \tm) + \epsilon(4 \tm - 6 \tb + \epsilon + 2  n)}{4}.
\end{align*}

We now multiply by $4$ and solve the inequality: 
\begin{align}
    &  n^{2}  - 6n( \tb+ 2 \tm) +9 \tb^{2} + 4 \tb \tm + 4 \tm^{2}  + \epsilon \left(- 6 \tb + 4 \tm + \epsilon + 2 n\right) \geq 0, \nonumber\\
    & n \geq 3 \tb + 4 \sqrt{\tm} \sqrt{2 \tb + 2 \tm - \epsilon} + 6 \tm - \epsilon. \label{eq:intermediate}
\end{align}

By \irassum we have $n > \irBoundN$.
To prove (\ref{eq:intermediate}), we therefore show that  $\irBound \geq 3 \tb + 4 \sqrt{\tm} \sqrt{2 \tb + 2 \tm } + 6 \tm$:
\begin{align}
    & \irBound  \geq 3 \tb + 4 \sqrt{\tm} \sqrt{2 \tb + 2 \tm } + 6 \tm \nonumber\\
    \iff& 2\tb+6\tm+\frac{2\tb\tm}{\tb+2\tm}  \geq 4 \sqrt{\tm} \sqrt{2 \tb + 2 \tm} \nonumber\\
    \iff& \left(2\tb+6\tm+\frac{2\tb\tm}{\tb+2\tm}\right)^2  \geq 16 \tm(2 \tb + 2 \tm ) \nonumber\\
    \iff& - 16 \tm \left(\tb + 2 \tm\right) \left(2 \tb + 2 \tm\right) + \left(2 \tb \tm + 2 \tb \left(\tb + 2 \tm\right) + 6 \tm \left(\tb + 2 \tm\right)\right)^{2} \geq 0 \nonumber\\
    \iff& 4 \tb^{4} + 48 \tb^{3} \tm + 192 \tb^{2} \tm^{2} - 32 \tb^{2} \tm + 288 \tb \tm^{3} - 96 \tb \tm^{2} + 144 \tm^{4} - 64 \tm^{3} \geq 0. \label{eq:proof:ir:disc:final}
\end{align}

We observe that (\ref{eq:proof:ir:disc:final}) holds as $144 \tm^{4} \geq 64 \tm^{3}$,  $288 \tb \tm^{3} \geq 96 \tb \tm^{2}$, and $192 \tb^{2} \tm^{2} \geq 32 \tb^{2} \tm$, therefore proving \sfklassum~\ref{assum:disc}.

\item \textit{Proof of satisfaction of \sfklassum~{\em\ref{assum:r1}}} ($\alpha(\objw.\qd-1)-(\objw.\qf-1)(n-\tb)-(\objw.\qd-1)^2 > 0$):

Let us consider the quantity on the left-hand side of \sfklassum~\ref{assum:r1} and substitute 
$\qf=\irqfval$, $\qd= \irqdval$, and $\alpha=\left\lfloor \frac{3n-\tb}{2}\right\rfloor -\tm$:
\begin{align*}
    & ~~~~ \alpha (\qd - 1) - (\qf-1) (n-\tb) - (\qd - 1)^2, \\
    & = \left(\left\lfloor \frac{3n-\tb}{2}\right\rfloor -\tm\right)\left(\left\lfloor \frac{ n+3\tb}{2}\right\rfloor +3\tm\right) - \left( \left\lfloor \frac{n+\tb}{2}\right\rfloor\right) (n-\tb) \\
    & ~~~~- \left(\left\lfloor \frac{n+3\tb}{2}\right\rfloor +3\tm\right)^2.
\end{align*}

We now observe that $\left(\left \lfloor \frac{m}{2} \right\rfloor\right) = \left(\frac{m-\epsilon}{2} \right)$ with $\epsilon=0$ if $m=2k$ is even, and $\epsilon=1$ if $m=2k+1$ is odd, and rewrite the expression accordingly:
\begin{align*}
    & ~~~~ \frac{3n-\tb-2\tm-\epsilon}{2}\cdot \frac{n+3\tb+6\tm-\epsilon}{2} - \frac{(n+\tb-\epsilon)(n-\tb)}{2} \\
    & ~~~~ - \left(\frac{n+3\tb+6\tm-\epsilon}{2}\right)^{2}, \\
    &=  \frac{(n+3\tb+6\tm-\epsilon)(3n-\tb-2\tm-\epsilon-n-3\tb-6\tm+\epsilon)}{4}  - \frac{(n+\tb-\epsilon)(n-\tb)}{2}, \\
    &= \frac{(n+3\tb+6\tm-\epsilon)(2n-4\tb-8\tm)}{4} - \frac{(n+\tb-\epsilon)(n-\tb)}{2}, \\
    &=  \frac{- 12 \tb^{2} - 48 \tb \tm + 4 \tb \epsilon + 2 \tb n - 48 \tm^{2} + 8 \tm \epsilon + 4 \tm n - 2 \epsilon n + 2 n^{2} + 2 \tb^{2} - 2 \tb \epsilon + 2 \epsilon n - 2 n^{2}}{4}, \\
    &= \frac{- 10 \tb^{2} - 48 \tb \tm + 2 \tb \epsilon + 2 \tb n - 48 \tm^{2} + 8 \tm \epsilon + 4 \tm n}{4}.
\end{align*}

As the coefficients of $n$ are all positive, we can lower-bound the quantity using $n>\irBound$:
\begin{align*}
    & ~~~~ \frac{- 10 \tb^{2} - 48 \tb \tm - 48 \tm^{2} +  2n(\tb + 2 \tm)  + 2 \epsilon(\tb  + 8 \tm) }{4}, \\
    &= \frac{- 10 \tb^{2} - 48 \tb \tm - 48 \tm^{2} +  2(5\tb+12\tm+\frac{2\tb\tm}{\tb+2\tm}) (\tb + 2 \tm)  + 2  \epsilon(\tb  + 8 \tm) }{4}, \\
    &= \frac{- 10 \tb^{2} - 48 \tb \tm - 48 \tm^{2} + 10 \tb^{2} + 44 \tb \tm + 48 \tm^{2} + 4\tb\tm + 2  \epsilon(\tb  + 8 \tm) }{4}, \\
    & = \frac{ \epsilon(\tb  + 8 \tm) }{2} \geq 0,
\end{align*}
which proves all previous inequalities and thus \sfklassum~\ref{assum:r1}.

\item \textit{Proof of satisfaction of \sfklassum~{\em\ref{assum:r0}}} ($\alpha(\objw.\qd-1-\tb)-(\objw.\qf-1)(n-\tb)-(\objw.\qd-1-\tb)^2 \geq 0$):

Let us consider the quantity on the left-hand side of \sfklassum~\ref{assum:r0} and substitute $\qf=\irqfval$, $\qd= \irqdval$, and $\alpha=\left\lfloor \frac{3n-\tb}{2}\right\rfloor -\tm$:
\begin{align*}
    & ~~~~ \alpha(\qd -1 -\tb) - (\qf -1)(n - \tb)  - (\qd -1 - \tb)^2, \\             
    & = \left(\left\lfloor \frac{3n-\tb}{2}\right\rfloor -\tm\right)\left(\left\lfloor \frac{ n+3\tb}{2}\right\rfloor +3\tm- \tb \right) - \left( \left\lfloor \frac{n+\tb}{2}\right\rfloor\right) (n-\tb) \\
    & ~~~~ - \left(\left\lfloor \frac{ n+3\tb}{2}\right\rfloor +3\tm- \tb \right)^2.
\end{align*}

We now observe that $\left(\left \lfloor \frac{m}{2} \right\rfloor\right) = \left(\frac{m-\epsilon}{2} \right)$ with $\epsilon=0$ if $m=2k$ is even, and $\epsilon=1$ if $m=2k+1$ is odd, and rewrite the expression accordingly:
\begin{align*}
    & = \left( \frac{3n-\tb-\epsilon}{2} -\tm\right)\left( \frac{ n+3\tb-\epsilon}{2} +3\tm- \tb \right) - \left( \frac{n+\tb-\epsilon}{2}\right) (n-\tb) \\
    & ~~~~ - \left( \frac{ n+3\tb-\epsilon}{2} +3\tm- \tb \right)^2, \\
    & = \left( \frac{3n-\tb-2\tm-\epsilon}{2} \right)\left( \frac{ n+\tb+6\tm-\epsilon}{2}  \right) - \left( \frac{n+\tb-\epsilon}{2}\right) (n-\tb) \\
    & ~~~~ - \left( \frac{ n+\tb+6\tm-\epsilon}{2} \right)^2, \\
    & = \frac{ (n+\tb+6\tm-\epsilon)(3n-\tb-2\tm-\epsilon-n-\tb-6\tm+\epsilon)}{4} - \left( \frac{(n+\tb-\epsilon)(n-\tb)}{2}\right), \\
    & = \frac{ (n+\tb+6\tm-\epsilon)(2n-2\tb-8\tm)}{4} - \left( \frac{(n+\tb-\epsilon)(n-\tb)}{2}\right), \\
    & = \frac{ (n+\tb+6\tm-\epsilon)(n-\tb-4\tm)-(n+\tb-\epsilon)(n-\tb)}{2}, \\
    & = \frac{- 10 \tb \tm - 24 \tm^{2} + 4 \tm \epsilon + 2 \tm n}{2}.
\end{align*}

As the coefficients of $n$ are all positive, we can lower bound using $n>\irBound>5\tb+12\tm$ to obtain:
\begin{align*}
    & = \frac{- 10 \tb \tm - 24 \tm^{2} + 4 \tm \epsilon + 2 \tm (5\tb+12\tm)}{2}, \\
    & = \frac{- 10 \tb \tm - 24 \tm^{2} + 4 \tm \epsilon + 10 \tb\tm +24\tm^2}{2}, \\
    & = 2 \tm \epsilon \geq 0,
\end{align*}
which proves \sfklassum~\ref{assum:r0}. \qedhere
\end{itemize}
\end{proof}

\subsubsection{Correctness proof}

This section proves the following theorem:
\begin{theorem}[\mbrCtheo] \label{theo:ir-mbrb-correctness}
If \irassum is verified, then Alg.~{\em \ref{alg:ir-mbrb}} implements {\em MBRB} with the guarantee $\lmbr = \lir$.
\end{theorem}

\noindent The proof follows from the next lemmas.

\begin{lemma} \label{lemma:witness-sufficient}
  $c-\tm \geq \objw.\kld$.
\end{lemma}

\begin{proof}
This proof is presented in reverse order: we start with the result we want to prove and finish with a proposition we know to be true.
In this manner, given two consecutive propositions, we only need that the latter implies the former and not necessarily the converse.
\todo{repeat this for every proof of the same kind in the appendices?}
We want to show that:
\begin{align*}
    c-\tm &\geq \kir = \objw.\kld, \\
    c-\tm &> \frac{c\lfloor\frac{n+\tb}{2}\rfloor}{c-\tb-4\tm}, \tag{as $x \geq \lfloor y \rfloor+1 \iff x>y$}\\
    c-\tm &> \frac{c\frac{n+\tb}{2}}{c-\tb-4\tm}, \\
    c-\tm &> \frac{c(n+\tb)}{2 (c-\tb-4\tm)}, \\
    c-\tm &> \frac{c(n+\tb)}{2c-2\tb-8\tm}, \\
    (c-\tm)(2c-2\tb-8\tm) &> c(n+\tb), \tag{as $2c-2\tb-8\tm > 0$ by \irassum}\\
    (c-\tm)(2c-2\tb-8\tm) &> c(c-2\tb) \geq c(n+\tb), \tag{as $n \leq c+\tb$}\\
    (c-\tm)(2c-2\tb-8\tm)-c(c-2\tb) &> 0, \\
    c^{2} + 2\tb\tm -4\tb c + 8\tm^{2} - 10\tm c &> 0, \\
    2 \tb \tm + 8 \tm^{2} + c^{2} + c \left(- 4 \tb - 10 \tm\right) &> 0.
\end{align*}

The left-hand side of the above inequality is a second-degree polynomial, whose roots we can solve:
\begin{align*}
    \left[ 2 \tb + 5 \tm - \sqrt{4 \tb^{2} + 18 \tb \tm + 17 \tm^{2}}, 2 \tb + 5 \tm + \sqrt{4 \tb^{2} + 18 \tb \tm + 17 \tm^{2}}\right].
\end{align*}
We now need to show that:
\begin{align*}
    c > 2 \tb + 5 \tm + \sqrt{4 \tb^{2} + 18 \tb \tm + 17 \tm^{2}}.
\end{align*}
By \irassum, we know that:
\begin{align*}
    n \geq 5 \tb + 12 \tm + \frac{2 \tb \tm}{\tb + 2 \tm} + 1,
\end{align*}
and thus that:
\begin{align*}
    n &\geq 5 \tb + 12 \tm + 1, \\
    c &\geq 4 \tb + 12 \tm + 1.
\end{align*}
So we want to show that:
\begin{align*}
    4 \tb + 12 \tm + 1 &> 2 \tb + 5 \tm + \sqrt{4 \tb^{2} + 18 \tb \tm + 17 \tm^{2}}, \\
    2 \tb + 7 \tm + 1 &> \sqrt{4 \tb^{2} + 18 \tb \tm + 17 \tm^{2}}.
\end{align*}
It is easy to see that the right-hand side of the above inequality is non-negative, so we get:
\begin{align*}
    \left(2 \tb + 7 \tm + 1\right)^{2} &> 4 \tb^{2} + 18 \tb \tm + 17 \tm^{2}, \\
    4 \tb^{2} + 28 \tb \tm + 4 \tb + 49 \tm^{2} + 14 \tm + 1 &> 4 \tb^{2} + 18 \tb \tm + 17 \tm^{2}, \\
    10 \tb \tm + 4 \tb + 32 \tm^{2} + 14 \tm + 1 &> 0.
\end{align*}
This concludes the proof.
\end{proof}

\begin{lemma}[\mbrVprop] \label{lemma:ir-mbrb-validity}
If a correct process $p_i$ mbrb-delivers an \app $m$ from a correct process $p_j$ with sequence number \sn, then $p_j$ mbrb-broadcast $m$ with sequence number \sn.
\end{lemma}

\begin{proof}
If $p_i$ mbrb-delivers $(m,\sn,j)$ at line~\ref{IRMBRB-mbrb}, then it \kl-delivered $(\witnessm(m),(\sn,j))$ using \objw.
From \klVprop, and as $\objw.\kv \geq 1$, we can assert that at least one correct process $p_i'$ \kl-cast $(\witnessm(m),(\sn,j))$ at line~\ref{IRMBRB-klcast-witness}, after having received an $\initm(m,\sn)$ \imp from $p_j$.
And as $p_j$ is correct and the network channels are authenticated, then $p_j$ has ur-broadcast $\initm(m,\sn)$ at line~\ref{IRMBRB-mbrb}, during a $\mbrbroadcast(m,\sn)$ invocation.
\end{proof}

\begin{lemma}[\mbrNDNprop] \label{lemma:ir-mbrb-no-duplication}
A correct process $p_i$ mbrb-delivers at most one \app from a process $p_j$ with sequence number \sn.
\end{lemma}

\begin{proof}
By \klNDNprop, we know that a correct process $p_i$ can \kl-deliver at most one $\readym(-)$ with identity $(\sn,j)$. Therefore, $p_i$ can mbrb-deliver only one \app from $p_j$ with sequence number \sn. 
\end{proof}

\begin{lemma}[\mbrNDYprop] \label{lemma:ir-mbrb-no-duplicity}
No two distinct correct processes mbrb-deliver different \apps from a process $p_i$ with the same sequence number \sn.
\end{lemma}

\begin{proof}
As $\objw.\nodpty = \ttrue$, then, by \klNDYprop, we know that no two correct processes can \kl-deliver two different \apps with the same identity using \objw at line~\ref{IRMBRB-mbrb-dlv}.
Hence, no two correct processes mbrb-deliver different \apps for a given sequence number \sn and sender $p_i$.
\end{proof}

\begin{lemma}[\mbrLDprop] \label{lemma:ir-mbrb-local-delivery}
If a correct process $p_i$ mbrb-broadcasts an \app $m$ with sequence number \sn, then at least one correct process $p_j$ eventually mbrb-delivers $m$ from $p_i$ with sequence number~\sn.
\end{lemma}

\begin{proof}
If $p_i$ mbrb-broadcasts $(m,\sn)$ at line~\ref{IRMBRB-mbrb}, then it invokes ur-broadcasts $\initm(m,\sn)$.
By the definition of the MA, the \imp $\initm(m,\sn)$ is then received by at least $c-\tm$ correct processes at line~\ref{IRMBRB-klcast-witness}, which then \kl-cast $(\witnessm(m),\sn,i)$.
But thanks to Lemma~\ref{lemma:witness-sufficient}, we know that:
$$    c-\tm \geq \objw.\kld = \kir.$$

As $p_i$ is correct and ur-broadcasts only one \imp $\initm(-,\sn)$, then no correct process \kl-casts any different $(\witnessm(-),\sn,i)$, \klLDprop applies and at least one correct processes eventually \kl-delivers $(\witnessm(m),(\sn,i))$ using \objw and thus mbrb-delivers $(m,\sn,i)$ at line~\ref{IRMBRB-mbrb-dlv}.
\end{proof}

\begin{lemma}[\mbrGDprop] \label{lemma:ir-mbrb-global-delivery}
If a correct process $p_i$ mbrb-delivers an \app $m$ from a process $p_j$ with sequence number \sn, then at least \linebreak
$\lmbr = \lir $ correct processes mbrb-deliver $m$ from $p_j$ with sequence number~\sn.
\end{lemma}

\begin{proof}
\sloppy If $p_i$ mbrb-delivers $(m,\sn,j)$ at line~\ref{IRMBRB-mbrb-dlv}, then it has \kl-delivered $(\witnessm(m),(\sn,j))$ using \objw. As $\objw.\nodpty = \ttrue$, we can assert from \klWGDprop and \klNDYprop that at least $\objw.\lgd = \lval  $ correct processes eventually \kl-deliver $(\witnessm(m),(\sn,j))$ using \objw and thus mbrb-deliver $(m,\sn,j)$ at line~\ref{IRMBRB-mbrb-dlv}. By substituting the values of \qf and \qd, we obtain
$\objw.\lgd  = \lir = \lmbr  $  thus proving the lemma. 
\end{proof}


\section{Proof of the Signature-Based \texorpdfstring{\kl}{k2l}-Cast Implementation (Algorithm~\ref{alg:sb-klcast})}
\label{sec:sb-klcast-proofs}
For the proofs provided in this section, let us remind that, given two sets $A$ and $B$, we have $|A \cap B| = |A|+|B|-|A \cup B|$.
Moreover, the number of correct processes $c$ is superior or equal to $n-\tb$. 
Additionally, if $A$ and $B$ are both sets containing a majority of correct processes, we have $|A \cup B| \leq c$, which implies that $|A \cap B| \geq |A|+|B|-c$.
Furthermore, let us remind the assumptions of Alg.~\ref{alg:sb-klcast}:
\begin{itemize}
    \item \sbklassum~\ref{assum:no-partition}: $c > 2\tm$,
    
    \item \sbklassum~\ref{assum:dlv-tshld}: 
    $c-\tm \geq \qd \geq \tb+1$.
\end{itemize}

\subsection{Safety Proof}
\label{sec:sb-klcast-safety}

\begin{lemma} \label{lemma:n-sign-if-kldv}
If a correct process $p_i$ \kl-delivers $(m,\id)$, then at least $\qd-n+c$ correct processes have signed $(m,\id)$ at line~{\em \ref{SBKL-klc-sign}}.
\end{lemma}

\begin{proof}
If $p_i$ \kl-delivers $(m,\id)$ at line~\ref{SBKL-dlv}, then it sent \qd valid signatures for $(m,\id)$ (because of the predicate at line~\ref{SBKL-cond-dlv}).
The effective number of Byzantine processes in the system is $n-c$, such that $0 \leq n-c \leq \tb$.
Therefore, $p_i$ must have sent at least $\qd-n+c$ (which, due to \sbklassum~\ref{assum:dlv-tshld}, is strictly positive because $\qd > \tb \geq n-c$) valid distinct signatures for $(m,\id)$ that correct processes made at line~\ref{SBKL-klc-sign}, during a $\klcast(m,\id)$ invocation.
\end{proof}

\begin{lemma}[\klVprop] \label{lemma:sb-kl-validity}
If a correct process $p_i$ \kl-delivers an \app $m$ with identity \id, then at least $\kv = \qd-n+c$ correct processes \kl-cast $m$ with identity \id.
\end{lemma}

\begin{proof}
The condition at line~\ref{SBKL-klc-cond} implies that the correct processes that \kl-cast $(m,\id)$ constitute a superset of those that signed $(m,\id)$ at line~\ref{SBKL-klc-sign}. Thus, by Lemma~\ref{lemma:n-sign-if-kldv}, their number is at least $\kv=\qd-n+c$.
\end{proof}

\begin{lemma}[\klNDNprop] \label{lemma:sb-kl-no-duplication}
A correct process \kl-delivers at most one \app $m$ with identity~\id.
\end{lemma}

\begin{proof}
This property derives trivially from the predicate at line~\ref{SBKL-cond-dlv}.
\end{proof}

\begin{lemma}[\klNDYprop] \label{lemma:sb-kl-conditional-no-duplicity}
If the Boolean 
$\nodpty = \qd > \frac{n+\tb}{2}$ is \ttrue, then no two different correct processes \kl-deliver different \apps with the same identity~\id.
\end{lemma}

\begin{proof}
Let $p_i$ and $p_j$ be two correct processes that respectively \kl-deliver $(m,\id)$ and $(m',\id)$. We want to prove that, if the predicate $(\qd > \frac{n+\tb}{2})$ is satisfied, then $m=m'$.
    
Thanks to the predicate at line~\ref{SBKL-cond-dlv}, we can assert that $p_i$ and $p_j$ must have respectively sent at least \qd valid signatures for $(m,\id)$ and $(m',\id)$, made by two sets of processes, that we respectively denote $A$ and $B$, such that $|A| \geq \qd > \frac{n+\tb}{2}$ and $|B| \geq \qd > \frac{n+\tb}{2}$.
We have $|A \cap B| > 2\frac{n+\tb}{2}-n = \tb$. Hence, at least one correct process $p_x$ has signed both $(m,\id)$ and $(m',\id)$.
But because of the predicates at lines~\ref{SBKL-klc-cond}, $p_x$ signed at most one couple $(-,\id)$ during a $\klcast(m,\id)$ invocation at line~\ref{SBKL-klc-sign}.
We conclude that $m$ is necessarily equal to $m'$.
\end{proof}

\subsection{Liveness Proof}
\label{sec:sb-klcast-liveness}

\begin{lemma} \label{lemma:rcv-all-sigs}
All signatures made by correct processes at line~{\em\ref{SBKL-klc-sign}} are eventually received by at least $c-\tm$ correct processes at line~{\em\ref{SBKL-rcv}}.
\end{lemma}

\begin{proof}
Let $\{s_1,s_2,...\}$ be the set of all signatures for $(m,\id)$ made by correct processes at line~\ref{SBKL-klc-sign}. We first show by induction that, for all $z$, at least $c-\tm$ correct processes receive all signatures $\{s_1,s_2,...,s_z\}$ at line~\ref{SBKL-rcv}.

Base case $z = 0$. As no correct process signed $(m,\id)$, the proposition is trivially satisfied.

Induction. We suppose that the proposition is verified at $z$: signatures $s_1,s_2,...,s_z$ are received by a set of at least $c-\tm$ correct processes that we denote $A$. We now show that the proposition is verified at $z+1$: at least $c-\tm$ correct processes eventually receive all signatures $s_1,s_2,...,s_{z+1}$.

The correct process that makes the signature $s_{z+1}$ ur-broadcasts a $\bundlem(m,\id,\sigsv)$ \imp (at line~\ref{SBKL-klc-bcast}) where \sigsv contains $s_{z+1}$.
\sloppy
From the definition of the MA, $\bundlem(m,\id,\sigsv)$ is eventually received by a set of at least $c-\tm$ correct processes that we denote $B$. We have $|A \cap B| = 2(c-\tm)-c = c-2\tm > 2\tm-2\tm = 0$ (from \sbklassum~\ref{assum:no-partition}).
Hence, at least one correct process $p_j$ eventually receives all signatures $s_1,s_2,...,s_{z+1}$, and thereafter ur-broadcasts $\bundlem(m,\id,\sigsv')$ where $\{s_1,s_2,...,s_{z+1}\} \subseteq \sigsv'$.
Again, from the definition of the MA, $\bundlem(m,\id,\sigsv')$ is eventually received by a set of at least $c-\tm$ correct processes at line~\ref{SBKL-rcv}.
\end{proof}

\begin{lemma} \label{lemma:no-dlv-if-no-klc}
If no correct process \kl-casts $(m,\id)$ at line~\ref{SBKL-klc}, then no correct process \kl-delivers $(m,\id)$ at line~\ref{SBKL-dlv}.
\end{lemma}

\begin{proof}
Looking for a contradiction, let us suppose that a correct process $p_i$ \kl-delivers $(m,\id)$ while no correct process \kl-cast $(m,\id)$. Because of the condition at line~\ref{SBKL-cond-dlv}, $p_i$ must have ur-broadcast at least \qd valid signatures for $(m,\id)$, out of which at most \tb are made by Byzantine processes. As $\qd > \tb$ (\sbklassum~\ref{assum:dlv-tshld}), we know that $\qd-\tb > 0$. Hence, at least one correct process must have \kl-cast $(m,\id)$. Contradiction.
\end{proof}

\begin{lemma}[\klLDprop] \label{lemma:sb-kl-local-delivery}
If at least $\kld = \qd$ correct processes \kl-cast an \app $m$ with identity \id and no correct process \kl-casts an \app $m' \neq m$ with identity~\id, then at least one correct process $p_i$ \kl-delivers the \app $m$ with identity~\id.
\end{lemma}

\begin{proof}
As no correct process \kl-casts an \app $m' \neq m$ with identity~\id, then Lemma~\ref{lemma:no-dlv-if-no-klc} holds, and no correct process can \kl-deliver $(m',\id)$ where $m' \neq m$. Moreover, no correct process can sign $(m',\id)$ where $m' \neq m$ at line~\ref{SBKL-klc-sign}, and thus all $\kld = \qd$ correct processes that invoke $\klcast(m,\id)$ at line~\ref{SBKL-klc} also pass the condition at line~\ref{SBKL-klc-cond}, and then sign $(m,\id)$ at line~\ref{SBKL-klc-sign}. From Lemma~\ref{lemma:rcv-all-sigs}, we can assert that all \qd signatures are received at line~\ref{SBKL-rcv} by a set of at least $c-\tm$ correct processes, that we denote $A$. Let us consider $p_j$, one of the processes of $A$. There are two cases:
\begin{itemize}
    \item If $p_j$ passes the condition at line~\ref{SBKL-rcv-cond}, then it sends all \qd signatures at line~\ref{SBKL-rcv-bcast}, then invokes $\checkdelivery()$ at line~\ref{SBKL-rcv-chk-dlv}, passes the condition at line~\ref{SBKL-cond-dlv} (if it was not already done before) and \kl-delivers $(m,\id)$ at line~\ref{SBKL-dlv};
    
    \item If $p_j$ does not pass the condition at line~\ref{SBKL-rcv-cond}, then it means that it has already sent all \qd signatures before, whether it be at line~\ref{SBKL-klc-bcast} or~\ref{SBKL-rcv-bcast}, but after that, it necessarily invoked $\checkdelivery()$ (at line~\ref{SBKL-klc-chk-dlv} or~\ref{SBKL-rcv-chk-dlv}, respectively), passed the condition at line~\ref{SBKL-cond-dlv} (if it was not already done before) and \kl-delivered $(m,\id)$ at line~\ref{SBKL-dlv}. \qedhere
\end{itemize}
\end{proof}

\begin{lemma}[\klWGDprop] \label{lemma:sb-kl-weak-global-delivery}
If a correct process \kl-delivers an \app $m$ with identity \id, then at least $\lgd = c-\tm$ correct processes \kl-deliver an \app $m'$ with identity \id (each of them possibly different from $m$).
\end{lemma}

\begin{proof}
\sloppy
If $p_i$ \kl-delivers $(m,\id)$ at line \ref{SBKL-dlv}, then it has necessarily ur-broadcast the $\bundlem(m,\id,\sigsv)$ \imp containing the \qd valid signatures before, whether it be at line~\ref{SBKL-klc-bcast} or \ref{SBKL-rcv-bcast}.
From the definition of the MA, a set of at least $c-\tm$ correct processes, that we denote $A$, eventually receives this $\bundlem(m,\id,\sigsv)$ \imp at line~\ref{SBKL-rcv}. 
If some processes of $A$ do not pass the condition at line~\ref{SBKL-rcv-cond} upon receiving this $\bundlem(m,\id,\sigsv)$ \imp, it means that they already ur-broadcast all signatures of \sigsv.
Thus, in every scenario, all processes of $A$ eventually ur-broadcast all signatures of \sigsv at line~\ref{SBKL-klc-bcast} or~\ref{SBKL-rcv-bcast}.
After that, all processes of $A$ necessarily invoke the $\checkdelivery()$ operation at line~\ref{SBKL-klc-chk-dlv} or~\ref{SBKL-rcv-chk-dlv}, respectively, and then evaluate the condition at line~\ref{SBKL-cond-dlv}.
Hence, all correct processes of $A$, which are at least $c-\tm = \lgd$, \kl-deliver some \app for identity \id at line~\ref{SBKL-dlv}, whether it be $m$ or any other \app.
\end{proof}

\begin{lemma}[\klSGDprop] \label{lemma:sb-kl-strong-global-delivery}
If a correct process \kl-delivers an \app $m$ with identity \id, and no correct process \kl-casts an \app $m' \neq m$ with identity \id, then at least $\lgd = c-\tm$ correct processes \kl-deliver $m$ with identity \id.
\end{lemma}

\begin{proof}
If a correct process \kl-delivers $(m,\id)$ at line~\ref{SBKL-dlv}, then by Lemma~\ref{lemma:sb-kl-weak-global-delivery}, we can assert that at least $\lgd = c-\tm$ correct process eventually \kl-deliver some \app (not necessarily $m$) with identity \id.
Moreover, as no correct process \kl-casts $(m',\id)$ with $m' \neq m$, then Lemma~\ref{lemma:no-dlv-if-no-klc} holds, and we conclude that all \lgd correct processes \kl-deliver $(m,\id)$.
\end{proof}


\section{Numerical Evaluation}
\label{apx:numerical}
This section presents additional numerical results that complement those of Section~\ref{sec:numer-eval-mbrb}, and provides concrete lower-bound values for the \kld and \lgd parameters of the \kl-cast objects used in the reconstructed Bracha MBRB algorithm (Alg.~\ref{alg:b-mbrb}, page~\pageref{alg:b-mbrb}).
Results were obtained by considering a network with $n=100$ processes and varying
values of \tb and \tm.
Fig.~\ref{fig:BrachaE-heatmap} and Fig.~\ref{fig:BrachaR-heatmap} present the values of \kld and \lgd for the \obje and \objr of Alg.~\ref{alg:b-mbrb}.

The numbers in each cell show the value of \kld (Figs.~\ref{fig:BrachaEk-heatmap} and~\ref{fig:BrachaRk-heatmap}), resp. \lgd (Figs.~\ref{fig:BrachaEl-heatmap} and~\ref{fig:BrachaRl-heatmap}) that is required, resp. guaranteed, by the corresponding \kl-cast object.
The two plots show the two roles of the two \kl-cast objects.
The first, \obje, needs to provide agreement among the possibly different messages sent by Byzantine processes (Fig.~\ref{fig:BrachaE-heatmap}).
As a result, it can operate in a more limited region of the parameter space.
\objr, on the other hand, would, in principle, be able to support larger values of \tm and \tb, but it needs to operate in conjunction with \obje (Fig.~\ref{fig:BrachaR-heatmap}).

Fig.~\ref{fig:IR-full-heatmap} on page~\pageref{fig:IR-full-heatmap} already displays the values of $\ell$ provided by \objw in the Imbs-Raynal algorithm.
Fig.~\ref{fig:IR-k-heatmap} complements it by showing the required values of $k$ for \objw.
The extra constraint introduced by chaining the two objects suggests that a single \kl-cast algorithm could achieve better performance.
But this is not the case if we examine the performance of the reconstructed Imbs-Raynal algorithm depicted in Fig.~\ref{fig:IR-full-heatmap}. The reason lies in the need for higher quorum values in \objw due to $\single=\ffalse$. 
In the future, we plan to investigate if variants of this algorithm can achieve tighter bounds and explore the limits of signature-free \kl-cast-based broadcast in the presence of an MA and Byzantine processes.


\newcommand{\widthheatmap}{0.35}

\begin{figure}[ht]
\centering{
\begin{subfigure}[t]{0.45\textwidth}
\includegraphics[scale=\widthheatmap]{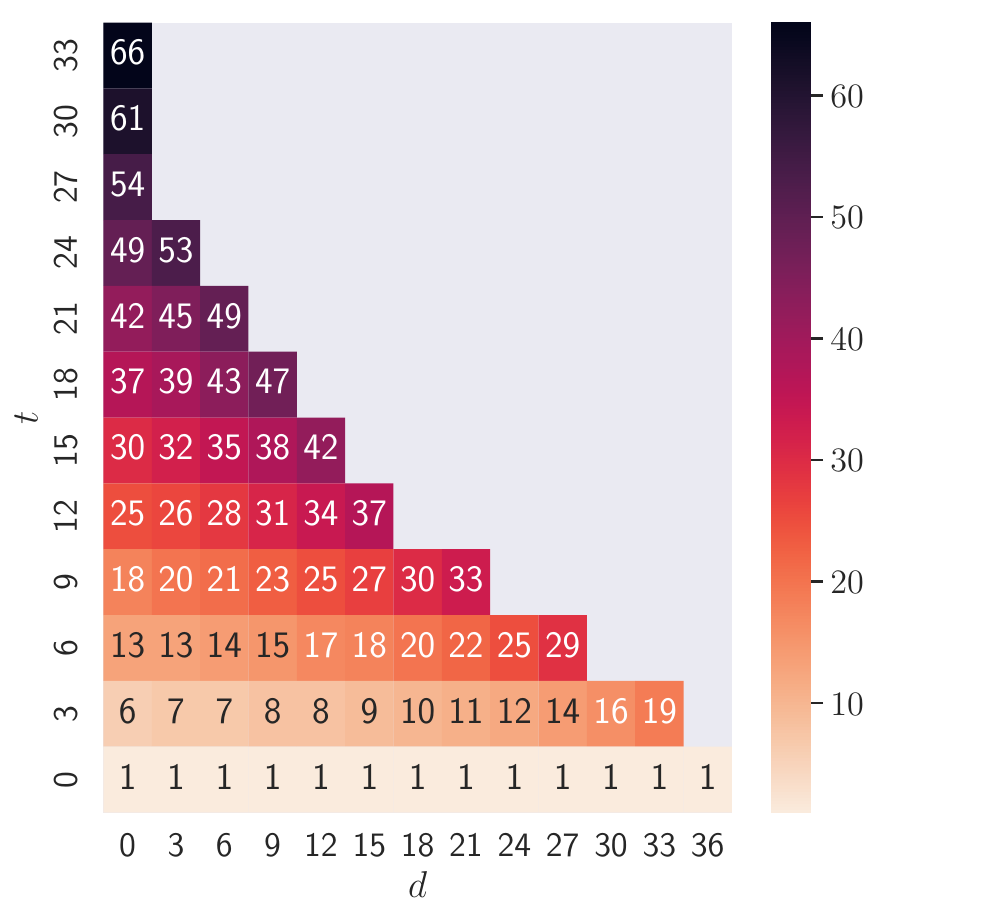}
\caption{Minimum required \kld}
\label{fig:BrachaEk-heatmap} 
\hspace{10ex}
\end{subfigure}
\begin{subfigure}[t]{0.45\textwidth}
\includegraphics[scale=\widthheatmap]{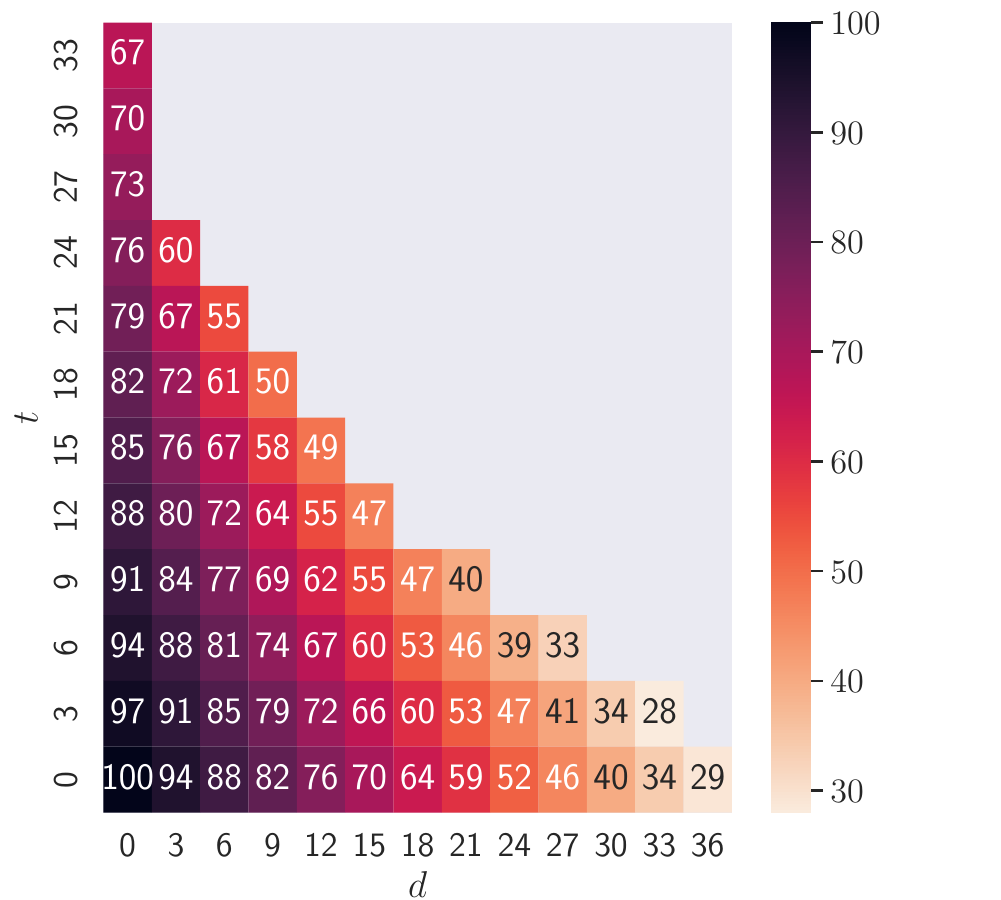}
\caption{Minimum provided \lgd}
\label{fig:BrachaEl-heatmap} 
\end{subfigure}
}
\caption{Required values of \kld and provided values of \lgd for \obje in the reconstructed Bracha BRB algorithm with varying values of \tb and \tm within the ranges that satisfy \bassum}
\label{fig:BrachaE-heatmap}
\end{figure}

\begin{figure}[ht]
\centering{
\begin{subfigure}[t]{0.45\textwidth}
\includegraphics[scale=\widthheatmap]{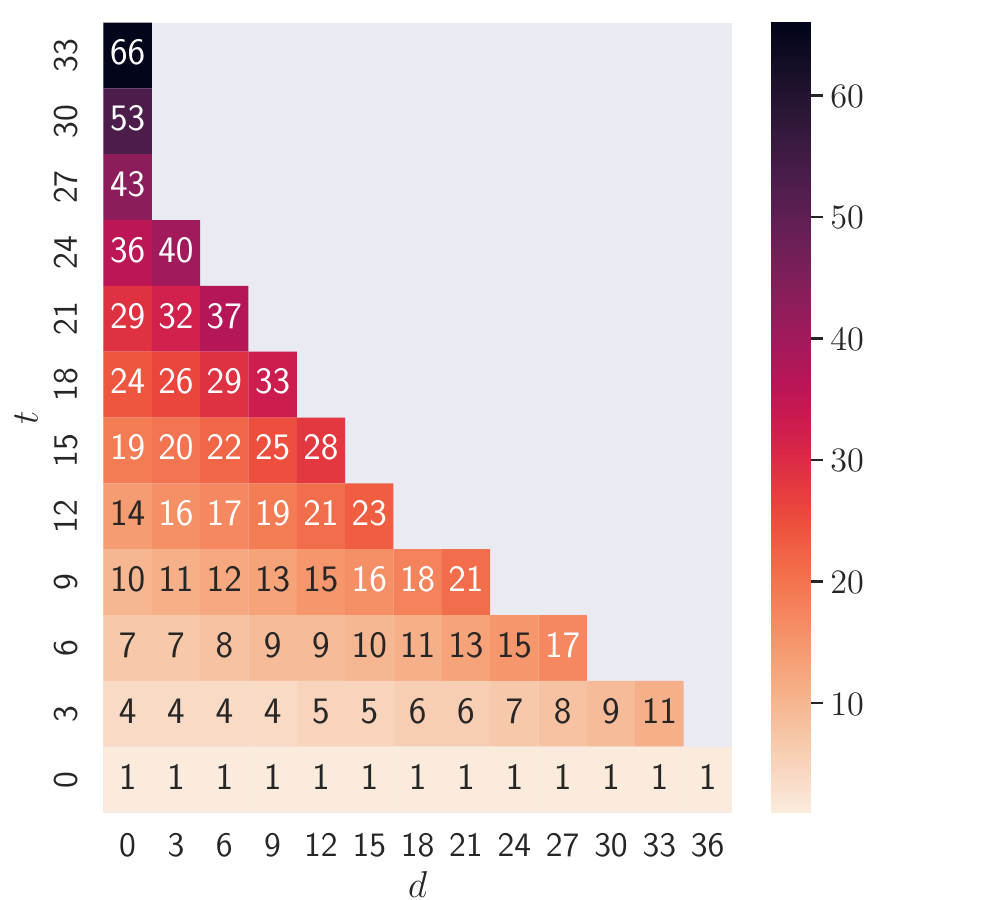}
\caption{Minimum required \kld}
\label{fig:BrachaRk-heatmap} 
\hspace{10ex}
\end{subfigure}
\begin{subfigure}[t]{0.45\textwidth}
\includegraphics[scale=\widthheatmap]{plots/klcast-kmin-OBJ_R-BRACHA-l-heatmap.pdf}
\caption{Minimum provided \lgd}
\label{fig:BrachaRl-heatmap} 
\end{subfigure}
}
\caption{Required values of \kld and provided values of \lgd for \objr in the reconstructed Bracha BRB algorithm with varying values of \tb and \tm within the ranges that satisfy \bassum}
\label{fig:BrachaR-heatmap}
\end{figure}

\begin{figure}[ht]
\centering{
\includegraphics[scale=\widthheatmap]{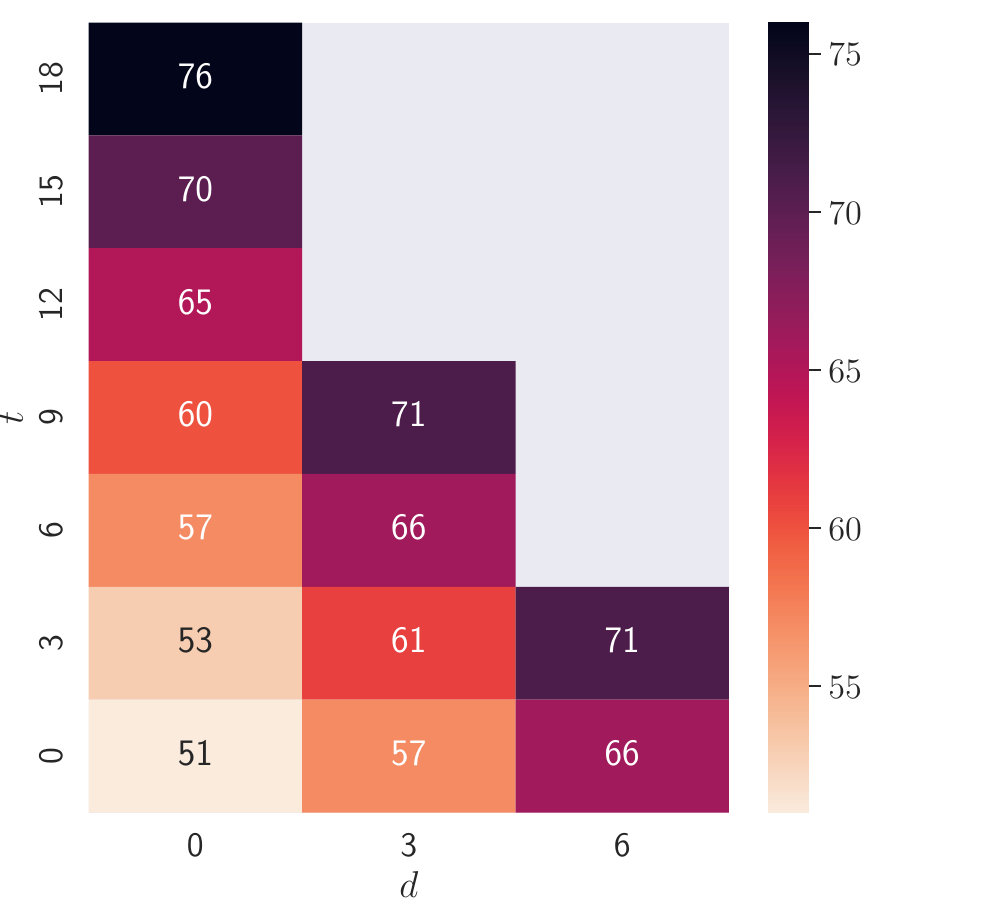}
}
\caption{Required values of \kld for \objw in the reconstructed Imbs \& Raynal BRB algorithm with varying values of \tb and \tm within the ranges that satisfy \irassum}
\label{fig:IR-k-heatmap}
\end{figure}


\end{document}